\RequirePackage{luatex85}% FINAL remove - is necessary for e.g. using lualatex with xy
\ifcsname directlua\endcsname%
  \RequirePackage{pdftexcmds}%
  \makeatletter\let\pdfstrcmp\pdf@strcmp\makeatother%
\fi

\documentclass[UKenglish,a4paper,cleveref]{lipics-v2021}
\nolinenumbers % easier diffing, XXX

\usepackage[appendix=append,bibliography=common]{apxproof}

\renewcommand{\appendixsectionformat}[2]{Proofs and Additional
  Material for \S#1\ (#2)}

\usepackage{hyperref}
\AtBeginDocument{\hypersetup{breaklinks=true,pdfborder={1 1 1},linkbordercolor=[rgb]{1,0.8,0.8},citebordercolor=[rgb]{0.6,1,0.6}}}% final remove

\usepackage{amssymb,amsmath,latexsym,stmaryrd}
\usepackage{mathtools}
\usepackage{microtype}
\usepackage{csquotes}
\usepackage{enumerate}
\usepackage{wrapfig}
\usepackage{ifthen}
\usepackage{xifthen}
\usepackage{graphicx,epsfig,color}
\usepackage[arrow,matrix,frame,curve,cmtip]{xy}\CompileMatrices
\usepackage{tikz}
\usepackage{adjustbox}
\usepackage{xfrac}

\makeatletter
\def\Gin@getbase#1{%
  \edef\Gin@base{\filename@area\filename@base}%
  \edef\Gin@ext{#1}%
}
\makeatother

\usepackage{amsthm}
\usepackage{cleveref}
\usepackage{soul}

\newcommand{\renewtheorem}[1]{%
  \expandafter\let\csname #1\endcsname\relax%
  \expandafter\let\csname c@#1\endcsname\relax%
  \expandafter\let\csname end#1\endcsname\relax%
  \newtheorem{#1}%
}
\theoremstyle{plain}
  \newtheoremrep{theorem}{Theorem}
	\newtheoremrep{lemma}[theorem]{Lemma}
	\newtheoremrep{corollary}[theorem]{Corollary}
	\newtheoremrep{proposition}[theorem]{Proposition}
	\newtheoremrep{exercise}[theorem]{Exercise}
	\newtheoremrep{definition}[theorem]{Definition}
	\theoremstyle{definition}
	\newtheoremrep{example}[theorem]{Example}
	\theoremstyle{remark}
	\newtheoremrep{note}[theorem]{Note}
	\newtheoremrep{remark}[theorem]{Remark}
	\theoremstyle{claimstyle}
	\newtheoremrep{claim}[theorem]{Claim}
\theoremstyle{plain}

\gdef\ExaEndSeen{0}
\newcommand\ExaEndHere{%
  \leavevmode\unskip\penalty9999 \hbox{}\nobreak\hfill\quad\hbox{$\lrcorner$}%
  \gdef\ExaEndSeen{1}%
}
\gdef\RemEndSeen{0}
\newcommand\RemEndHere{%
  \leavevmode\unskip\penalty9999 \hbox{}\nobreak\hfill\quad\hbox{$\lrcorner$}%
  \gdef\RemEndSeen{1}%
}
\AtBeginEnvironment{example}{\gdef\ExaEndSeen{0}}
\AtBeginEnvironment{remark}{\gdef\RemEndSeen{0}}
\AtEndEnvironment{example}{\if\ExaEndSeen0\ExaEndHere\fi}
\AtEndEnvironment{remark}{\if\RemEndSeen0\RemEndHere\fi}

\usepackage{silence}
\usetikzlibrary{matrix,shapes.geometric,backgrounds,fit}
\usetikzlibrary{positioning,shapes,arrows.meta,calc}

\tikzset{%
  gn/.style={circle,fill=black,inner sep=0pt, minimum size=6pt, prefix after command={\pgfextra{\tikzset{every label/.style={font=\footnotesize}}}}},
  gnw/.style={circle,draw,inner sep=0pt, minimum size=6pt, prefix after command={\pgfextra{\tikzset{every label/.style={font=\footnotesize}}}}},
  gni/.style={circle,draw,inner sep=1.5pt, minimum size=6pt, font=\footnotesize},
  roundbox/.style={rectangle, rounded corners=4pt},
  >=To,
  gedge/.style={->,>=latex},
  condtri/.style={draw,dart,dart tail angle=135,dart tip angle=50,inner xsep=0.1em, inner ysep=0.2em,anchor=tip},
  fancydotted/.style={dash pattern=on 1.25pt off 1.75pt},
  gninlinablefont/.style={font=\small},
  gninlinable/.style={gninlinablefont,circle,draw,inner sep=0.2pt, minimum size=8pt},
  gninlinablesub/.style={gninlinable,font=\scriptsize},
}
\definecolor{graphmorphismgreen}{rgb}{0.4,0.84,0.3}
\definecolor{graphmorphismfailred}{rgb}{0.9,0.2,0.4}
\tikzset{
  cospanintcommon/.style={draw=black!38,line width=0.5pt},
  cospanint/.style={->,cospanintcommon},
  cospanintmono/.style={cospanintcommon,>->},
  cospanarr/.style={->,line width=0.5pt},
  graphmor/.style={cospanarr},
  graphmorcol/.style={->,line width=0.5pt,color=graphmorphismgreen},
  graphmorcolfail/.style={->,dash pattern=on 1.25pt off 1.75pt,line width=0.5pt,color=graphmorphismfailred},
}

\pgfdeclarelayer{bg}    % declare background layer
\pgfdeclarelayer{background}    % declare background layer (???)
\pgfsetlayers{bg,background,main}  % set the order of the layers (main is the standard layer)

  \newcommand{\dotEx}{\ .\exists\ }
  \newcommand{\dotAll}{\ .\forall\ }
  \newcommand{\dotFalse}{\ . \condfalse}
  \newcommand{\dotTrue}{\ . \condtrue}
  \def\faintborder{%
    \begin{pgfonlayer}{bg}
      \draw[black!40, fill=black!5, rounded corners=2pt, line width=0.2pt, overlay]
      ($(current bounding box.south west)+(-1pt,-1pt)$) rectangle
      ($(current bounding box.north east)+(1pt,0.8pt)$);
    \end{pgfonlayer}
  }
  \def\fGraph#1#2{% #1 = node whose baseline should be used, #2 = code
    \tikz[baseline=(#1.base)]{#2 \faintborder}%
  }
  \def\fcGraph#1#2{% #1 = node whose baseline should be used, #2 = code
    \tikz[baseline=(#1.base),x=0.6cm]{#2 \faintborder}%
  }

\usepackage{mauflpic}

\theoremstyle{plain}
\newtheorem{algorithm}[theorem]{Algorithm}

\renewcommand{\phi}{\varphi}
\renewcommand{\epsilon}{\varepsilon}
\renewcommand{\emptyset}{\varnothing}

\tikzset{
  vertex/.style={shape=circle,draw=black,minimum height=.5cm,inner sep=0pt},
  mono/.style={>->},
  epi/.style={->>},
  arlabel/.style={circle,inner sep=1pt,outer sep=0pt,font=\scriptsize},
  ontop/.style={preaction={draw,-,line width=3pt,white}}
}

\newcommand\comp{\mathop{;}}

\newcommand\arright[1][]{%
 \ifthenelse{\equal{#1}{}}%
 {\rightarrow}%
 {\mathbin{%
     \mathchoice%
     {\xrightarrow{#1}}%
     {\scalebox{.8}[1]{$\textstyle\relbar$}{\raisebox{.23ex}{$\scriptstyle #1$}}{\shortrightarrow}}%
     {\scalebox{.8}[1]{$\scriptstyle\relbar$}{\raisebox{.15ex}{$\scriptscriptstyle #1$}}{\shortrightarrow}}%
     {\scalebox{.8}[1]{$\scriptscriptstyle\relbar$}{\raisebox{.15ex}{$\scriptscriptstyle #1$}}{\shortrightarrow}}%
 }}
}

\newcommand\arleft[1][]{%
 \ifthenelse{\equal{#1}{}}%
 {\leftarrow}%
 {\mathchoice%
   {\xleftarrow{#1}}
   {\mathbin{{\textstyle\shortleftarrow}{\raisebox{.23ex}{$\scriptstyle #1$}}\scalebox{.8}[1]{$\textstyle\relbar$}}}
   {\mathbin{{\scriptstyle\shortleftarrow}{\raisebox{.15ex}{$\scriptscriptstyle #1$}}\scalebox{.8}[1]{$\scriptstyle\relbar$}}}
   {\mathbin{{\scriptscriptstyle\shortleftarrow}{\raisebox{.15ex}{$\scriptscriptstyle #1$}}\scalebox{.8}[1]{$\scriptscriptstyle\relbar$}}}
 }
}

\newcommand\step[1][]{
  \ifthenelse{\equal{#1}{}}{
    \Rightarrow
  }{
    \Rightarrow_{#1}
  }
}

\newcommand\red[1][]{
  \ifthenelse{\equal{#1}{}}{
    \Rightarrow^*
  }{
    \Rightarrow_{#1}^*
  }
}

\newcommand{\cbox}[1]{\vspace{0.2cm}\noindent %
  \fbox{\parbox{.97\textwidth}{#1}}\vspace{0.2cm}}
\renewcommand{\cbox}[1]{}

\newcommand{\short}[1]{}
\newcommand{\full}[1]{#1}

\renewcommand{\short}[1]{#1}
\renewcommand{\full}[1]{}

\DeclareMathOperator{\RO}{RO}

\newcommand{\id}{\mkern1mu\mathrm{id}}
\newcommand{\condtrue}{\mkern1mu\mathrm{true}}
\newcommand{\condfalse}{\mkern1mu\mathrm{false}}

\newcommand{\graphf}{\textbf{Graph}_{\textbf{fin}}}
\newcommand{\graphfinj}{\textbf{Graph}_{\textbf{fin}}^{\textbf{inj}}}
\def\ILC{\mathbf{ILC}}
\def\Cospan{\mathbf{Cospan}}

\def\defeq{:=}

\newcommand{\notmodels}{\mkern3mu\not\mkern-3mu\models}

\newcommand{\catC}{\ensuremath{\mathbf{C}}}
\newcommand{\quantor}{\ensuremath{\mathcal{Q}}}
\newcommand{\dom}{\ensuremath{\mathsf{dom}}}
\newcommand{\cod}{\ensuremath{\mathsf{cod}}}
\newcommand{\Cond}{\ensuremath{\mathsf{Cond}}}

\newcommand{\Seq}{\mathsf{Seq}} % composable sequences

\newcommand{\uptoShift}{u_{\downarrow}}
\newcommand{\uptoConj}{u_{\land}}
\newcommand{\uptoConjCond}{\mathcal U_{\land}}
\newcommand{\uptoRecomp}{u_{\fatsemi}}
\newcommand{\uptoIso}{u_{\cong}}

\newcommand{\condiso}{\cong}

\newcommand{\natzero}{\mathbb N_0}

\def\biglor{\bigvee}
\def\bigland{\bigwedge}

\newcommand{\tdots}{. \mkern1mu . \mkern1mu . \mkern1mu}

\makeatletter
\newcommand{\biggg}{\bBigg@{3}}
\newcommand{\Biggg}{\bBigg@{4}}
\newcommand{\bigggg}{\bBigg@{5}}
\newcommand{\Bigggg}{\bBigg@{6}}
\makeatother

\usepackage{xspace}
\def\coincide{\discretionary{co}{incide}{co\kern.75pt incide}\xspace}
\def\coinductive{\discretionary{co}{inductive}{co\kern.75pt inductive}\xspace}
\def\Coinductive{\discretionary{Co}{inductive}{Co\kern.75pt inductive}\xspace}
\def\coinductively{\discretionary{co}{inductively}{co\kern.75pt inductively}\xspace}

\newenvironment{proofparts}{%
  \quad\par% don't put the list bullet in the middle of the line
  \begin{itemize}%
}{%
  \end{itemize}%
}
\newcommand{\proofPartNoNewline}[1]{\item \textbf{(#1):}}
\newcommand{\proofPart}[1]{\proofPartNoNewline{#1}\\}

\title{Coinductive Techniques for Checking Satisfiability of
  Generalized Nested Conditions}

\author{Lara Stoltenow}{University of Duisburg-Essen,
  Germany}{lara.stoltenow@uni-due.de}{https://orcid.org/0009-0009-1667-8573}{}
\author{Barbara K\"onig}{University of Duisburg-Essen,
  Germany}{barbara\_koenig@uni-due.de}{https://orcid.org/0000-0002-4193-2889}{}
\author{Sven Schneider}{Hasso Plattner Institute at the University of
  Potsdam, Germany}{sven.schneider@hpi.de}{http://orcid.org/0000-0001-9828-618X}{}
\author{Andrea Corradini}{Universit\`a di Pisa, Italy}{andrea@di.unipi.it}{https://orcid.org/0000-0001-6123-4175}{Research partially supported by the Italian MUR under the PRIN 20228KXFN2
``STENDHAL'' and by the University of Pisa under the PRA 2022\_99 ``FM4HD''.}
\author{Leen Lambers}{Brandenburgische Technische Universit\"at
  Cottbus-Senftenberg, Germany}{leen.lambers@b-tu.de}{https://orcid.org/0000-0001-6937-5167}{}
\author{Fernando Orejas}{Universitat Polit{\`e}cnica de Catalunya,
  Spain}{orejas@lsi.upc.edu}{http://orcid.org/0000-0002-3023-4006}{Research
  partially supported by MCIN/ AEI /10.13039/501100011033 under grant PID2020-112581GB-C21.}

\authorrunning{Stoltenow, K\"onig, Schneider, Corradini, Lambers, Orejas}

\Copyright{L. Stoltenow, B. K\"onig, S. Schneider, A. Corradini,
  L. Lambers and F. Orejas}


\begin{document}

\maketitle

\begin{abstract}
  We study nested conditions, a generalization of first-order logic to
  a categorical setting, and provide a tableau-based (semi-decision)
  procedure for checking (un)satisfiability and finite model
  generation. This generalizes earlier results on graph
  conditions. Furthermore we introduce a notion of witnesses, allowing
  the detection of infinite models in some cases.
  To ensure completeness, paths in a tableau must be fair, where
  fairness requires that all parts of a condition are processed
  eventually. Since the correctness arguments are non-trivial, we rely
  on coinductive proof methods and up-to techniques that structure the
  arguments.
  We distinguish between two types of categories: categories where all
  sections are isomorphisms, allowing for a simpler tableau calculus
  that includes finite model generation; in categories where this
  requirement does not hold, model generation does not work, but we
  still obtain a sound and complete calculus.
  %  \`a la Pennemann (essentially first-order
%   logic for graphs, e.g.\ testing for occurrence of certain subgraphs
%   etc.) and develop an algorithm to check unsatisfiability, or find
%   and output finite models if they exist, or (in some cases) also
%   detect certain types of infinite models.
\end{abstract}

% ---------------------------------------------------------------------------- %
\section{Introduction}
\label{sec:introduction}

Nested graph conditions (called graph conditions subsequently) are a
well-known specification technique for graph transformation systems
\cite{hp:correctness-nested-conditions} where they are used, e.g., to
specify graph languages and application conditions.  While their
definition is quite different from first-order logic (FOL), they have
been shown to be equivalent to FOL in
\cite{r:representing-fol,hp:correctness-nested-conditions}.  They are
naturally equipped with operations such as shift, a form of partial
evaluation, which is difficult to specify directly in FOL.  This
operation can be used to compute weakest preconditions and strongest
postconditions for graph transformation systems
\cite{bchk:conditional-reactive-systems}.

In \cite{bchk:conditional-reactive-systems} it
has also been observed that graph conditions can be generalized to the
categorical setting of reactive systems \cite{lm:derive-bisimulation}
as an alternative to the previously considered instantiation to graphs
and injective graph morphisms that is equivalent to FOL.  Further
possible instantiations include cospan categories where the
graphs, equipped with an inner and an outer interface, are the arrows,
as well as Lawvere theories.  To derive analysis techniques for all
such instantiations, we consider nested conditions in the general
categorical setting.

Here we are in particular interested in satisfiability checks on the
general categorical level. As in FOL, satisfiability can be an
undecidable problem (depending on the category), and we propose a
semi-decision procedure that can simultaneously serve as a model
finder. For FOL there are several well-known methods for
satisfiability checking, for instance resolution or tableau
proofs~\cite{f:fol-theorem-proving}, while model generation is
typically performed separately. The realization that satisfiability
checking is feasible directly on graph conditions came in
\cite{p:algorithm-approximating-satisfiability,p:development-correct-gts},
and a set of tableau rules was presented
\cite{p:algorithm-approximating-satisfiability} without a proof of
(refutational) completeness that was later published in
\cite{lo:tableau-graph-properties}, together with a model generation
procedure \cite{slo:model-generation}. A generalization to
non-injective graph morphisms was given
in~\cite{nopl:navigational-logic}.

The contributions of the current paper can be summarized as follows:
\begin{itemize}
\item We generalize the tableau-based semi-decision procedure for
  graph conditions from \cite{lo:tableau-graph-properties} to the
  level of general categories, under some mild constraints (such as
  the existence of so-called representative squares
  \cite{bchk:conditional-reactive-systems}). We present a procedure
  that has some resemblance to the construction of a tableau in FOL.
  
  \item We distinguish between two cases: one
  simpler case in which all sections (arrows that have a right
  inverse) in the category under consideration are isomorphisms (\Cref{sec:satisfiability}); and a more
  involved case where this does not necessarily hold
  (\Cref{sec:general-case}). The tableau rules of the former case (\Cref{sec:satisfiability}) are easier to present and implement, and we can give
  additional guarantees, such as model generation whenever there
  exists a so-called finitely decomposable model, generalizing the
  notion of finite models. The latter case (\Cref{sec:general-case}) does not guarantee model generation and has more involved tableau rules, but it allows for instantiations to more categories, such as graphs and arbitrary morphisms. The results of both cases generalize
  \cite{lo:tableau-graph-properties,nopl:navigational-logic,slo:model-generation}
  from graphs and graph morphisms to an abstract categorical~level,
  which allows application to additional categories such as cospan
  categories and Lawvere theories (see \Cref{appendix-terms}).
  
\item The completeness argument for the satisfiability checking
  procedure -- in particular showing that non-termination implies the
  existence of an infinite model -- requires that the tableau
  construction satisfies a fairness constraint. The resulting proof is
  non-trivial and -- compared to the proof in
  \cite{lo:tableau-graph-properties} -- we show that it can be
  reformulated using up-to techniques. Here we give it a completely
  new and hopefully clarifying structure that relies on
  \coinductive methods
  \cite{p:complete-lattices-up-to,ps:enhancements-coinductive}.
  The alternative would be to inline the up-to techniques, or to rely on complex ad-hoc notation that are less clear and further complicate the proof.
  
  \item Furthermore we use \coinductive
  techniques to display witnesses for infinite models
  (\Cref{sec:witnesses}): in some cases where only infinite models
  exist and hence the tableau construction is non-terminating, we can
  still stop and determine that there does exist an infinite model.
  \Coinductive techniques
  \cite{s:bisimulation-coinduction,ps:enhancements-coinductive} are
  reasoning techniques based on greatest fixpoints, suitable to
  analyze infinite or cyclic structures.
  %Such techniques can be
  %enhanced with so-called up-to techniques
  %\cite{ps:enhancements-coinductive}.
  To the best of our knowledge, such techniques have not yet been
  employed in the context of satisfiability checking for FOL and graph
  conditions.
\end{itemize}

\noindent The main contribution compared to previous work consists of a
categorical generalization to reactive systems on the one hand, and
the use of \coinductive (up-to) techniques on the other hand.  The
implication of the first type of contribution is that the theory
becomes available for new instantiations such as adhesive categories
(which includes all variants of graphs, such as typed graphs, Petri
nets, but also algebraic specifications,
cf.~\cite{eept:fundamentals-agt}), as well as other cases such as
cospan categories and Lawvere theories. The second type of
contribution implies that the proofs (especially for completeness) can
now be presented in a more systematic way.

% \medskip\hrule\medskip

% Ideas:

% \begin{itemize}
% \item Alternative to satisfiability checks for first-order logic
%   (resolution, tableau)
% \item Direct implementation for use in graph transformation (where
%   such nested conditions are used)
% \item Generalization to other categories (cospans, terms, \dots)
% \item unsatisfiability \& model finding in the same procedure
% \item \coinductive \& up-to techniques (in proofs, maybe also for
%   witnesses for infinite models)
% \end{itemize}

% Related work: $\leadsto$ appendix

% \begin{itemize}
% \item \cite{p:development-correct-gts}: two algorithms:
%   SeekSat finds a model for every satisfiable condition, but might not terminate for unsat.
%   ProCon finds a proof for every unsatisfiable condition, using very similar techniques as we do ((supporting) lift = pulling forward).
%   However, the two procedures are not integrated and it is recommended to run them in parallel, which duplicates computational effort.
% \item \cite{lo:tableau-graph-properties}: limited to graphs, requires conversion to CNF, repeated application of extension rule produces a tableau that describes a DNF.
%   Already shows the criterion for recognizing satisfiable conditions, but does not use it in the proofs.
% \item \cite{slo:model-generation}: ...
% \end{itemize}

% ---------------------------------------------------------------------------- %

This paper is the full version of
\cite{sksclo:coinductive-satisfiability-nested-conditions}.

\section{Preliminaries}
\label{sec:preliminaries}

\subsection{Coinductive Techniques}
\label{sec:coinductive}

A \emph{complete lattice} is a partially ordered set
$(L, \sqsubseteq)$ where each subset $Y\subseteq L$ has a greatest
lower bound, denoted by $\bigsqcap Y$ and a least upper bound, denoted
by $\bigsqcup Y$.
% In this paper, the type of lattices that we consider are
% powersets of relations ordered by inclusion,
% i.e.\ the elements of the lattice are relations
% and thus the functions we consider map relations to relations.

A function $f \colon L \to L$ is \emph{monotone} if for all
$l_1, l_2 \in L$,\; $l_1 \sqsubseteq l_2$ implies
$f(l_1) \sqsubseteq f(l_2)$, \emph{idempotent} if $f \circ f = f$, and
\emph{extensive} if $l \sqsubseteq f(l)$ for all $l \in L$.  % When $f$
%is monotone, extensive and idempotent it is called an \emph{(upper)
%  closure}.  In this case, $f(L) = \{ f(l) \mid l \in L \}$ is a
%complete lattice.

Given a monotone function $f \colon L \to L$ we are in particular
interested in its \emph{greatest fixpoint} $\nu f$.  By Tarski's
Theorem \cite{t:lattice-fixed-point},
$\nu f = \bigsqcup \{ x \mid x \sqsubseteq f(x) \}$, i.e., the
greatest fixpoint is the least upper bound of all
post-fixpoints. Hence for showing that $l\sqsubseteq \nu f$ (for some
$l\in L$), it is sufficient to prove that $l$ is below some
post-fixpoint $l'$, i.e., $l \sqsubseteq l'\sqsubseteq f(l')$.

In order to employ up-to techniques one defines a monotone function
$u \colon L \to L$ (the up-to function) and checks whether $u$ is
\emph{$f$-compatible}, that is $u \circ f \sqsubseteq f \circ u$.  If
that holds every post-fixpoint $l$ of $f\circ u$ (that is
$l \sqsubseteq f(u(l))$) is below the greatest fixpoint of $f$
($l\sqsubseteq \nu f$). This simple technique can often greatly
simplify checking whether a given element is below the greatest
fixpoint. For more details see~\cite{p:complete-lattices-up-to}.

% if $l \sqsubseteq \nu(f \circ u)$ by showing that $l$ is a
% post-fixpoint of $f\circ u$.  Typically, the characteristics of $u$
% should make it easier to prove $l \sqsubseteq f(u(l))$ than proving
% $l \sqsubseteq f(l)$.  This is clearly the case when $u$ is extensive,
% since extensiveness of $u$ and monotonicity of $f$ implies
% $f(l) \sqsubseteq f(u(l))$ and thus obtaining $l \sqsubseteq f(u(l))$
% is easier than obtaining $l \sqsubseteq f(l)$.  A \emph{sound up-to
%   function} for $f$ is any monotone function $u$ such that
% $\nu (f \circ u) \sqsubseteq \nu f$ and hence $l\sqsubseteq f(u(l))$
% implies $l\sqsubseteq \nu f$.  A monotone function $u$ is
% \emph{$f$-compatible} if $u \circ f \sqsubseteq f \circ u$. It is
% known that every $f$-compatible up-to function is sound for $f$ and
% that $u(\nu f) = \nu f$.\oldtodo{\textbf{B:} This might help to show that
%   satisfaction is preserved by recomposition and show that the shift
%   property for satisfaction holds.}  For more details
% see~\cite{p:complete-lattices-up-to}.

\subsection{Categories}

% \oldtodo{We could motivate the choice of category (in particular when
%   choosing cospan categories) by their applicability to reactive
%   systems with cospans (cf.\ \cite{lmcs:8951})}

We will use standard concepts from category theory. Given an arrow
$f\colon A\to B$, we write $\dom(f)=A$, $\cod(f)=B$. For two arrows
$f \colon A \to B$, $g \colon B \to C$ we denote their composition by
${f;g} \colon A \to C$.
% Usually written $g \circ f$, we chose $f;g$ to better match the
% reading order of the diagram
% $A \xrightarrow f B \xrightarrow g C = A \xrightarrow{f;g} C$.
%
An arrow $s \colon A \to B$ is a \emph{section} (also known as
\emph{split mono}) if there exists $r\colon B\to A$ such that
$s;r = \id$. That is, sections are those arrows $s$ that have a
right-inverse~$r$.  Arrows that have a left-inverse (in this case $r$) are
called \emph{retractions}.

As in graph rewriting we will consider the category $\graphf$, which
has finite graphs as objects and graph morphisms as arrows.  We also
consider $\graphfinj$, the subcategory of $\graphf$ that has the same
objects, but only injective, i.e.\ mono, graph morphisms. In this
category the sections are exactly the isos, while this is not the case
in $\graphf$.

\begin{toappendix}
  \paragraph*{Graphs and graph morphisms}

  We will define in more detail which graphs and graph morphisms we
  are using: in particular, a graph is a tuple $G = (V,E,s,t,\ell)$,
  where $V,E$ are sets of nodes respectively edges, $s,t\colon E\to V$
  are the source and target functions and $\ell\colon V\to \Lambda$
  (where $\Lambda$ is a set of labels) is the node labelling
  function. In the examples we will always omit node labels by
  assuming that there is only a single label.

  A graph $G$ is finite if both $V$ and $E$ are finite.

  Furthermore, given two graphs $G_i = (V_i,E_i,s_i,t_i,\ell_i)$,
  $i\in\{1,2\}$, a graph morphism $\phi\colon G_1\to G_2$ consists of
  two maps $\phi_V\colon V_1\to V_2$, $\phi_E\colon E_1\to E_2$ such
  that $\phi_V\circ s_1 = s_2\circ \phi_E$, $\phi_V\circ t_1 =
  t_2\circ \phi_E$ and $\ell_1 = \ell_2\circ\phi_V$.

  In the examples, the mapping of a morphism is given implicitly by
  the node identifiers: for instance, $
    \fcGraph{n1}{\node[gninlinable] (n1) at (0,0) {$1$}; \node[gninlinable] (n2) at (1,0) {$2$};}
    \rightarrow \fcGraph{n1}{
      \node[gninlinable] (n1) at (0,0) {$1$};
      \node[gninlinable] (n2) at (2,0) {$2$};
      \node[gninlinable] (n3) at (1,0) {$3$};
      \draw[gedge] (n1) to (n3);
      \draw[gedge] (n2) to (n3);
    }
  $ adds the node identified by $3$ and adds two edges from the existing
  nodes identified by $1$ and $2$.
\end{toappendix}

Another important example category that will be used in
\Cref{sec:witnesses} is based on cospans: note that reactive systems
instantiated with cospans
\cite{lmcs:8951,ss:reactive-cospans,s:deriving-congruences} yield
exactly double-pushout rewriting
\cite{eps:gragra-algebraic}.  Given a base category $\mathbf D$ with
pushouts, the category $\Cospan(\mathbf D)$ has as objects the objects
of $\mathbf D$ and as arrows \emph{cospans}, which are equivalence
classes of pairs of arrows of the form
$\smash{A \xrightarrow{f_L} X \xleftarrow{f_R} B}$, where the middle
object is considered up to isomorphism. Cospan composition is
performed via pushouts (for details see Appendix~\ref{app:cospans}).

\begin{toappendix}
  \paragraph*{Cospans and cospan composition}
  \label{app:cospans}
  Two cospans \mbox{$%
    f \colon A \xrightarrow{f_L} X \xleftarrow{f_R} B,\ %
    g \colon B \xrightarrow{g_L} Y \xleftarrow{g_R} C$} are composed
  by taking the pushout $(p_L, p_R)$ of $(f_R, g_L)$ as shown in
  \Cref{fig-cospan-compo}.  The result is the cospan
  \mbox{${f;g} \colon A \xrightarrow{f_L;p_L} Z \xleftarrow{g_R;p_R}
    C$}, where $Z$ is the pushout object of $f_R,\; g_L$.  We see an
  arrow $f \colon A \to C$ of $\Cospan(\mathbf D)$ as an object~$B$
  of~$\mathbf D$ equipped with two interfaces $A,C$ and corresponding
  arrows $f_L,f_R$ to relate the interfaces to $B$, and composition
  glues the inner objects of two cospans via their common
  interface.

  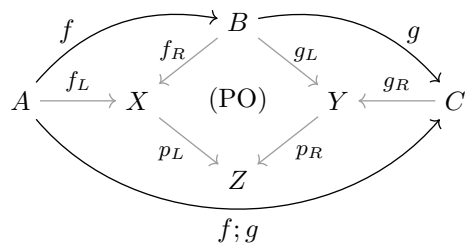
\begin{figure}[h]
    \centering\begin{tikzpicture}[x=1.10cm,y=1.10cm]% PAGEBREAK-ADJUST

      \def\sqwo{1.2} \def\sqoo{1.4} \def\sqho{0.95}

      \node (a) at (-\sqwo-\sqoo,0) {$A$}; \node (x) at (-\sqwo,0)
      {$X$}; \node (b) at (0,\sqho) {$B$}; \node (y) at (\sqwo,0)
      {$Y$}; \node (c) at (\sqwo+\sqoo,0) {$C$}; \node (z) at
      (0,-\sqho) {$Z$};

      \def\smf{\footnotesize} \draw[cospanint] (a) --
      node[above]{\smf$f_L$} (x); \draw[cospanint] (b) --
      node[above,pos=0.7]{\smf$f_R\ $} (x); \draw[cospanint] (b) --
      node[above,pos=0.7]{\smf\kern5pt$g_L$} (y); \draw[cospanint] (c)
      -- node[above]{\smf$g_R$} (y);

      \draw[cospanint] (x) -- node[below,pos=0.4]{\smf$p_L$\kern8pt}
      (z); \draw[cospanint] (y) --
      node[below,pos=0.4]{\smf\kern12pt$p_R$} (z);

      \draw[cospanarr] (a) to[bend left=30] node[above,pos=0.2]{$f$}
      (b); \draw[cospanarr] (b) to[bend left=30]
      node[above,pos=0.8]{$\ g$} (c); \draw[cospanarr] (a) to[bend
      right=50] node[below]{$f;g$} (c);

      \node at (0,0) {\textup{(PO)}};

    \end{tikzpicture}%
    \caption{Composition of cospans $f$ and~$g$ is done via pushouts}%
    \label{fig-cospan-compo}%
  \end{figure}
  In order to make sure that arrow composition in $\Cospan(\mathbf D)$
  is associative on the nose, we quotient cospans up to
  isomorphism. In more detail: two cospans
  $f \colon A \xrightarrow{f_L} X \xleftarrow{f_R} B$,
  $g \colon A \xrightarrow{g_L} Y \xleftarrow{g_R} B$ are equivalent
  whenever there exists an iso $\iota\colon X\to Y$ such that
  $f_L;\iota = g_L$, $f_R;\iota = g_R$. Then, arrows are equivalence
  classes of cospans.
\end{toappendix}

A~cospan is \emph{left-linear} if its left leg $f_L$ is mono.  For adhesive
categories~\cite{ls:adhesive-journal}, the composition of left-linear
cospans again yields a left-linear cospan, and $\ILC(\mathbf D)$ is
the subcategory of $\Cospan(\mathbf D)$ where the arrows are
restricted to left-linear cospans.
% Left-linearity is useful since we rely on adhesive categories where
% pushouts along monos are well-behaved.  In particular, they always
% exist and form van Kampen squares, the latter being a requirement for
% representative squares based on borrowed context diagrams
% (\Cref{sec:representative-squares-shift}).

Note that $\graphfinj$ can be embedded into $\ILC(\graphf)$ by
transforming an injective graph morphism $f$ to a left-linear cospan with $f$ as the left leg
and $\id$ as the right leg.

Another application are Lawvere theories, where arrows are (tuples of)
terms, an approach we explore in \Cref{appendix-terms}.

\subsection{Generalized Nested Conditions}
\label{sec:reactive-systems-conditions}

% \cbox{Which reference for reactive systems?
%   \cite{lm:derive-bisimulation,l:congruences-reactive} or
%   \cite{s:rewrite-rules-congruences}?}

We consider (nested) conditions over an arbitrary category $\catC$ in
the spirit of reactive systems
\cite{lm:derive-bisimulation,l:congruences-reactive}.  Following
\cite{r:representing-fol,hp:correctness-nested-conditions}, we define
conditions as finite tree-like structures, where nodes are annotated
with quantifiers and objects, and edges are annotated with arrows.

\begin{definition}[Condition]
  \label{def:condition-new}
  Let $\catC$ be a category. A condition $\mathcal{A}$ over an object
  $A$ in $\catC$ is defined inductively as follows: it is either
  \begin{itemize}
    \item
      a finite conjunction of universals
      $\bigland\nolimits_{i \in \{1,\dots,n\}} \forall f_i.\mathcal A_i
      = \forall f_1.\mathcal A_1 \land \ldots \land \forall f_n.\mathcal A_n$, or
    \item
      a finite disjunction of existentials
      $\biglor\nolimits_{i \in \{1,\dots,n\}} \exists f_i.\mathcal A_i
      = \exists f_1.\mathcal A_1 \lor \ldots \lor \exists f_n.\mathcal A_n$
  \end{itemize}
  where $f_i\colon A\to A_i$ are arrows in $\catC$ and
  $\mathcal A_i \in \Cond(A_i)$ are conditions. We call
  $A = \RO(\mathcal A)$ the \emph{root object} of the condition
  $\mathcal A$.  Each subcondition $\mathcal Q f_i.\mathcal A_i$
  ($\mathcal Q \in \{\forall,\exists\}$) is called a \emph{child
    of~$\mathcal A$}.  The constants $\condtrue_A$ (empty conjunction) and
  $\condfalse_A$ (empty disjunction) serve as the base cases.
  We will omit subscripts in $\condtrue_A$ and $\condfalse_A$ when clear
  from the context. % move directly after constant definition?

  The set of all conditions over $A$ is denoted by $\Cond(A)$.
  % for the inductive definition.
% maybe we can avoid this?
%  For a condition $\mathcal A$, $\condchild(\mathcal A) = \{ (f_i, \mathcal A_i) \mid i \in I \}$ is the set of its top-level children in tuple form.
\end{definition}

Instantiated with $\graphf$ respectively $\graphfinj$, conditions are
equivalent to graph conditions as defined in
\cite{hp:correctness-nested-conditions}, and equivalence to
first-order logic has been shown in \cite{r:representing-fol}.
Intuitively, these conditions require the existence of certain subgraphs or patterns, or that whenever a given subgraph occurs, the surroundings of the match satisfy a child condition.
For instance,
$\forall\ \emptyset \to
\fcGraph{n1}{\node[gninlinable] (n1) at (0,0) {$1$}; \node[gninlinable] (n2) at (1,0) {$2$}; \draw[gedge] (n1) to (n2); }
\dotEx
\fcGraph{n1}{\node[gninlinable] (n1) at (0,0) {$1$}; \node[gninlinable] (n2) at (1,0) {$2$}; \draw[gedge] (n1) to (n2); }
\to
\fcGraph{n1}{\node[gninlinable] (n1) at (0,0) {$1$}; \node[gninlinable] (n2) at (1,0) {$2$}; \draw[gedge] (n1) to[bend left=10] (n2); \draw[gedge] (n2) to[bend left=10] (n1); }
\dotTrue$
requires that for every edge, a second edge in the reverse direction also exists.
For additional examples of conditions we refer to
\Cref{exa-prove-unsat,exa-find-finite-models,ex-ray-graphs} given
later.

% redundant anyway (see paragraph before def.):
%A condition $\mathcal{A}$ can be viewed as a finite tree, with
%conditions as nodes and $\RO(\mathcal{A})$ as the root object, and
%edges labelled with arrows.
To simplify our algorithms and their proofs,
the definition of conditions requires that conjunctions contain only
universal children and disjunctions only existential children (e.g.,
$\exists f.\mathcal A \land \exists g.\mathcal B$ is excluded).
However, this can be simulated using
$\forall \id.\exists f.\mathcal A \land \forall \id.\exists g.\mathcal B$, and similarly
for disjunctions of universals. Hence we sometimes write
$\mathcal{A}\land \mathcal{B}$ or $\mathcal{A}\lor \mathcal{B}$ for
arbitrary conditions in the proofs.
% In this regard they are equivalent to reactive systems conditions as
% defined in~\cite{bchk:conditional-reactive-systems,lmcs:8951}, but use
% a notation that is closer to that
% of~\cite{NegativeAC,hp:correctness-nested-conditions}.

% \oldtodo{
%   \textbf{B:} maybe parameterize the definition below by a class
%   $\mathcal{M}$ of arrows. Require that $c,\alpha\in\mathcal{M}$.
% }

While in \cite{bchk:conditional-reactive-systems} a model for a
condition was a single arrow, we have to be more general, since there
are some satisfiable conditions that have no finite models. Here we
want to work in categories of finite graphs (so that conditions are
finite), but at the same time we want to consider infinite models. The
solution is to evaluate conditions on infinite sequences of arrows
$\bar a = [ a_1, a_2, a_3, \tdots ]$, where
$A \xrightarrow{a_1} A_1 \xrightarrow{a_2} A_2 \xrightarrow{a_3}
\dots$, called \emph{composable sequences}.%
\footnote{Another option would be to work in the category of
  potentially infinite graphs. However, that would allow conditions
  based on infinite graphs for which satisfiability checks become
  algorithmically infeasible.}  We define
$\dom(\bar{a}) = \dom(a_1) = A$ and we call such a sequence
\emph{finite} iff for some index $k$ all $a_i$ with $i > k$ are
identities.

Intuitively, the model is represented by the ``composition'' of the
infinite sequence of arrows. In the category $\graphfinj$ this would
amount to taking the limit of this sequence. As we will later see, it
does not play a role how exactly an infinite structure is decomposed
into arrows, as all decompositions are equivalent with respect to
satisfaction. 

\begin{definition}[Satisfaction]
  \label{def:satisfaction}
  Let $\mathcal A \in \Cond(A)$.  Let
  $\bar a = [ a_1, a_2, a_3, \tdots ]$ be a composable sequence with
  $A = \dom(\bar{a})$.  We define the satisfaction relation
  $\bar a \models \mathcal A$ as follows:
  \begin{itemize}
  \item $\bar a \models \bigland_{i\in I} \forall f_i.\mathcal A_i$
    iff for every  $i\in I$ and every arrow
    $g \colon \RO(\mathcal A_i) \to B$ and \mbox{all $n \in \natzero$}
    we have: if~$a_1;\tdots;a_n = f_i;g$, then
    $[g, a_{n+1}, \tdots] \models \mathcal A_i$.
  \item $\bar a \models \biglor_i \exists f_i.\mathcal A_i$ iff there
    exists $i\in I$ and an arrow $g \colon \RO(\mathcal A_i) \to B$
    and \mbox{some $n \in \natzero$} such that $a_1;\tdots;a_n = f_i;g$
    and $[g, a_{n+1}, \tdots] \models \mathcal A_i$.
  \end{itemize}
\end{definition}
Note that this covers the base cases ($\bar a \models \condtrue$,
$\bar a \notmodels \condfalse$ for every sequence $\bar a$).
Furthermore $a_1;\tdots;a_n$ equals the identity whenever $n=0$.
For a finite sequence $\bar{a} = [a_1,\tdots,a_k,\id,\id,\tdots]$ this
means that $\bar{a} \models \mathcal{A}$ iff $a = a_1;\tdots;a_k$ is a
model for $\mathcal{A}$ according to the definition of satisfaction
given in \cite{bchk:conditional-reactive-systems}. In
this case we write $[a_1,\tdots,a_k]\models \mathcal{A}$ or simply
$a \models \mathcal{A}$.

\begin{remark}
  In the following we use $\Cond$ to denote the set\footnote{Actually,
    without restrictions these are proper classes rather than sets. We
    tacitly assume that we are working in the corresponding skeleton
    category where no two different objects are isomorphic and assume
    that we can consider $\Cond$, $\Seq$ as sets.} of all conditions
  and $\Seq$ as the set of all composable sequences of arrows
  (potential models).
\end{remark}

We write $\mathcal A \models \mathcal B$ ($\mathcal{A}$ implies
$\mathcal{B}$) if $\RO(\mathcal A) = \RO(\mathcal B)$ and for every
$\bar a$ with $\dom(\bar a) = \RO(\mathcal A)$ we
have: if $\bar a \models \mathcal A$, then
$\bar a \models \mathcal B$.  Two conditions are equivalent
($\mathcal A \equiv \mathcal B$) if $\mathcal A \models \mathcal B$
and $\mathcal B \models \mathcal A$.

Every condition can be transformed to an
equivalent condition that alternates between $\forall,\exists$ by
inserting $\forall \id$ or $\exists \id$ as needed. Such conditions
are called \emph{alternating}.

% ???
% We feel that viewing conditions in the cospan category 
%  $\mathbf{C}=\mathit{ILC}(\mathbf{D})$
%gives us an additional and very helpful layer of abstraction.

% Conjunctions of universal conditions and disjunctions of existential
% conditions can be easily done by merging their set of children.

We also define what it means for two conditions to be isomorphic.
%by
%requiring that the two trees are connected by isos on every level,
%forming commutative squares.
It is easy to see that isomorphic conditions are equivalent, but not
necessarily vice versa.

\begin{definition}[Isomorphic Conditions]
  For conditions $\mathcal A, \mathcal B$ and an iso
  \mbox{$h\colon\RO(\mathcal{B})\to \RO(\mathcal{A})$}, we say that
  $\mathcal A, \mathcal B$ are \emph{isomorphic}
  ($\mathcal A \condiso \mathcal B$) wrt.\ $h$, whenever both are
  universal, i.e.,\linebreak
  $\mathcal A = \bigland_{i\in I} \forall f_i. \mathcal{A}_i$,\ %
  $\mathcal B = \bigland_{j\in J} \forall g_j. \mathcal{B}_j$, and for
  each $i\in I$ there exists $j\in J$ and an~iso
  $h_{j,i}\colon \RO(\mathcal{B}_j) \to \RO(\mathcal{A}_i)$ such
  that $h;f_i = g_j;h_{j,i}$ and $\mathcal A_i \condiso \mathcal B_j$
  wrt.\ $h_{j,i}$; and vice versa (for each $j\in J$ there
  exists $i\in I$ \dots). Analogously if both conditions are
  existential.
\end{definition}

\begin{toappendix}
  \paragraph*{Equivalence laws for conditions}
  We rely on the results given in the following two propositions that
  were shown in \cite{bchk:conditional-reactive-systems}. They were originally
  stated for satisfaction with single arrows, but it is easy to see they are
  valid for possibly infinite sequences as well: conditions have finite depth
  and satisfaction only refers to finite prefixes of the sequence.

  \begin{proposition}[Adjunction]
    \label{prop:adjunction}
    Let $\mathcal{A},\mathcal{B}$ be two conditions with root object
    $A$, let $\mathcal{C},\mathcal{D}$ be two conditions with root object $B$
    and let $\phi\colon A\to B$. Then it holds that:
    \begin{enumerate}
    \item $\mathcal{A}\models \mathcal{B}$ implies
      $\mathcal{A}_{\downarrow \phi}\models \mathcal{B}_{\downarrow
        \phi}$.
    \item $\mathcal{C}\models \mathcal{D}$ implies
      $\quantor\phi.\mathcal{C}\models \quantor\phi.\mathcal{D}$ for
      $\quantor\in\{\exists,\forall\}$.
    \item
      $\exists \phi.(\mathcal{A}_{\downarrow \phi}) \models
      \mathcal{A}$ and for every $\mathcal{C}$ with
      $\exists \phi.\mathcal{C}\models \mathcal{A}$ we have that
      $\mathcal{C} \models \mathcal{A}_{\downarrow \phi}$.
    \item
      $\mathcal{A}\models \forall \phi.(\mathcal{A}_{\downarrow
        \phi})$ and for every $\mathcal{C}$ with
      $\mathcal{A}\models \forall \phi.\mathcal{C}$ we have that
      $\mathcal{A}_{\downarrow \phi}\models \mathcal{C}$.
    \end{enumerate}
  \end{proposition}

  % \subsection{Axioms for Conditions}
  % \label{sec:laws-conditions}

  \begin{proposition}[Laws for conditions]
    \label{prop:laws-conditions}
    One easily obtains the following laws for shift and
    quantification, conjunction and disjunction:
    \begin{align*}
      \mathcal{A}_{\downarrow\id} &\equiv \mathcal{A} &
      \mathcal{A}_{\downarrow \phi\comp\psi} &\equiv
      (\mathcal{A}_{\downarrow \phi})_{\downarrow \psi} \\
      \forall \id.\mathcal{A} &\equiv \mathcal{A} & \forall
      (\phi\comp\psi).\mathcal{A} &\equiv \forall \phi.\forall
      \psi.\mathcal{A} \\
      \exists \id.\mathcal{A} &\equiv \mathcal{A} & \exists
      (\phi\comp\psi).\mathcal{A} &\equiv \exists \phi.\exists
      \psi.\mathcal{A} \\[5pt]
      (\mathcal{A}\land \mathcal{B})_{\downarrow \phi} &\equiv 
      \mathcal{A}_{\downarrow \phi}\land \mathcal{B}_{\downarrow \phi} 
      &
      (\mathcal{A}\lor \mathcal{B})_{\downarrow \phi} &\equiv 
      \mathcal{A}_{\downarrow \phi}\lor \mathcal{B}_{\downarrow \phi} \\
      \forall \phi.(\mathcal{A}\land \mathcal{B}) &\equiv 
      \forall \phi.\mathcal{A} \land \forall \phi.\mathcal{B} &
      \exists \phi.(\mathcal{A} \lor \mathcal{B})  &\equiv 
      \exists \phi.\mathcal{A} \lor \exists \phi.\mathcal{B}
    \end{align*}
  \end{proposition}
\end{toappendix}

\subsection{Representative Squares and the Shift Operation}
\label{sec:representative-squares-shift}

We will now define the notion of representative squares, which
describe representative ways to close a span of arrows. They
generalize idem pushouts~\cite{lm:derive-bisimulation} and
borrowed context diagrams~\cite{ek:congruence-dpo-journal}.

% \pagebreak

\begin{definition}[Representative squares~\cite{bchk:conditional-reactive-systems}]\label{def:representative-squares}%
\floatingpicspaceright[4]{2.9cm}
\begin{floatingpicright}{2.8cm}%
  \hfill\begin{tikzpicture}[x=0.9cm,y=0.9cm,baseline=2pt]
    \node (a) at (0,0) {$A$};
    \node (b) at (1.5,0) {$B$};
    \node (c) at (0,-1.7) {$C$};
    \node (d) at (1.1,-1.3) {$D$};
    \node (ds) at (1.9,-2.0) {$D'$};
    \draw[->] (a) -- node[above,overlay]{$\alpha_1$} (b);
    \draw[->] (a) -- node[left]{$\alpha_2$} (c);

    \draw[->,fancydotted] (b) -- node[left,pos=0.3]{$\beta_1$} (d);
    \draw[->,fancydotted] (c) -- node[above,pos=0.35]{$\beta_2$} (d);

    \draw[->,fancydotted] ($(d.center)+(5pt,-5pt)$) -- ($(ds.center)+(-5pt,5pt)$);
    \node at ($(d.center)+(11pt,-4pt)$) {$\gamma$};
    \draw[->] (b) -- node[right]{$\delta_1$} (ds);
    \draw[->] (c) -- node[below,overlay]{$\delta_2$} (ds);
  \end{tikzpicture}
\end{floatingpicright}
A class $\kappa$ of commuting squares in a category $\catC$ is
\emph{represen\-tative}
if % $\kappa$ satisfies the following condition:
for every commuting square $\alpha_1 ; \delta_1 = \alpha_2 ; \delta_2$
in $\catC$ there exists a representative square
$\alpha_1 ; \beta_1 = \alpha_2 ; \beta_2$ in $\kappa$ and an arrow
$\gamma$ such that $\delta_1 = \beta_1;\gamma$ and
$\delta_2 = \beta_2;\gamma$.

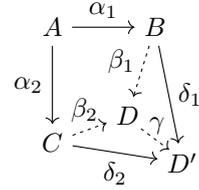
  \floatingpicspaceright[1]{2.9cm} For two arrows
  $\alpha_1 \colon A \to B,\ \alpha_2 \colon A \to C$, we define
  $\kappa(\alpha_1,\alpha_2)$ as the set of pairs of arrows
  $(\beta_1,\beta_2)$ which, together with $\alpha_1,\alpha_2$, form
  representative squares in $\kappa$.
\end{definition}

% Next sentence previously was something like "...and call every square in \kappa ... representative.
% This was removed already because of space restrictions. The remainder no longer makes sense now...
%In the following, we fix a \emph{representative} class $\kappa$ of
%squares.  Also note that the class of all squares of $\catC$ is
%representative.  

Compared to weak pushouts, more than one square might be needed to represent all commuting squares that extend a given span $(\alpha_1,\alpha_2)$.
In categories with pushouts (such as $\graphf$), pushouts are the most
natural candidate for representative squares.  In $\graphfinj$
pushouts do not exist, but jointly epi squares can be used
instead. For cospan categories, one can use borrowed context
diagrams~\cite{ek:congruence-dpo-journal} (see
Appendix~\ref{app:borrowed-context} for a summary).

For many categories of interest -- such as $\graphf$ and
$\ILC(\graphf)$ -- we can guarantee a choice of $\kappa$ such that
each set $\kappa(\alpha_1,\alpha_2)$ is finite and computable.
In the rest of this paper, we assume that we work in such
a category, and use such a class $\kappa$. Hence the
constructions described below are effective since the finiteness of
the transformed conditions is preserved.
%
% \oldtodo{\textbf{B:} Mention that for Lawvere theories, representative squares are
%   related to unification.}
%
\begin{toappendix}
  \paragraph*{Borrowed context diagrams}
  \label{app:borrowed-context}
% \footnote{Consider the cospan $n \colon \emptyset \rightarrow N \leftarrow \emptyset$, where $N$ is a single isolated node. No pushout exists for $(n,n)$.}

For cospan categories over adhesive categories (such as $\ILC(\graphf)$), borrowed context
diagrams -- initially introduced as an extension of DPO
rewriting~\cite{ek:congruence-dpo-journal} -- can be used as representative squares. Before
we can introduce such diagrams, we first need the notion of jointly
epi.

\begin{definition}[Jointly epi]
  A pair of arrows $f \colon B \to D,\ g \colon C \to D$ is
  \emph{jointly epi} (\emph{JE}) if for each pair of arrows
  $d_1, d_2 \colon D \to E$ the following holds: if $f;d_1 = f;d_2$
  and $g;d_1 = g;d_2$, then $d_1 = d_2$.
\end{definition}

In $\graphf$ jointly epi equals jointly surjective, meaning
that each node or edge of $D$ is required to have a preimage under $f$
or $g$ or both (it contains only images of $B$ or~$C$).

This criterion is similar to, but weaker than a pushout:
For jointly epi graph morphisms $d_1 \colon B \to D,\ d_2 \colon C \to D$,
there are no restrictions on which elements of $B,C$ can be merged in~$D$.
However, in a pushout constructed from morphisms $a_1 \colon A \to B,\ a_2 \colon A \to C$,
elements in~$D$ can (and must) only be merged if they have a common preimage in $A$.
(Hence every pushout generates a pair of jointly epi arrows, but not vice versa.)

\begin{definition}[Borrowed context diagram~\cite{DBC-CRS}]
  A commuting diagram in the category $\ILC(\catC)$, where $\catC$ is
  adhesive, is a \emph{borrowed context diagram} whenever it has the
  form of the diagram shown in \Cref{fig-bc-diag}, and the four
  squares in the base category $\catC$ are pushout (PO), pullback (PB)
  or jointly~epi (JE) as indicated. Arrows depicted as
  $\rightarrowtail$ are mono. In particular $L\rightarrowtail G^+$,
  $G\rightarrowtail G^+$ must be jointly epi.
\end{definition}

\Cref{fig-bc-venn} shows a more concrete version of
\Cref{fig-bc-diag}, where graphs and their overlaps are depicted by
Venn diagrams (assuming that all morphisms are injective). Because of
the two pushout squares, this diagram can be interpreted as
composition of cospans $a;f = \ell;c = D \rightarrow G^+ \leftarrow K$
with extra conditions on the top left and the bottom right square.
The top left square fixes an overlap $G^+$ of $L$ and $G$, while $D$
is contained in the intersection of $L$ and $G$ (shown as a hatched
area). Being jointly epi ensures that it really is an overlap and does
not contain unrelated elements.
The top right pushout corresponds to the left
pushout of a DPO rewriting diagram. It contains a total match of $L$
in $G^+$.  Then, the bottom left pushout gives us the minimal borrowed
context $F$ such that applying the rule becomes possible.
The top left and the bottom left squares together ensure that the contexts to be considered are not larger than necessary.
The bottom right pullback ensures that the interface $K$ is as large
as possible.

For more concrete examples of borrowed context diagrams, we refer to~\cite{ek:congruence-dpo-journal,lmcs:8951}.

\begin{figure}[t]
  \begin{subfigure}[b]{0.45\textwidth}
    \centering\begin{tikzpicture}[x=1.10cm,y=1.10cm]% PAGEBREAK-ADJUST
      \def\sqw{1.5}
      \def\sqh{1.48}
      \pgfmathsetmacro\sqww{2*\sqw}
      \pgfmathsetmacro\sqhh{2*\sqh}

      % very tight fit (the normal bounding box leaves a bit of space above and below the top and bottom cospan labels, which makes it hard to properly align it with the rounded boxes in the other figure)
      \path[use as bounding box] (-0.7,0.8) rectangle (2*\sqw+0.7,-2*\sqh-0.87);

      \node (d) at (0,0) {$D$};
      \node (l) at (\sqw,0) {$L$};
      \node (i) at (\sqww,0) {$I$};
      \node (g) at (0,-\sqh) {$G$};
      \node (gp) at (\sqw,-\sqh) {$\mkern-3mu G^+\mkern-5mu$};
      \node (c) at (\sqww,-\sqh) {$C$};
      \node (j) at (0,-\sqhh) {$J$};
      \node (f) at (\sqw,-\sqhh) {$F$};
      \node (k) at (\sqww,-\sqhh) {$K$};

      \def\sqlfs{\footnotesize}
      \node at (0.5*\sqw,-0.5*\sqh) {\sqlfs\textup{JE}};
      \node at (1.5*\sqw,-0.5*\sqh) {\sqlfs\textup{PO}};
      \node at (0.5*\sqw,-1.5*\sqh) {\sqlfs\textup{PO}};
      \node at (1.5*\sqw,-1.5*\sqh) {\sqlfs\textup{PB}};

      \draw[cospanintmono]
        (d) edge (l) edge (g)
        (l) edge (gp)
        (g) edge (gp)
        (i) edge (c)
        (j) edge (f);
      \draw[cospanint]
        (j) edge (g)
        (k) edge (f) edge (c)
        (f) edge (gp)
        (i) edge (l)
        (c) edge (gp);

      \draw[cospanarr] (d.north east) to[bend left =10] node[above]{$\ell$} (i.north west);
      \draw[cospanarr] (d.south west) to[bend right=10] node[left ]{$a$} (j.north west);
      \draw[cospanarr] (i.south east) to[bend left =10] node[right]{$c$} (k.north east);
      \draw[cospanarr] (j.south east) to[bend right=10] node[below]{$f$} (k.south west);
    \end{tikzpicture}%
    \caption{Structure of a borrowed context diagram.
    The inner, lighter arrows are morphisms of the base category $\catC$,
    while the outer arrows are morphisms of $\ILC(\catC)$.
    }%
    \label{fig-bc-diag}%
  \end{subfigure}\hfill%\hspace{1cm}
  \begin{subfigure}[b]{0.52\textwidth}
    \centering
    \begin{tikzpicture}
      \def\heightAdjustA{1.00}% PAGEBREAK-ADJUST
      \def\heightAdjustB{1.00}% PAGEBREAK-ADJUST
      % general diagram layout (3x3)
      \pgfmathsetmacro\posX{2.6}
      \pgfmathsetmacro\posY{-1.65}
      \pgfmathsetmacro\boxSxp{1.1}
      \pgfmathsetmacro\boxSxm{0.95}
      \pgfmathsetmacro\boxSyp{0.68*\heightAdjustB}
      \pgfmathsetmacro\boxSym{0.62*\heightAdjustB}
      \pgfmathsetmacro\boxW{\boxSxp+\boxSxm}
      \pgfmathsetmacro\boxH{\boxSyp+\boxSym}

      % graph/circle shapes
      \def\circGinner{((0.05,0) ellipse[x radius=0.53, y radius=0.24]}
      \def\circGouter{((0.05,0) ellipse[x radius=0.65, y radius=0.36]}
      \def\circLinner{((0.6,0.14) circle[radius=0.28]}
      \def\circLouter{((0.6,0.14) circle[radius=0.4]}
      \def\allArea{(-\boxSxm,-\boxSym) rectangle (\boxSxp,\boxSyp)} % for clipping

      % draw all outlines
      \newcommand{\drawgl}{%
        \draw \circGinner;
        \draw \circGouter;
        \draw \circLinner;
        \draw \circLouter;
      }

      % fill style
      \tikzset{vennfill/.style={fill=black!40}}

      % shorthand for drawing both circle fills
      % (usually, this follows a series of clip commands)
      \newcommand{\fillboth}{%
        \fill[vennfill] \circGouter;
        \fill[vennfill] \circLouter;
      }

      % draw a G-L overlap (with box around it and so on)
      % at given position with specified fill commands
      \newcommand{\boxat}[5]{% x y name label fillcommands
        \begin{scope}[shift={(#1*\posX,#2*\posY*\heightAdjustA)}]
          \node (#3) at ($(\boxSxp,\boxSyp)!.5!(-\boxSxm,-\boxSym)$) [draw,roundbox,cospanintcommon,minimum width=\boxW cm, minimum height=\boxH cm] {};
          \node[anchor=south west] at (-\boxSxm,-\boxSym) {$\scriptstyle{#4}$};

          \begin{scope}[even odd rule]
            #5
          \end{scope}
          \drawgl
        \end{scope}
      }

      \boxat00DD{% fill the intersection of L and G, but in a different shade
        \clip \circLouter;
        \clip \circGouter;
        \foreach \d in {0,1,...,20} {
          \draw[black!40, line width=0.7pt] (0, -1cm + 0.3pt + \d*2pt) -- +(1cm,1cm);
        }}
        %\fill[fill=black!30] \circGouter;}

      \boxat10LL{% fill all of L
        \fill[vennfill] \circLouter;}

      \boxat20II{% fill only the interface of L
        \clip \circLinner \allArea;
        \fill[vennfill] \circLouter;}

      \boxat01GG{% fill all of G
        \fill[vennfill] \circGouter;}

      \boxat11{Gp}{G^+}{% fill all of L and G
        \fillboth}

      \boxat21CC{% fill all except the inside of L
        \clip \circLinner \allArea;
        \fillboth}

      \boxat02JJ{% fill only the interface of G
        \clip \circGinner \allArea;
        \fill[vennfill] \circGouter;}

      \boxat12FF{% fill all except the inside of G
        \clip \circGinner \allArea;
        \fillboth}

      \boxat22KK{% fill both interfaces, but not the inside
        \clip \circGinner \allArea;
        \clip \circLinner \allArea;
        \fillboth}

      % connecting the boxes
      \def\arrowfromto#1#2#3{
        % ignore mono for readability (no cospanint#3)
        \draw[cospanint] (#1) edge (#2);
      }
      \def\cospanfromto#1#2#3{% left middle right | name namelabelpos anchor namelabelopts
        \arrowfromto{#1}{#2}{mono}
        \arrowfromto{#3}{#2}{}
      }
      \def\polabels#1#2#3#4{
        \def\sqlfs{\scriptsize\color{gray}}
        \node at ($(D)!0.5!(Gp)$) {\sqlfs\textup{#1}};
        \node at ($(L)!0.5!(C)$) {\sqlfs\textup{#2}};
        \node at ($(G)!0.5!(F)$) {\sqlfs\textup{#3}};
        \node at ($(Gp)!0.5!(K)$) {\sqlfs\textup{#4}};
      }

      \cospanfromto{D}{L}{I} %{\ell}{above}{north}{bend left=10}
      \cospanfromto{D}{G}{J} %{a}{left}{west}{bend right=10}
      \cospanfromto{J}{F}{K}
      \cospanfromto{I}{C}{K}
      \cospanfromto{L}{Gp}{F}
      \cospanfromto{G}{Gp}{C}
      \polabels{JE}{PO}{PO}{PB}
    \end{tikzpicture}
    \caption{Borrowed context diagrams represented as Venn diagrams.
    The outer circles represent graphs $L,G$, and
    the area between the inner and outer circles represents their interfaces $I,J$.}
    \label{fig-bc-venn}
  \end{subfigure}
  \caption{Borrowed context diagrams}
\end{figure}
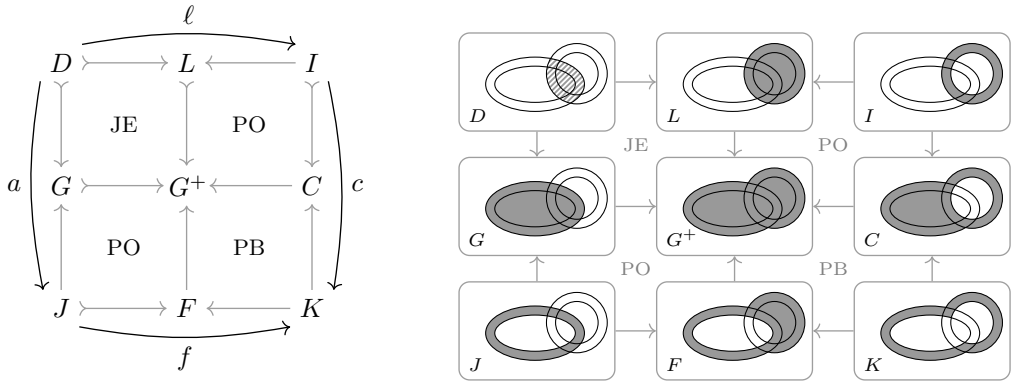

For cospan categories over adhesive categories, borrowed context
diagrams form a represen\-tative class of
squares~\cite{bchk:conditional-reactive-systems}. Furthermore, for
some categories (such as $\graphfinj$), there are -- up to isomorphism
-- only finitely many jointly epi squares for a given span of monos
and hence only finitely many borrowed context diagrams given $a,\ell$
(since pushout complements along monos in adhesive categories are
unique up to isomorphism).

Whenever the two cospans $\ell,a$ are in $\ILC(\graphfinj)$, it is
easy to see that $f,c$ are in  $\ILC(\graphfinj)$, i.e., they consist
only of monos, i.e., injective morphisms.

Note also that representative squares in $\graphfinj$ are simply
jointly epi squares and they can be straighforwardly extended to
squares of $\ILC(\graphfinj)$.
\end{toappendix}
%
% \subsection{Shift}
% \label{sec:shift}
%
% \subsection{Shift and Quantification of Conditions}
% \label{sec:shift-quantification-conditions}
%

One central operation is the shift of a condition along an arrow.  The
name shift is taken from an analogous operation for nested application
conditions (see \cite{p:development-correct-gts}).
% Intuitively a shift corresponds to a partial evaluation, where we
% assume that the arrows on which the condition is to be evaluated are
% of the form $\phi\comp c$ for a fixed $\phi$.

\cbox{Check whether the term ``shift'' is used in
  \cite{eept:fundamentals-agt}. Where was this operation defined
  first?} 

\begin{definition}[Shift of a Condition]
  \label{def:shift}
  \label{prop:shift}
  Given a fixed class of representative squares $\kappa$,
  the \emph{shift of a condition $\mathcal A$ along an arrow
  $c \colon \RO(\mathcal A) \to B$} is inductively defined as follows:
  \[
    \Big( \bigland_{i \in I} \forall f_i.\mathcal A_i
    \Big)_{\downarrow c} = \bigland_{i \in
        I}\bigland_{(\alpha,\beta) \in \kappa(f_i,c)}
    \forall \beta.({\mathcal A_i}_{\downarrow \alpha})
    \mkern35mu\hfill% TODO: consider moving somewhere else
    \begin{tikzpicture}[baseline={(0,-0.2)}]
      \path[use as bounding box] (-0.2,0) rectangle (1.6,-0.75);
      \node (sqtl) at (0,0) {};
      \node (sqtr) at (1.25,0) {};
      \node (sqbl) at (0,-0.9) {};
      \node (sqbr) at (1.25,-0.9) {};

      \draw[->] (sqtl) -- node[overlay,above]{$f_i$} (sqtr);
      \draw[->] (sqtl) -- node[left]{$c$} (sqbl);
      \draw[->] (sqtr) -- node[right]{$\alpha$} (sqbr);
      \draw[->] (sqbl) -- node[overlay,below]{$\beta$} (sqbr);
    \end{tikzpicture}
  \]
  Shifting of existential conditions is performed analogously.
\end{definition}

% It is easy to see that shift distributes over conjunction and
% disjunction:
% $\big( \bigland_{i \in I} \forall f_i.\mathcal A_i \big)_{\downarrow
%   c} = \bigland_{i \in I} \big( \forall f_i.\mathcal A_i
% \big)_{\downarrow c}$.

The shift operation can be understood as a partial evaluation of
$\mathcal A$ under the assumption that $c$ is already ``present''.  In
particular it satisfies
$[c;d_1,d_2,\tdots] \models \mathcal{A} \iff [d_1,d_2,\tdots] \models
\mathcal{A}_{\downarrow c}$. This implies that while the
representation of the shifted condition may differ depending on the
chosen class of representative squares, the resulting conditions are
equivalent. Since we assume that each set $\kappa(f_i,c)$ is finite,
shifting a finite condition will again result in a finite condition.

As an example in $\graphfinj$, shifting $\forall\ \emptyset \to
\fcGraph{n1}{\node[gninlinable] (n1) at (0,0) {$1$}; }
\dotEx
\fcGraph{n1}{\node[gninlinable] (n1) at (0,0) {$1$}; }
\to
\fcGraph{n1}{\node[gninlinable] (n1) at (0,0) {$1$}; \node[gninlinable] (n2) at (1,0) {$2$}; \draw[gedge] (n1) to (n2); }
\dotTrue$
(every node has an outgoing edge)
over
$\emptyset \to \fcGraph{n0}{\node[gninlinable] (n0) at (0,0) {\phantom1}; }$
(a node exists)
yields
$\forall\ %
\fcGraph{n0}{\node[gninlinable] (n0) at (0,0) {\phantom1}; }
\to
\fcGraph{n0}{\node[gninlinable] (n0) at (0,0) {\phantom1}; }
\dotEx
\fcGraph{n0}{\node[gninlinable] (n0) at (0,0) {\phantom1}; }
\to
\fcGraph{n1}{\node[gninlinable] (n0) at (0,0) {}; \node[gninlinable] (n1) at (1,0) {$1$}; \draw[gedge] (n0) to (n1); }
\dotTrue \land
\mathmbox{\forall\ %
\fcGraph{n0}{\node[gninlinable] (n0) at (0,0) {\phantom1}; }
\to
\fcGraph{n1}{\node[gninlinable] (n0) at (-0.75,0) {}; \node[gninlinable] (n1) at (0,0) {$1$}; }}
\ .\big(
\exists\ %
\fcGraph{n1}{\node[gninlinable] (n0) at (-0.75,0) {}; \node[gninlinable] (n1) at (0,0) {$1$}; }
\to
\fcGraph{n1}{\node[gninlinable] (n0) at (-0.75,0) {}; \node[gninlinable] (n1) at (0,0) {$1$}; \node[gninlinable] (n2) at (1,0) {$2$}; \draw[gedge] (n1) to (n2); }
\dotTrue \lor \exists\ %
\fcGraph{n1}{\node[gninlinable] (n0) at (-0.75,0) {}; \node[gninlinable] (n1) at (0,0) {$1$}; }
\to
\fcGraph{n1}{\node[gninlinable] (n0) at (-1,0) {}; \node[gninlinable] (n1) at (0,0) {$1$}; \draw[gedge] (n1) to (n0); }
\dotTrue \big)
$
(the designated node has an outgoing edge, and so does every other
node, possibly to the designated node).

\begin{toappendix}
  \paragraph*{Visualization of shifts}
  Given a condition $\mathcal{A}$ and an arrow $c\colon A =
  \RO(\mathcal{A})\to B$, we will visualize shifts in diagrams as follows:

\noindent\mbox{}\hfill\begin{tikzpicture}
  \def\sqw{1.75}
  % \node (z) at (-1*\sqw,0) {$Z$};
  \node (a) at (0,0) {$A$}; \node (b) at (1*\sqw,0) {$B$}; \node (x)
  at (2*\sqw,0) {$X$};
  % \draw[->,fancydotted] (z) -- node[above]{$a$} (a);
  \draw[->] (a) -- node[above]{$c$} (b); \draw[->] (b) --
  node[above]{$d$} (x); \node[condtri,dart tip angle=30,shape border
  rotate=270,rotate around={-10:(a.center)},
  label={[rotate=0,anchor=south,label
    distance=1pt]above:{$\mkern7mu\mathcal
      A\vphantom{{}_{\downarrow}}$}}] at (a.north) {\kern3pt};
  \node[condtri,dart tip angle=30,shape border rotate=270,rotate
  around={-10:(b.center)}, label={[rotate=0,anchor=south,label
    distance=1pt]above:{$\mathcal A_{\downarrow c}$}}] at (b.north)
  {\kern3pt};
  % \node at (0,-1.4) {}; % PAGEBREAK-ADJUST % ???
\end{tikzpicture}\hfill\mbox{}

Remember that for an arrow $d\colon B\to X$ it holds that $d\models
\mathcal{A}_{\downarrow c} \iff c;d \models \mathcal{A}$.
\end{toappendix}

\begin{toappendix}
  \paragraph*{Proofs}

If one arrow is a section, the following property is useful:
  
\begin{lemmarep}[Representative squares preserve sections]
  \label{rsq-preserve-sections}
  Let $\alpha_1 \colon A \to B,\ \alpha_2 \colon A \to C$.  If
  $\alpha_2$ is a section, then there exists some
  $(\beta_1, \beta_2) \in \kappa(\alpha_1, \alpha_2)$ such that
  $\beta_1$ is a section.
\end{lemmarep}

\begin{proof}
  Consider the commuting square given below on the left, where $r_2$
  is some right-inverse of $\alpha_2$, i.e., $\alpha_2;r_2 = \id$.
  Since every commuting square can be reduced to a representative
  square, there exist $\beta_1, \beta_2, \gamma$ as given in the
  diagram below on the right.  Since $\beta_1 ; \gamma = \id_B$ is an
  isomorphism, $\beta_1$ is a section and $\gamma$ a retraction.

  \centering\begin{tikzpicture}[x=1.10cm,y=1.10cm]% PAGEBREAK-ADJUST

  \def\sqw{2}
  \def\sqh{2}
  \def\sqhsp{0.3} % how much to inset C, B compared to the full D' square size
  \def\sqhspi{0.4} % how much to additionally inset D (in addition to \sqhsp)
  \def\sqhsr{0.25} % shrink of the left side. was 0 in masters thesis, but a tiny offset is acceptable (it should however not be exactly halfway between D and D', because this might suggest that D' is "split" to some smaller object D and some larger-than-D' object that is also called D'
  \def\sqhsrr{0.25} % like sqhsr but for vertical size of the first diagram

  % %%% LEFT - PRE R-STEP %%%
  \begin{scope}[shift={(0,0)}]
    \node (a) at (0,0) {$A$};
    \node (b) at (\sqw-\sqhsr,0) {$B$};
    \node (c) at (0,-\sqh+\sqhsrr) {$C$};
    \node (ds) at (\sqw-\sqhsr,-\sqh+\sqhsrr) {$B$};
    \draw[->] (a) -- node[above]{$\alpha_1$} (b);
    \draw[->] (a) -- node[left]{$\alpha_2$} (c);

    \draw[->] (b) -- node[right]{$\id_B$} (ds);
    \draw[->] (c) -- node[below]{$r_2 ; \alpha_1$} (ds);
  \end{scope}

  \node at (\sqw+1,-0.5*\sqh) {$\rightarrow$};

  % %%% RIGHT %%%
  \begin{scope}[shift={(\sqw+2.5,0)}]
    \node (a) at (0,0) {$A$};
    \node (b) at (\sqw-\sqhsp,0) {$B$};
    \node (c) at (0,-\sqh+\sqhsp) {$C$};
    \node (d) at (\sqw-\sqhsp-\sqhspi,-\sqh+\sqhsp+\sqhspi) {$D$};
    \node (ds) at (\sqw,-\sqh) {$B$};
    \draw[->] (a) -- node[above]{$\alpha_1$} (b);
    \draw[->] (a) -- node[left]{$\alpha_2$} (c);
    \draw[->] (b) -- node[left,pos=0.25]{$\beta_1$} (d);
    \draw[->] (c) -- node[above,pos=0.3]{$\beta_2$} (d);

    \draw[->] (d.center)+(5pt,-5pt) -- (ds);
    \node at ($(d.center)+(12pt,-4pt)$) {$\gamma$};
    \draw[->] (b) -- node[right]{$\id_B$} (ds);
    \draw[->] (c) -- node[below]{$r_2 ; \alpha_1$} (ds);
  \end{scope}
  \end{tikzpicture}%
\end{proof}

\end{toappendix}

In the case where $\alpha_1$ in \Cref{def:representative-squares} is
an iso, we can always assume that
$\kappa(\alpha_1,\alpha_2) = \{(\alpha_1^{-1};\alpha_2,\ \id)\}$ and
we will use this assumption in the paper.

% A consequence of this lemma is that shifting a condition along an isomorphism can be assumed to not increase the size of the condition.
% In the rest of the paper, we assume that $\kappa$ has this property.

\subsection{Further Concepts}

% There are examples of graph conditions with (only) infinite models, such as \Cref{ex-ray-graphs}.
% It then depends on the category whether this condition is satisfiable or not:
% in $\graphf$, this condition is unsatisfiable, because $\graphf$ has only finite graphs and there is no finite graph that satisfies this condition.
% However, in the category of infinite graphs, the condition is satisfiable.

Our goal is to develop a procedure that finds a finite model if one
exists, produces unsatisfiability proofs if a condition has neither
finite nor infinite models, and otherwise does not terminate. In order
to state the correctness of this procedure, we will need an abstract
notion of finiteness and to this aim we introduce \emph{finitely
  decomposable morphisms}.
  %\oldtodo{\textbf{B:} compare with ``locally
  %presentable'' \cite{ar:locally-presentable-categories} or
  %``compact''?}
Intuitively this means that every infinite
decomposition contains only finitely many non-isomorphisms.

% We will not represent infinite models directly (i.e.\ as a single
% infinite graph or morphism), but instead as infinite
% \emph{sequences} of arrows, each of which describe a finite graph.
% In this setting, a finite model is then represented as a finite
% sequence of such arrows.

% This construction gives the intended results when instantiated with $\graphf$.
% Naturally, it could also be used with \emph{infinite} graphs, in which case we would of course lose the property that a finite sequence of arrows always represents a finite model.
% To better characterize categories for which this construction is effective, we introduce the notion of \emph{finitely decomposable} morphisms:

\begin{definition}[Finitely decomposable morphism]
  \label{def:f-d-morphism}
  A morphism $m \colon A \to B$ is \emph{finitely decomposable} if for
  every infinite sequence of $(f_i, g_i)$, $i\in\natzero$, such that
  $f_0 = m$ and $f_i = g_i ; f_{i+1}$ (cf.\ the diagram below), only
  finitely many $g_i$ are non-isomorphisms.
\end{definition}

\iffalse
\noindent\mbox{}\hfill\begin{tikzpicture}[x=3cm,y=-1cm]
  \node (a) at (0,0) {$A$};
  \node (b) at (1,0) {$B$};
  \node (n1) at (0,1) {};
  \node (n2) at (0,2) {};
  \node (n3) at (0,2.2) {$\vdots$};
  \draw[->] (a) -- node[above]{$m = f_0$} (b);
  \draw[->] (a) -- node[left]{$g_0$} (n1);
  \draw[->] (n1) -- node[left]{$g_1$} (n2);
  \draw[->] (n1) to[bend right=10] node[above,pos=0.2]{$f_1$} (b.220);
  \draw[->] (n2) to[bend right=15] node[above,pos=0.2]{$f_2$} (b.250);
\end{tikzpicture}\hfill\mbox{}
\else
\noindent\mbox{}\hfill\begin{tikzpicture}[x=2cm,y=-1cm]
  \node (a) at (0,0) {$A$};
  \node (b) at (0,1) {$B$};
  \node (n1) at (1,0) {};
  \node (n2) at (2,0) {};
  \node (n3) at (2.2,0) {$\cdots$};
  \draw[->] (a) -- node[left]{$m = f_0$} (b);
  \draw[->] (a) -- node[above,overlay]{$g_0$} (n1);
  \draw[->] (n1) -- node[above,overlay]{$g_1$} (n2);
  \draw[->] (n1.255) to[bend left=10] node[above,pos=0.75]{$f_1$} (b.25);
  \draw[->] (n2.255) to[bend left=15] node[above,pos=0.5]{$f_2$} (b.0);
\end{tikzpicture}\hfill\mbox{}
\fi

% The notion of finite decomposability intuitively describes models
% that can only be created and found in finitely many steps, which
% typically implies it is finite.  For instance, injective morphisms
% between finite graphs satisfy this definition: they can be
% decomposed into a sequence of morphisms that add just a single graph
% element, and cannot be decomposed further (except into
% isomorphisms).  Noninjective graph morphisms are typically not
% finitely decomposable, since infinite decompositions that split and
% merge the same elements exist, yielding infinitely many
% non-isomorphisms.  Morphisms involving (at least one) infinite graph
% also violate this property: an infinite morphism can be decomposed
% into an infinite sequence of (finite) morphisms that each add a
% single graph element, and hence are infinitely many non-isos.

Note that in $\graphfinj$ all arrows are finitely decomposable, while
this is not the case in~$\graphf$.
%
% For cospans of graphs, it is not sufficient to restrict to finite
% graphs as objects of the category, since this only restricts to
% cospans with finite interfaces, but leaves their inner objects
% unrestricted.  Here, finite decomposability ensures that the inner
% object is also finite.
%
% The f.d.\ property will later be used to show that every finitely decomposable (i.e., ``finite'') model can be found in finite time if it exists.
%
% There also is a related, but weaker, property of a category that
% \emph{all sections are isos}.  An example for a category that
% satisfies ``all sections are isos'', but where not all morphisms are
% finitely decomposable, is that of infinite graphs and injective graph
% morphisms.  (In that case, morphisms that are not isos would only have
% a right inverse that is not injective, but if the category is
% restricted to injective morphisms, this right inverse does not exist.)
%
In $\graphf$, there exists a section $s$ (with associated retraction
$r$ such that $s;r = \id$) that is \emph{not} an iso (example:
$s = \mbox{$\fcGraph{n1}{\node[gninlinable] (n1) at (0,0)
    {$1$};} \rightarrow \fcGraph{n1}{\node[gninlinable] (n1) at (0,0)
    {$1$}; \node[gninlinable] (n2) at (0.7,0) {$2$};}$}$,
$r = \mbox{$\fcGraph{n1}{\node[gninlinable] (n1) at (0,0)
    {$1$}; \node[gninlinable] (n2) at (0.7,0)
    {$2$};}\rightarrow \fcGraph{n1}{\node[gninlinable] (n1) at (0,0)
    {$1$};} $}$).  Then, the identity on the domain of $s$ has a
decomposition into infinitely many non-isos (an alternating sequence
of $s$ and $r$, more concretely: $g_{2i} = s$, $g_{2i+1} = r$ and
$f_{2i} = \id$, $f_{2i+1} = r$) and is hence not finitely
decomposable.

% Finally, we also modify the definition of satisfaction to work with
% models represented as (in)finite sequences of arrows.  Note that there
% might not be a one-to-one correspondence between the structure of the
% condition and the sequence of arrows that make up a particular model.
% In particular, it is not clear how many of the leading arrows
% contribute to matching the arrow $f$ of a condition $\exists f.\dots$
% (this applies to both finite and infinite graphs):

% Note that $n \in \natzero$ also permits $n=0$, which is necessary for
% satisfaction to work as intended: Let $f \ne \id$ be some arrow, then
% it is expected that $[f] \models \exists \id . \exists f . \condtrue$.
% Here, $n=0$ would be used for the outer (leftmost) existential
% quantifier.

% various fixpoint theory definitions (compatibility etc.):
% \Cref{sec:fixpoint-theory} (feel free to move)

While satisfaction is typically defined inductively (as in
\Cref{def:satisfaction}), i.e., as a least fixpoint, we can also view
it \coinductively, i.e., as a  greatest fixpoint, due to the fact that
all conditions are finite.

\begin{propositionrep}[Fixpoint function for satisfaction]
  \label{prop:satisfaction-function}
  Let $\bar a = [ a_1, a_2, a_3, \tdots ]\in \Seq$ be a
  composable sequence of arrows.  We define the function
  $s \colon \mathcal P(\Seq \times \Cond) \to \mathcal
  P(\Seq \times \Cond)$ as follows: Let
  $P\subseteq \Seq\times \Cond$, then
  \begin{itemize}
  \item $(\bar a, \bigland_i \forall f_i.\mathcal A_i) \in s(P)$ iff
    for every $i\in I$ and every arrow
    $g \colon \RO(\mathcal A_i) \to B$ and all~$n \in \natzero$ we
    have: if~$a_1;\tdots;a_n = f_i;g$, then
    $([g, a_{n+1}, \tdots],\mathcal A_i)\in P$.
  \item $(\bar a, \biglor_i \exists f_i.\mathcal A_i) \in s(P)$ iff
    there exists $i\in I$ and an arrow
    $g \colon \RO(\mathcal A_i) \to B$ and some~$n \in \natzero$ such
    that $a_1;\tdots;a_n = f_i;g$ and
    $([g, a_{n+1}, \tdots],\mathcal A_i)\in P$.
  \end{itemize}
  The least and greatest fixpoint of $s$ ($\mu s$, $\nu s$) \coincide
  and they equal the satisfaction relation~$\models$.
\end{propositionrep}

\begin{proof}
  First note that \Cref{def:satisfaction} is seen inductively the
  least fixpoint spelled out there is exactly the least fixpoint of
  $s$, hence ${\models} = \mu s$.
  
  It also holds that $\nu s = {\models}$ (the greatest fixpoint of $s$
  equals the satisfaction relation).
  The correctness of this characterization (i.e.\ that ``$s$ is the
  right function'') can be shown as follows:
  ${\models} \subseteq \nu s$ because $\models$ is a fixpoint of $s$.
  To show $\nu s \subseteq {\models}$, let
  $(\bar c, \mathcal A) \in \nu s = s(\nu s)$ (the latter holds
  because $\nu s$ is a fixpoint).  We show
  $(\bar c, \mathcal A) \in {\models}$ by induction on the nesting
  depth of $\mathcal A$:
\begin{itemize}
\item
  Let the depth be 1, i.e., $\mathcal A$ has no children.
  Then it must be an empty universal quantification, i.e., $\condtrue$.
  Hence also $\bar c \models \mathcal A$.
\item
  Let the statement hold for depth $d$ and let $\mathcal A$ have depth $d+1$.
  Let $(\bar c, \mathcal A) \in s(\nu s)$.
  If $\mathcal A = \biglor_i \exists f_i.\mathcal A_i$ is existential,
  there exists $n, i, g$ such that $c_1;\dots;c_n = f_i;g$ and $([g, c_{n+1}, \dots], \mathcal A_i) \in \nu s$.
  $\mathcal A_i$ has depth at most $d$, so by the induction hypothesis also $[g, c_{n+1}, \dots] \models \mathcal A_i$.
  Then also $\bar c \models \mathcal A$.
  The universal case $\mathcal A = \bigland_i \forall f_i.\mathcal A_i$ can be shown analogously.
  \qedhere
\end{itemize}
\end{proof}

\section{Satisfiability Checking in the Restricted Case}
\label{sec:satisfiability}

%\oldtodo{\textbf{Le:} I propose to give categories where all sections are
%  isos a proper name. Maybe ``expanding categories''.}
%
Given a condition $\mathcal{A}$, we want to know whether $\mathcal{A}$
is satisfiable and generates a finitely decomposable model, if it
exists.
Here we provide a procedure that works under the assumption that we
are working in a category where all sections are isos. This is for
instance true for $\graphfinj$ and $\ILC(\graphfinj)$, where
non-trivial right-inverses do not exist. It does not hold for
non-injective graph morphisms (see counterexample above) or
left-linear cospans (counterexample:
\mbox{$\id = \fcGraph{n1}{\node[gninlinable] (n1) at (0,0)
    {$1$};} \rightarrow \fcGraph{n1}{\node[gninlinable] (n1) at (0,0)
    {$1$};} \leftarrow \fcGraph{n1}{\node[gninlinable] (n1) at (0,0)
    {$1$}; \node[gninlinable] (n2) at (0.7,0) {$2$};} \ ;\ %
  \fcGraph{n1}{\node[gninlinable] (n1) at (0,0)
    {$1$}; \node[gninlinable] (n2) at (0.7,0)
    {$2$};} \rightarrow \fcGraph{n1}{\node[gninlinable] (n1) at (0,0)
    {$1$}; \node[gninlinable] (n2) at (0.7,0)
    {$2$};} \leftarrow \fcGraph{n1}{\node[gninlinable] (n1) at (0,0)
    {$1$};} $}).

The general case where this assumption does not hold will be
treated in \Cref{sec:general-case}.

\subsection{Tableau Calculus}

The underlying idea is fairly straightforward: we take an alternating
condition $\mathcal{A}$ and whenever it is existential, that is
$\mathcal{A} = \biglor_{i \in I} \exists f_i.\mathcal A_i$, we branch
and check whether some $\mathcal{A}_i$ is satisfiable. If instead it
is universal, i.e.,
$\mathcal{A} = \bigland_{i \in I} \forall f_i.\mathcal A_i$, we check
whether some $f_i$ is an iso. If that is not the case, clearly the
sequence of identities on $\RO(\mathcal{A})$ is a model, since there
is no arrow $g$ such that $\id = f_i;g$, assuming that all sections
are isos. If however some $f_i$ is an iso, we invoke a pull-forward
rule (see below for more details) that transforms the universal condition
into an existential condition and we continue from there. We will show
that this procedure works whenever the pull-forward follows a
\emph{fair} strategy: in particular every iso (respectively one of its
successors) must be pulled forward eventually.

The pull-forward rule relies on the equivalence
$(\mathcal{A} \land \exists f.\mathcal{B}) \equiv \exists
f.(\mathcal{A}_{\downarrow f}\land \mathcal{B})$
(cf.~Appendix~\ref{app:exists-and-shift}).

\begin{toappendix}
  \label{app:exists-and-shift}
  \begin{lemma}
    \label{lem:exists-and-shift}
    Let $\mathcal{A},\mathcal{B}$ be two conditions with root objects
    $A,B$ and let $f\colon A\to B$. Then
    \[ (\mathcal{A} \land \exists f.\mathcal{B}) \equiv \exists
      f.(\mathcal{A}_{\downarrow f}\land \mathcal{B}) \qquad
      (\mathcal{A} \lor \forall f.\mathcal{B}) \equiv \forall
      f.(\mathcal{A}_{\downarrow f}\lor \mathcal{B}) \]
  \end{lemma}

\begin{proof}
  For
  $(\mathcal A \land \exists f.\mathcal B) \equiv \exists f.(\mathcal
  A_{\downarrow f} \land \mathcal B)$: let $c$ be an arrow with source
  object $\RO(\mathcal{A})$. Then:
  \begin{align*}
    c \models (\mathcal A \land \exists f.\mathcal B) \iff& c \models
    \mathcal A \land c \models \exists f.\mathcal B
    \\
    \text{(\Cref{def:satisfaction})} \iff& c \models \mathcal A \land
    \exists \alpha (c = f \comp \alpha \:\land\: \alpha \models
    \mathcal B)
    \\
    \iff& \exists \alpha (c = f \comp \alpha \:\land\: c \models
    \mathcal A \:\land\: \alpha \models \mathcal B)
    \\
    \iff& \exists \alpha (c = f \comp \alpha \:\land\: f \comp \alpha
    \models \mathcal A \:\land\: \alpha \models \mathcal B)
    \\
    \text{(\Cref{prop:shift})} \iff& \exists \alpha (c = f \comp
    \alpha \:\land\: \alpha \models \mathcal A_{\downarrow f} \land
    \alpha \models \mathcal B)
    \\
    \iff& \exists \alpha (c = f \comp \alpha \:\land\: \alpha \models
    (\mathcal A_{\downarrow f} \land \mathcal B))
    \\
    \text{(\Cref{def:satisfaction})} \iff& c \models \exists
    f.(\mathcal A_{\downarrow f} \land \mathcal B)
  \end{align*}
  The dual equivalence can be derived from the first one by negation.
\end{proof}
\end{toappendix}

\begin{toappendix}
\begin{lemma}\label{all-iso-equiv-exists-iso}
  If $f$ is an isomorphism, then
  $\forall f . \mathcal A \equiv \exists f . \mathcal A$.
\end{lemma}

\begin{proof}
  \[\def\arraystretch{1.2}\setlength{\arraycolsep}{0pt}\begin{array}{rlrrl}% x | <=> Q alpha: if | ... = | ... | then ... models A
    &\phantom{\iff{}} \mathrlap{c \models \forall f . \mathcal A}
    \\
    \text{(def sat)}&\iff \forall \alpha\colon \text{if} & c =& f ; \alpha &\text{ then $\alpha \models \mathcal A$}
    \\
    \text{($f$ iso)}
    &\iff \forall \alpha\colon \text{if} & f^{-1} ; c =&\mkern\medmuskip f^{-1} ; f ; \alpha &\text{ then $\alpha \models \mathcal A$}
    \\
    &\iff \forall \alpha\colon \text{if} & f^{-1} ; c =& \alpha &\text{ then $\alpha \models \mathcal A$}
    \\
    &\iff \forall \alpha\colon \text{if} & f^{-1} ; c =& \alpha &\text{ then $f^{-1} ; c \models \mathcal A$}
    \\
    &\iff &&& \phantom{\text{ and }}\mkern\thinmuskip f^{-1} ; c \models \mathcal A
    \\
    \text{(existential introduction)}&\iff \exists \alpha\colon & f^{-1} ; c =& \alpha &\text{ and $f^{-1} ; c \models \mathcal A$} \\
    &\iff \exists \alpha\colon & f^{-1} ; c =& \alpha &\text{ and $\alpha \models \mathcal A$} \\
    \text{($f$ iso)}&\iff \exists \alpha\colon & f ; f^{-1} ; c =& f ; \alpha &\text{ and $\alpha \models \mathcal A$} \\
    &\iff \exists \alpha\colon & c =& f ; \alpha &\text{ and $\alpha \models \mathcal A$} \\
    \text{(def sat)}&\iff \mathrlap{c \models \exists f . \mathcal A}
  \end{array}\]

  % Correctness of $(*)$:
  % $\Rightarrow$ antecedent is true.
  % $\Leftarrow$ the universal statement is known to hold for $\alpha = f^{-1};c$, and the implication is trivially true for $\alpha \ne f^{-1};c$.
  % \qedhere
  % \oldtodo{L: This is a stronger variant of \Cref{lem:all-section-implies-exists-section} that even shows equivalence (but also requires that $f$ is an iso and not just a section).}
\end{proof}
\end{toappendix}

\begin{lemmarep}[Pulling forward isomorphisms]\label{lem:pull-forward-isos}
  Let
  \raisebox{0pt}[\height][\depth+2pt]{$\bigland_{i \in I} \forall
    f_i.\mathcal A_i$} be a universal condition and assume
  for some $p \in I$, $f_p$ is an iso and
  $\mathcal A_p = \biglor_{j \in J} \exists g_j.\mathcal B_j$. Then
  $f_p$ can be \emph{pulled forward}:
  \[
    \bigland_{i \in I} \forall f_i.\mathcal A_i
    \:\equiv\:
    \exists f_p.\biglor_{j \in J} \exists g_j.\Bigg(
        B_j \land
        \Big( \bigland_{m \in I \setminus \{p\}} \forall f_m.\mathcal
        A_m \Big)_{\downarrow f_p;g_j}
      \Bigg)
  \]
\end{lemmarep}

\begin{proof}
  \def\landMneP{\bigland_{m \in I \setminus \{p\}}}
  \def\allFmAm{\forall f_m.\mathcal A_m}
  \def\lorJinJ{\biglor_{j \in J}}
%  \oldtodo[color=sidenote]{\textbf{L:} Let's just pretend that
%  conjunctions of existentials, as used in several lines in the
%  proof, are actually possible. Even then, by the end of the proof,
%  the result actually is a proper condition again, so this should be
%  fine.}
  The following calculations rely on \Cref{prop:laws-conditions}.
  \[
  \def\MatchesInDefBelow#1#2{\color{black!12}\underbracket{\color{black}\textstyle #1}_{\begingroup\color{black!40}#2\endgroup}}
  \def\MatchesInDefAbove#1#2{\color{black!12}\overbracket{\color{black}\textstyle #1}^{\begingroup\color{black!40}#2\endgroup}}
  \renewcommand{\arraystretch}{1.6}%
  \setlength{\arraycolsep}{0.3pt}%
  \begin{array}{rlll}
      \left( \bigland_{i \in I} \forall f_i.\mathcal A_i \right)
    &=
      \forall f_p.\mathcal A_p
      &\land\;{}& \Big( \landMneP \allFmAm \Big)
    \\&=
      \Big(\forall f_p.\big(\lorJinJ \exists g_j.\mathcal B_j \big)\Big)
      \;{}&\land& \Big( \landMneP \allFmAm \Big)
    \\ \text{(\Cref{all-iso-equiv-exists-iso})} &\equiv
      \MatchesInDefBelow{ \Big(\exists f_p.\big(\lorJinJ \exists g_j.\mathcal B_j \big)\Big) }{\exists f.\mathcal B}
      &\land& \MatchesInDefBelow{ \Big( \landMneP \allFmAm \Big) }{\mathllap{\land\quad \mathcal A}}
    \\ \text{(\Cref{lem:exists-and-shift})} &\equiv
      \MatchesInDefAbove{ \exists f_p.\bigg(
        \lorJinJ \exists g_j.\mathcal B_j }{\exists f.(\mathcal B}
        &\land& \MatchesInDefAbove{ \Big(\landMneP (\allFmAm)\Big)_{\downarrow f_p}
      \bigg) }{\mathllap{\land\quad \mathcal A_{\downarrow f}})}
    \\ \text{(distributivity)} &\equiv
      \exists f_p.\lorJinJ \bigg(
        \exists g_j.\mathcal B_j
        &\land& \landMneP (\allFmAm)_{\downarrow f_p}
      \bigg)
    \\ \text{(\Cref{lem:exists-and-shift})} &\equiv
      \exists f_p.\lorJinJ \exists g_j . \Big(
        \mathcal B_j
        &\land& \landMneP (\allFmAm)_{\downarrow f_p \downarrow g_j}
      \Big)
    \\ &\equiv
      \exists f_p.\lorJinJ \exists g_j . \Big(
        \mathcal B_j
        &\land& \landMneP (\allFmAm)_{\downarrow f_p ; g_j}
      \Big)
    % \\ &\equiv
    %   \lorJinJ \exists f_p;g_j. \Big(
    %     \mathcal B_j
    %     &\land& \landMneP (\allFmAm)_{\downarrow f_p ; g_j}
    %   \Big)
    \end{array}\] Note that if
  $\left( \bigland_{i \in I} \forall f_i.\mathcal A_i \right)$ is
  alternating, then so is the condition after $\exists f_p$ in the
  last line above.
\end{proof}

% In the remainder of this section, we are working in categories where
% all sections are isomorphisms, i.e., for all arrows $f,g$, whenever
% $f;g = \id$, then $f$ is an iso.  The general case is discussed
% later in \Cref{sec:general-case}.

\begin{definition}[SatCheck tableau construction rules]\label{satcheck.new}
  Given an alternating condition $\mathcal{A}$, we give rules for the
  construction of a tableau for $\mathcal{A}$ that has conditions as
  nodes, $\mathcal A$ as root node, and edges ($\to$) labeled with arrows.
  The tableau is extended at its leaf nodes as
  follows:
  \[\begin{array}{ll|ll}
    \text{For every $p \in I$:}
    &&&
    \text{For \emph{one} $p \in I$ such that $f_p$ is iso and
      $\mathcal A_p = \biglor\nolimits_{j \in J} \exists g_j . \mathcal B_j$:
    }
    \\
    \displaystyle\biglor_{i \in I} \exists f_i . \mathcal A_i
    \:\xrightarrow{f_p}\:
    \mathcal A_p
    &&&
    % increases the depth of #2 by length #1 below its natural size.
    % used below to make the \underbrace not stick as close to the formula as it normally would
    \def\addExtraDepthTo#1#2{\raisebox{0pt}[\height][\depth+#1]{#2}}
    \displaystyle\bigland_{i \in I} \forall f_i . \mathcal A_i
    \:\xrightarrow{f_p}\:
    \underbrace{%
    \biglor_{j \in J} \exists g_j . \addExtraDepthTo{2pt}{\bigg(}
      \mathcal B_j
    }_{\mathcal A_p}
      \land
    \underbrace{%
      \Big( \bigland_{\mkern-8mu m \in I \setminus \{p\} \mkern-12mu} \forall f_m . \mathcal A_m \Big)_{\downarrow f_p ; g_j}
    \vphantom{\addExtraDepthTo{2pt}{\bigg)}}
    }_{\text{\clap{other children, shifted to include $g_j$}}}
    \mkern-4mu
    \bigg)
  \end{array}\]
  \begin{itemize}
  \item For existential conditions, for \emph{each}(!) child condition
    $\exists f_p.\mathcal A_p$, add a new descendant.%, i.e.\ for every
%    $p \in I$:
%  \[
%    \biglor_{i \in I} \exists f_i . \mathcal A_i
%    \:\xrightarrow{f_p}\:
%    \mathcal A_p
%  \]
  \item For universal conditions, non-deterministically pick
    \emph{one}(!)  child condition $\forall f_p.\mathcal A_p$ that can
    be pulled forward ($f_p$ is an iso), pull it forward (cf.\ %
    \Cref{lem:pull-forward-isos}), and add the result as its (only)
    descendant. If a universal condition contains no isos, then add no
    descendant.%, i.e.\ for \emph{one} $p \in I$ such that $f_p$ is iso and
%  $\mathcal A_p = \biglor\nolimits_{j \in J} \exists g_j . \mathcal
%  B_j$:
%  \[
%    % increases the depth of #2 by length #1 below its natural size.
%    % used below to make the \underbrace not stick as close to the formula as it normally would
%    \def\addExtraDepthTo#1#2{\raisebox{0pt}[\height][\depth+#1]{#2}}
%    %
%    \bigland_{i \in I} \forall f_i . \mathcal A_i
%    \:\xrightarrow{f_p}\:
%    \underbrace{%
%    \biglor_{j \in J} \exists g_j . \addExtraDepthTo{2pt}{\biggg(}
%      \mathcal B_j
%    }_{\mathcal A_p}
%      \land
%    \underbrace{%
%      \Big( \bigland_{\mkern-8mu m \in I \setminus \{p\} \mkern-12mu} \forall f_m . \mathcal A_m \Big)_{\downarrow f_p ; g_j}
%    \vphantom{\addExtraDepthTo{2pt}{\biggg)}}
%    }_{\text{\clap{other children, shifted to include $g_j$}}}
%    \mkern-4mu
%    \biggg)
%  \]
  \end{itemize}
  % Note that a node with an existential condition typically has
  % multiple descendants, but a universal condition has only at most one
  % descendant (non-deterministically choose one of the isos, if there
  % are any).  

  %With each pair of $\forall$-condition and one of its $\forall$-grandchildren, we associate a \emph{successor relation} that relates top-level child conditions to their shifted counterparts, as detailed in the following \Cref{lts-asai-successor-relation}.
  A \emph{branch} of a tableau is a (potentially infinite) path
  $\mathcal A_0 \xrightarrow{u_1} \mathcal{A}_1 \xrightarrow{e_1}
  \mathcal{A}_2 \xrightarrow{u_2} \dots$ starting from the root node.
  A finite branch is \emph{extendable} if one of the tableau
  construction rules is applicable at its leaf node and would result
  in new nodes (hence, a branch where the leaf is an empty existential
  or a universal without isomorphisms is not extendable).  A branch is
  \emph{closed} if it ends with an empty existential, otherwise it is
  \emph{open}.
\end{definition}

Due to \Cref{lem:pull-forward-isos} each universal condition is (up to
iso $f_p$) equivalent to its (unique) descendant (if one exists),
while an existential condition is equivalent to the disjunction of its
descendants prefixed with existential quantifiers.

The labels along one branch of the tableau are arrows between the root
objects of the conditions. Their composition corresponds to the prefix
of a potential model being constructed step by step.
% , which is evident
% especially for existential conditions.  If a descendant condition is
% satisfiable, then its parent is also satisfiable.
% Not every branch describes a model, though: an extendable branch (such
% as
% $\exists f.\forall \id.\condfalse \xrightarrow{f} \forall
% \id.\condfalse$) suggests that a model (starting with $f$) could
% exist, but there may or may not be such a model (in this case, it can
% be extended only to
% $\forall \id.\condfalse \xrightarrow{\raisebox{-1.5pt}[0pt][0pt]{$\scriptstyle\smash\id$}} \condfalse$, which is
% unsatisfiable and closes the branch).
%
Finite paths represent a model if they are open and not extendable.
For infinite paths, we need an additional property to make sure that
the procedure does not ``avoid'' a possibility to show
unsatisfiability of a condition.
%
% As a concrete example, let $\mathcal U$ be unsatisfiable and
% $\mathcal I$ be satisfiable with an infinite model, but not with a
% finite model.  Let
% $\mathcal A = \forall \id. \mathcal I \land \forall \id. \mathcal U$.
% Since $\mathcal I$ has an infinite model, it is possible to first pull
% forward $\forall \id. \mathcal I$, then constructing the infinite path
% that, in the limit, describes the infinite model of $\mathcal I$.
% The infinite path constructed by repeatedly
% pulling forward parts of $\mathcal I$ suggests that it is an infinite
% model for $\mathcal A$.  However, $\mathcal A$ is actually
% unsatisfiable, which could have been shown by pulling forward
% $\forall \id. \mathcal U$, followed by the steps to show the
% unsatisfiability of $\mathcal U$.

To capture that, we introduce the notion of fairness, meaning that all
parts of a condition are eventually used in a proof and are not postponed indefinitely (a related concept
is saturation, see e.g.\ %
\cite{lo:tableau-graph-properties}, though the definition deviates due to a different setup).  For this we first need to track how pulling forward
one child condition changes the other children by shifting.  We define
a \emph{successor relation} that, for each pair of $\forall$-condition
and one of its $\forall$-grandchildren, relates child conditions of
the $\forall$-condition to their shifted counterparts (the \emph{successors}) in the
$\forall$-condition of the second-next nesting level. The successor
relation is similar in spirit to the one used
in~\cite{lo:tableau-graph-properties}.
In this work, saturation is given in a more descriptive way and has to account for nesting levels in the tableau, a complication that we were able to avoid in the present paper.

\begin{definition}[Successor relation]
  \label{lts-asai-successor-relation}
  Assume in the construction of a tableau we have a path
  \[ \bigland_{i\in I} \forall f_i.\mathcal A_i
    \:\xrightarrow{f_p}\:
    \mathcal{C}
    \:\xrightarrow{g_j}\:
    \mathcal B_j \land \Big( \bigland_{\mkern-8mu m\in I\setminus\{p\} \mkern-8mu}
    \forall f_m.\mathcal A_m \Big)_{\downarrow f_p;g_j}
    =
    \mathcal B_j \land
    \bigland_{m\in I\setminus\{p\}}
    %\bigland_{\substack{(\alpha,\beta) \in\\\kappa(f_m,\; f_p;g_j)}}
    \bigland_{(\alpha,\beta) \in \kappa(f_m,\; f_p;g_j )} \forall
    \beta.({\mathcal A_m})_{\downarrow \alpha} \] where $\mathcal{C}$
  is the existential condition given in \Cref{satcheck.new}.  Then
  for each $m \in I \setminus \{p\}$,
  each $\forall \beta.({\mathcal A_m})_{\downarrow \alpha}$ where
  $(\alpha,\beta) \in \kappa(f_m,\; f_p;g_j)$ is a \emph{successor} of
  $\forall f_m.\mathcal A_m$. The transitive closure of the successor
  relation induces the \emph{indirect successor relation}.

  % \noindent\mbox{}\hfill\begin{tikzpicture}[x=1cm,y=1cm]
  %   % %%% LEFT - the forward direction %%%
  %   \begin{scope}[shift={(0,0)}]
  %     \node (a) at (0,-0.2) {};
  %     \node (b) at (2,-0.2) {};
  %     \node (c) at (0,-1.7) {};
  %     \node (ds) at (2,-1.7) {};
  %     \draw[->] (a) -- node[above]{$f_p;g_j$} (b);
  %     \draw[->] (a) -- node[left]{$f_m$} (c);

  %     \draw[->] (b) -- node[right]{$\beta$} (ds);
  %     \draw[->] (c) -- node[below]{$\alpha$} (ds);

  %     \node[condtri,shape border rotate=90, rotate around={-10:(c.center)}, label={[rotate=-10,anchor=north,label distance=1pt]below:{$\mathcal A_m$}}] at (c.south) {\kern7pt};
  %     \node[condtri,shape border rotate=90, rotate around={-10:(ds.center)}, label={[rotate=-10,anchor=north,label distance=1pt]below:{${\mathcal A_m}_{\downarrow \alpha}$}}] at (ds.south) {\kern7pt};
  %     \node[condtri,shape border rotate=270, rotate around={10:(a.center)}, label={[rotate=10,anchor=south,label distance=1pt]above:{$\forall f_m.\mathcal A_m$}}] at (a.north) {\kern7pt};
  %     \node[condtri,shape border rotate=270, rotate around={10:(b.center)}, label={[rotate=10,anchor=south,label distance=1pt]above:{$\forall \beta.({\mathcal A_m}_{\downarrow \alpha})$}}] at (b.north) {\kern7pt};
  %   \end{scope}
  % \end{tikzpicture}\hfill\mbox{}

  % The figure depicts the situation for a single $(\alpha,\beta) \in \kappa(f_m,\ f_p;g_j)$.
  % In each such case, the successor relation relates $\forall f_m.\mathcal A_m$ to $\forall \beta.({\mathcal A_m}_{\downarrow \alpha})$.
\end{definition}

% In our case, the successor relation is not transitive, and it relates
% only children of the conditions of one particular step, but not
% children of children.

\begin{definition}[Fairness]
  An infinite branch of a tableau is \emph{fair} if for each universal
  condition $\mathcal A$ on the branch and each child condition
  $\forall f_i.\mathcal A_i$ of $\mathcal A$ where $f_i$ is an iso, it
  holds that some indirect successor of $\forall f_i.\mathcal A_i$ is
  eventually pulled forward.
%\oldtodo{\textbf{Le:} Related to this: Is it possible to define
%  fairness on the level of a single branch?  For tableau-based methods
%  it is more common, I think, to do this on the tableau level, or on
%  the set of branches (not on single ones) generated for the tableau.
%  (see saturation Def. 14 in \cite{lo:tableau-graph-properties}, in
%  particular second item)}
\end{definition}

\begin{remark}[Fairness strategies]
  \label{rem:fairness-strategies}
  One possible strategy that ensures fairness is to maintain for each
  incomplete branch a queue of child conditions for which a successor
  must be pulled forward. Then the first entry in this queue is
  processed. Note that by the assumption on $\kappa$ made earlier at the end of \Cref{sec:representative-squares-shift}, each iso in a universal condition that is not
  pulled forward has exactly one successor and the queue is modified
  by replacing each condition accordingly and adding newly generated
  child conditions with isos at the end. %This strategy guarantees
  % fairness since each condition will eventually be
  % processed. %Then, the number
  % of steps until a successor of a child is pulled forward is bounded
  % by the current size of the queue plus the breadth of the condition.
  %
  % For the example condition
  % $\mathcal A = \forall \id. \mathcal I \land \forall \id. \mathcal U$
  % from above, where $\mathcal U$ is unsatisfiable and $\mathcal I$ is
  % satisfiable with an infinite model, it is possible to pull forward
  % $\forall \id. \mathcal I$ first, but in the following step, before
  % further descending into $\mathcal I$, a successor of
  % $\forall \id. \mathcal U$ must be pulled forward first, thus
  % ensuring progress towards showing unsatisfiability.
\end{remark}

\subsection{Up-To Techniques, Fair Branches and Models}

While showing soundness of the tableau method is relatively
straightforward, the crucial part of the completeness proof is to show
that every infinite and fair branch of the tableau corresponds to a
model. The proof strategy is the following: given such a branch, we
aim to construct a witness for this model, by pairing conditions on
this path with the suffix consisting of the sequence of arrows
starting from this condition. If one could show that the set
$P \subseteq \Seq \times \Cond$ of pairs so obtained is a
post-fixpoint of the satisfaction function $s$ defined in
\Cref{prop:satisfaction-function} ($P\subseteq s(P)$), we could
conclude, as the satisfaction relation $\models$ is the greatest
fixpoint of $s$ (\Cref{prop:satisfaction-function}) and hence above
any post-fixpoint.

However, $P$ is in general not a post-fixpoint, which is mainly due to
the fact that universal conditions are treated ``sequentially'' one
after another and are ``pulled forward'' only if they become
isos. Hence, if we want to show that for a chain $[a_1,a_2,\tdots]$
the universal condition of the form
$\bigland_i \forall f_i.\mathcal{A}_i$ is satisfied, we have to prove
for every child $\forall f_i.\mathcal{A}_i$ that whenever
$a_1;\tdots;a_n = f_i;g$ it holds that
$[g,a_{n+1},\tdots]\models \mathcal{A}_i$. If
$\forall f_i.\mathcal A_i$ actually is the child that is pulled
forward in the next tableau step, $P$ contains the tuple required by
$s$.  If not, there is a ``delay'' and intuitively that means that we
can not guarantee that $P$ is indeed a post-fixpoint.

However it turns out that it is a post-fixpoint up-to
($P\subseteq s(u(P))$), where $u$ is a combination of one or more
suitable up-to functions.
We first explore several such up-to functions:
$\uptoConj(P)$ obtains new conditions by non-deterministically removing parts of conjunctions;
with $\uptoRecomp(P)$ we can arbitrarily recompose (decompose and compose) the arrows in a potential model;
with $\uptoShift(P)$ we can undo a shift;
and $\uptoIso(P)$ allows replacing conditions with isomorphic conditions.
For each, we show their $s$-compatibility (i.e., $u(s(P)) \subseteq s(u(P))$).

\begin{toappendix}
\begin{lemma}[Up-to conjunction removal]
  \label{thm:upto-conjunction-removal}
  Let $P \subseteq \Seq \times \Cond$, i.e.\ a set of tuples of
  potential model and condition. We inductively define a relation
  $\uptoConjCond$ containing a pair of conditions
  $(\mathcal A, \mathcal T)$ iff $\mathcal T$ is the same as
  $\mathcal A$ but with some conjunctions removed. That is,
  $\uptoConjCond$ contains
  % \uptoConjCond processes just conditions and \uptoConj processes pairs of model+condition
  \[
  \renewcommand{\arraystretch}{1.6}%
  \setlength{\arraycolsep}{0.3pt}%
  \begin{array}{rccll}
    \displaystyle \bigland_{i \in I} \forall f_i.\mathcal A_i
    &\mkern18mu \mathrel{\uptoConjCond} \mkern12mu
    &\displaystyle \bigland_{j \in J \subseteq I}
    &\forall f_j.\mathcal T_j
    &\text{\quad whenever $(\mathcal A_j, \mathcal T_j) \in
      \uptoConjCond$ for all $j\in J$}
    \\
    \displaystyle \biglor_{i \in I} \exists f_i.\mathcal A_i
    &\mkern18mu \mathrel{\uptoConjCond} \mkern12mu
    &\displaystyle \biglor_{i \in I}
    &\exists f_i.\mathcal T_i
    &\text{\quad whenever $(\mathcal A_i, \mathcal T_i) \in
      \uptoConjCond$ for all $i\in I$}
  \end{array}\]
  Also, define
  \mbox{$\uptoConj(P) = \{ (\bar c, \mathcal T) \mid (\bar c, \mathcal A) \in P,\ (\mathcal A, \mathcal T) \in \uptoConjCond \}$}.
  % It is easy to see that $\uptoConj$ is a closure.%
  %
  The function $\uptoConj$ is $s$-compatible.
  %Then $\uptoConj(s(P)) \subseteq s(\uptoConj(P))$.
  % i.e., $\uptoConj$ is $s$-compatible.
\end{lemma}

\begin{proof}
  % Monotonicity of $\uptoConj$: trivial.
  % Extensiveness by always setting $S' = S$.
  % Idempotence:
  % $\uptoConj \subseteq \uptoConj \circ \uptoConj$ by extensiveness and monotonicity,
  % $\uptoConj \circ \uptoConj \subseteq \uptoConj$ because all subtree removals can be combined into a single application of $\uptoConj$ (only conditions are changed and the only changes are removal of subtrees).}
  %
  % T wie Teilbedingung, A wie Alles
  We show that $(\bar c, \mathcal T) \in \uptoConj(s(P))$ implies
  $(\bar c, \mathcal T) \in s(\uptoConj(P))$.

  Let $(\bar c, \mathcal T) \in \uptoConj(s(P))$.
  $\mathcal T$ is obtained from some $\mathcal A$ by forgetting elements of a conjunction,
  i.e., $(\mathcal A, \mathcal T) \in \uptoConjCond$,
  and $(\bar c, \mathcal A) \in s(P)$.
  There are two cases:

  \begin{itemize}
    \proofPart{$\mathcal T = \biglor_{i \in I} \exists f_i.\mathcal T_i$, i.e.,
      existential}
%    $(\bar c, \mathcal T) \in \uptoConj(s(P))$ means that $(\bar c,
%    \mathcal A) \in s(P)$ (where $(\bar c, \mathcal A)$ is
%    constructed as described above).
    In this case we know that
    $\mathcal{A} = \biglor_{i \in I} \exists f_i.\mathcal A_i$ and
    $(\mathcal{A}_i,\mathcal{T}_i)\in \uptoConjCond$ for all $i\in
    I$.

    Since $(\bar c, \mathcal A) \in s(P)$, there exists some $i\in I$
    and some $g,n$ such that $c_1;\dots;c_n = f_i ; g$ and
    $([g, c_{n+1}, \dots], \mathcal A_i) \in P$.  We know that
    $(\mathcal{A}_i,\mathcal{T}_i)\in \uptoConjCond$, which implies
    that $([g, c_{n+1}, \dots], \mathcal T_i) \in \uptoConj(P)$.

    This however is exactly what is required for
    $(\bar c, \mathcal T) \in s(\uptoConj(P))$ to hold, according to
    the definition of $s$.

  \proofPart{$\mathcal T = \bigland_{j \in J} \forall f_j.\mathcal T_j$, i.e., universal}
    % $(\bar c, \mathcal T) \in \uptoConj(s(P))$ means that $(\bar c,
    % \mathcal A) \in s(P)$ (where $(\bar c, \mathcal A)$ is
    % constructed as described above).
  In this case we know that
  $\mathcal{A} = \bigland_{j \in J\subseteq I} \forall f_i.\mathcal
  A_i$ and $(\mathcal{A}_i,\mathcal{T}_i)\in \uptoConjCond$ for all
  $i\in J$.

  Since $(\bar c, \mathcal A) \in s(P)$, for all $i\in I$ and all
  $g,n$, it holds that if $c_1;\dots;c_n = f_i ; g$, then
  $([g, c_{n+1}, \dots], \mathcal A_i) \in P$.
  In these situations, from
  $([g, c_{n+1}, \dots], \mathcal A_i) \in P$ and
  $(\mathcal{A}_i,\mathcal{T}_i)$ (for $i\in J \subseteq I$), it
  follows that $([g, c_{n+1}, \dots], \mathcal T_i) \in \uptoConj(P)$
  for all $i\in J$.

  This however is exactly what is required for
  $(\bar c, \mathcal T) \in s(\uptoConj(P))$ to hold, according to the
  definition of $s$.  \qedhere
  \end{itemize}
\end{proof}

\begin{lemma}[Up-to recomposition]
  \label{thm:compat-uptorecomp}
  Let $P \subseteq \Seq \times \Cond$.
  Then, $\uptoRecomp(P)$ allows splitting and merging parts of the \mbox{sequence} of arrows, i.e.,
  $\uptoRecomp(P) = \{
  ([b_1, \dots, b_\ell, \bar c ], \mathcal A) \mid
  ([a_1, \dots, a_k, \bar c], \mathcal A) \linebreak\in P,\ %
  a_1;\dots;a_k = b_1;\dots;b_\ell
  \}$.
  % It is easy to see that $\uptoRecomp$ is a closure.
  %
  The function $\uptoRecomp$ is $s$-compatible.
  %It holds that $\uptoRecomp(s(P)) \subseteq s(\uptoRecomp(P))$.
  % i.e., $\uptoRecomp$ is \hyperref[prop:satisfaction-function]{$s$}-compatible.
\end{lemma}

\begin{proof}
  % Monotonicity of $\uptoRecomp$: trivial.
  % Extensiveness with $k=\ell=1$.
  % Idempotence:
  % $\uptoRecomp \subseteq \uptoRecomp \circ \uptoRecomp$ by extensiveness and monotonicity,
  % $\uptoRecomp \circ \uptoRecomp \subseteq \uptoRecomp$ because any element produced by subsequent recomposition of $k_1$ initial arrows to $\ell_1$ ones, then further transforming $k_2$ initial arrows of the result to $\ell_2$ ones can also be done in a single transformation step.}
  % % the exact values of the resulting k,l depend on whether l1>k2 or not
  %
  We show that
  $([ b_1, \dots, b_\ell, \bar c], \mathcal A) \in \uptoRecomp(s(P))$ implies
  $([ b_1, \dots, b_\ell, \bar c], \mathcal A) \in s(\uptoRecomp(P))$.

  Since $([ b_1, \dots, b_\ell, \bar c], \mathcal A) \in \uptoRecomp(s(P))$,
  there must be some $([ a_1, \dots, a_k, \bar c], \mathcal A) \in s(P)$
  such that $a_1 ; \dots ; a_k = b_1 ; \dots ; b_\ell$.
  Furthermore, let $\bar c = [ c_1, c_2, \dots ]$.

  To show that $([ b_1, \dots, b_\ell, \bar c], \mathcal A) \in s(\uptoRecomp(P))$,
  we consider the following cases:

  \begin{itemize}
    \proofPart{$\mathcal A$ is existential ($\exists$)} Since
    $([ a_1, \dots, a_k, \bar c], \mathcal A) \in s(P)$ $(*)$, there
    exists some child $\exists h_i.\mathcal A_i$ of $\mathcal A$, some
    $g$ and $n$ that satisfy the conditions of $s$.  We consider the
    following subcases, depending on the value of $n$ given:

    \begin{itemize}
    \proofPart{$n \geq k$, i.e., $h_i ; g$ ``points to somewhere'' in $\bar c$}
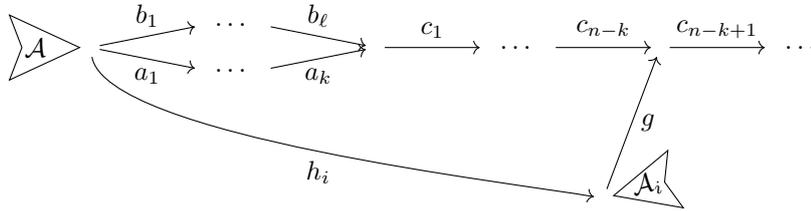
      \begin{tikzpicture}
        \def\clen{1.5}% length of a c_i arrow
        \def\dlen{0.75}% length of a "..." segment
        \def\bypos{0.3}% y-position of the b_i arrows
        \def\aypos{-0.3}% y-position of the a_i arrows
        \node (ab0) at (0,0) {};
        \node (b1) at (1*\clen,\bypos) {};
        \node (b2) at (1*\clen+1*\dlen,\bypos) {};
        \node (a1) at (1*\clen,\aypos) {};
        \node (a2) at (1*\clen+1*\dlen,\aypos) {};
        \node (c0) at (2*\clen+1*\dlen,0) {};
        \node (c1) at (3*\clen+1*\dlen,0) {};
        \node (c2) at (3*\clen+2*\dlen,0) {};
        \node (c3) at (4*\clen+2*\dlen,0) {};
        \node (c4) at (5*\clen+2*\dlen,0) {};
        \node (c5) at (5*\clen+3*\dlen,0) {};

        \draw[->] (ab0) -- node[above]{$b_1$} (b1);
        \draw[->] (b2) -- node[above]{$b_\ell$} (c0);
        \draw[->] (ab0) -- node[below]{$a_1$} (a1);
        \draw[->] (a2) -- node[below]{$a_k$} (c0);
        \node at ($(b1)!0.5!(b2)$) {$\dots$};
        \node at ($(a1)!0.5!(a2)$) {$\dots$};

        \draw[->] (c0) -- node[above]{$c_1$} (c1);
        \node at ($(c1)!0.5!(c2)$) {$\dots$};
        \draw[->] (c2) -- node[above]{$c_{n-k}$} (c3);
        \draw[->] (c3) -- node[above]{$c_{n-k+1}$} (c4);
        \node at ($(c4)!0.5!(c5)$) {$\dots$};

        \node (fitarg) at (3.5*\clen+2*\dlen,-2) {};

        \draw[->] (ab0) to[out=-80, in=170, looseness=0.4] node[below]{$h_i$} (fitarg);
        \draw[->] (fitarg) to node[right]{$g$} (c3);
        % not here yet: \draw[->,dashed] (fitarg) to[bend right=25] node[below]{$g'$} (c7);

        % not here yet: \draw[->] (c1) to[bend left=45,looseness=0.3] node[above]{\footnotesize$c_1;\dots;c_m,\ m > n$} (c7);

        %\node[condtri,shape border rotate=270] at (ab0.north) {$\mathcal A$};
        %\node[condtri,shape border rotate=90] at (fitarg.south) {$\mathcal A_i$\kern-0.5mm};
        \node[condtri,shape border rotate=0] at (ab0.west) {$\mathcal A$};
        \node[condtri,shape border rotate=180,rotate around={15:(fitarg.center)}] at (fitarg.east) {\raisebox{-1pt}[2pt][0pt]{$\mathcal A_i$}};
      \end{tikzpicture}

      Then we set $g = g'$ and $n' = n + \ell - k$ is the length of
      the prefix of $[b_1,\dots,b_\ell,\bar{c}]$ that is factored. It
      holds that:
 
      \begin{enumerate}
      \item
        $h_i ; g' = h_i ; g \stackrel{(1)}{=} a_1 ; \dots ; a_k ; c_1
        ; \dots ; c_{n-k} \stackrel{(2)}{=} b_1 ; \dots ; b_\ell ; c_1
        ; \dots ; c_{n'-\ell} $, where $(1)$ follows from $(*)$, and
        $(2)$ by substituting $a_1;\dots;a_k = b_1;\dots;b_\ell$.

      \item $([g, c_{n-k+1}, \dots], \mathcal A_i) \in P$ (with $(*)$)
        implies
        $([g', c_{n'-\ell+1}, \dots], \mathcal A_i) \in P \subseteq
        \uptoRecomp(P)$. This yields
        $([ b_1, \dots, b_\ell, \bar c], \mathcal A) \in
        s(\uptoRecomp(P))$, as desired.
      \end{enumerate}

    \proofPart{$n < k$, i.e., $h_i ; g$ ``points to somewhere'' in $a_1 \dots a_k$}
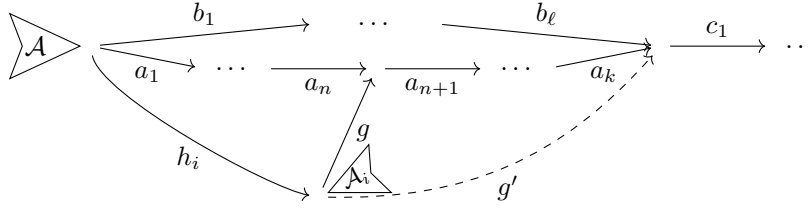
      \begin{tikzpicture}
        \def\clen{1.5}% length of a c_i arrow
        \def\dlen{0.75}% length of a "..." segment
        \def\bypos{0.3}% y-position of the b_i arrows
        \def\aypos{-0.3}% y-position of the a_i arrows
        \node (ab0) at (0,0) {};
        \node (b1) at (1.5*\clen+1*\dlen,\bypos) {};
        \node (b2) at (2.5*\clen+1*\dlen,\bypos) {};
        \node (a1) at (1*\clen,\aypos) {};
        \node (a2) at (1*\clen+1*\dlen,\aypos) {};
        \node (a3) at (2*\clen+1*\dlen,\aypos) {};
        \node (a4) at (3*\clen+1*\dlen,\aypos) {};
        \node (a5) at (3*\clen+2*\dlen,\aypos) {};
        \node (c0) at (4*\clen+2*\dlen,0) {};
        \node (c1) at (5*\clen+2*\dlen,0) {};
        \node (c2) at (5*\clen+3*\dlen,0) {};

        \draw[->] (ab0) -- node[above]{$b_1$} (b1);
        \draw[->] (b2) -- node[above]{$b_\ell$} (c0);
        \draw[->] (ab0) -- node[below]{$a_1$} (a1);
        \draw[->] (a2) -- node[below]{$a_n$} (a3);
        \draw[->] (a3) -- node[below]{$a_{n+1}$} (a4);
        \draw[->] (a5) -- node[below]{$a_k$} (c0);
        \node at ($(b1)!0.5!(b2)$) {$\dots$};
        \node at ($(a1)!0.5!(a2)$) {$\dots$};
        \node at ($(a4)!0.5!(a5)$) {$\dots$};

        \draw[->] (c0) -- node[above]{$c_1$} (c1);
        \node at ($(c1)!0.5!(c2)$) {$\dots$};

        \node (fitarg) at (1.5*\clen+1*\dlen,-2) {};

        \draw[->] (ab0) to[out=-80, in=170, looseness=0.4] node[below]{$h_i$} (fitarg);
        \draw[->] (fitarg) to node[right]{$g$} (a3);
        \draw[->,dashed] (fitarg) to[bend right=25] node[below]{$g'$} (c0);

        % \node[condtri,shape border rotate=270] at (ab0.north) {$\mathcal A$};
        % \node[condtri,shape border rotate=90] at (fitarg.south) {$\mathcal A_i$\kern-0.5mm};
        \node[condtri,shape border rotate=0] at (ab0.west) {$\mathcal A$};
        \node[condtri,shape border rotate=180,rotate around={25:(fitarg.center)},scale=0.9] at (fitarg.east) {\raisebox{-1pt}[1.8pt][0pt]{$\mathcal A_i$}};
      \end{tikzpicture}

      We do not generally have an exact matching $b_j$ to point our
      $g'$ to.  Instead, we choose $g' = g ; a_{n+1}; \dots ; a_k$ and
      $n' = \ell$ as the length of the prefix that is factored.  It
      holds that:

      \begin{enumerate}
      \item
        $ h_i ; g' = h_i ; g ; a_{n+1} ; \dots ; a_k = a_1 ; \dots ;
        a_n ; a_{n+1} ; \dots ; a_k = b_1 ; \dots ; b_\ell = b_1 ;
        \dots ; b_{n'} $
      \item $([g, a_{n+1}, \dots, a_k, \bar c], \mathcal A_i) \in P$
        (with $(*)$) implies (by recomposing the first $k-n+1$ arrows)
        $([g ; a_{n+1} ; \dots ; a_k, \bar c], \mathcal A_i) = ([g',
        \bar c], \mathcal A_i) \in \uptoRecomp(P)$. This yields
        $([ b_1, \dots, b_\ell, \bar c], \mathcal A) \in
        s(\uptoRecomp(P))$.
      \end{enumerate}
    \end{itemize}

    \proofPart{$\mathcal A$ is universal ($\forall$)} Since
    $([ a_1, \dots, a_k, \bar c], \mathcal A) \in s(P)$ $(*)$, all
    children $\forall h_i.\mathcal A_i$ of $\mathcal A$, all $g$ and
    $n$ satisfy the conditions of $s$.  We need to show the same for
    all children $\forall h_i.\mathcal A_i$, all $g'$ and all $n'$
    (where $n'$ is again the length of the prefix being factored) to
    establish
    $([ b_1, \dots, b_\ell, \bar c], \mathcal A) \in
    s(\uptoRecomp(P))$.  We consider the following subcases,
    depending on the value of $n'$:

    \begin{itemize}
      \proofPart{$n' \geq \ell$ and
        $b_1 ; \dots ; b_\ell ; c_1 ; \dots ; c_{n'-\ell} = h_i ; g'$}
      Since $b_1 ; \dots ; b_\ell = a_1 ; \dots ; a_k$, we have that
      $a_1 ; \dots ; a_k ; c_1 ; \dots ; c_{n'-\ell} = h_i ; g'$,
      which (with $(*)$) implies that
      $([g', c_{n-k+1}, \dots], \mathcal A_i) \in P \subseteq
      \uptoRecomp(P)$.

      % Have a look at $(*)$ for $g = g'$ and $n = n' - \ell + k$ ($n$
      % is adjusted for $\ell \ne k$, but points to the same $c_j$).
      % Since
      % $b_1 ; \dots ; b_\ell ; c_1 ; \dots ; c_{n'-\ell} = h_i ; g'$
      % (if not, the implication is trivially satisfied) we also have
      % $a_1 ; \dots a_k ; c_1 ; \dots ; c_{n-k} = h_i ; g$.  This means
      % that by $(*)$, we have
      % $([g, c_{n-k+1}, \dots], \mathcal A_i) \in P$.  By
      % extensiveness\footnote{$\uptoRecomp$ is easily seen to be
      %   extensive: use $k=\ell=1$ in its definition.} and rewriting,
      % this implies the required
      % $([g', c_{n'-\ell+1}, \dots], \mathcal A_i) \in \uptoRecomp(P)$.

    \proofPart{$n' < \ell$, and $b_1 ; \dots ; b_{n'} = h_i ; g'$}
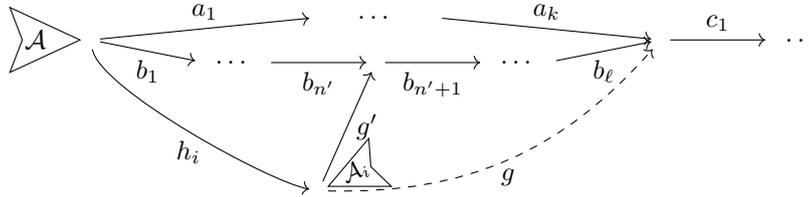
      \begin{tikzpicture}
        \def\clen{1.5}% length of a c_i arrow
        \def\dlen{0.75}% length of a "..." segment
        \def\aypos{0.3}% y-position of the a_i arrows
        \def\bypos{-0.3}% y-position of the b_i arrows
        \node (ab0) at (0,0) {};
        \node (a1) at (1.5*\clen+1*\dlen,\aypos) {};
        \node (a2) at (2.5*\clen+1*\dlen,\aypos) {};
        \node (b1) at (1*\clen,\bypos) {};
        \node (b2) at (1*\clen+1*\dlen,\bypos) {};
        \node (b3) at (2*\clen+1*\dlen,\bypos) {};
        \node (b4) at (3*\clen+1*\dlen,\bypos) {};
        \node (b5) at (3*\clen+2*\dlen,\bypos) {};
        \node (c0) at (4*\clen+2*\dlen,0) {};
        \node (c1) at (5*\clen+2*\dlen,0) {};
        \node (c2) at (5*\clen+3*\dlen,0) {};

        \draw[->] (ab0) -- node[above]{$a_1$} (a1);
        \draw[->] (a2) -- node[above]{$a_k$} (c0);
        \draw[->] (ab0) -- node[below]{$b_1$} (b1);
        \draw[->] (b2) -- node[below]{$b_{n'}$} (b3);
        \draw[->] (b3) -- node[below]{$b_{n'+1}$} (b4);
        \draw[->] (b5) -- node[below]{$b_\ell$} (c0);
        \node at ($(a1)!0.5!(a2)$) {$\dots$};
        \node at ($(b1)!0.5!(b2)$) {$\dots$};
        \node at ($(b4)!0.5!(b5)$) {$\dots$};

        \draw[->] (c0) -- node[above]{$c_1$} (c1);
        \node at ($(c1)!0.5!(c2)$) {$\dots$};

        \node (fitarg) at (1.5*\clen+1*\dlen,-2) {};

        \draw[->] (ab0) to[out=-80, in=170, looseness=0.4] node[below]{$h_i$} (fitarg);
        \draw[->] (fitarg) to node[right]{$g'$} (a3);
        \draw[->,dashed] (fitarg) to[bend right=25] node[below]{$g$} (c0);

        \node[condtri,shape border rotate=0] at (ab0.west)
        {$\mathcal A$}; \node[condtri,shape border rotate=180,rotate
        around={25:(fitarg.center)},scale=0.9] at (fitarg.east)
        {\raisebox{-1pt}[1.8pt][0pt]{$\mathcal A_i$}};
      \end{tikzpicture}

      This yields that
      $b_1 ; \dots ; b_\ell = h_i ; g'; b_{n'+1}; \dots ; b_\ell$,
      and with $(*)$ we can infer that
      $([g'; b_{n'+1}; \dots ; b_\ell, \bar{c}], \mathcal A_i) \in P$.
      Hence
      $([g', b_{n'+1}, \dots , b_\ell, \bar{c}], \mathcal A_i) \in
      \uptoRecomp(P)$.
    \end{itemize}
    Combining both cases, we obtain
    $([b_1,\dots,b_\ell,\bar{c}],\mathcal{A})\in s(\uptoRecomp(P))$,
    as desired.  \qedhere
  \end{itemize}
\end{proof}

% Note that this definition subsumes both up-to composition (if
% $([c_1,\dots,c_m, \bar c], \mathcal A) \in P$ then
% $([(c_1;\dots;c_m), \bar c], \mathcal A) \in u(P)$) and up-to
% decomposition (if $([(c_1;\dots;c_m), \bar c], \mathcal A) \in P$ then
% $([c_1,\dots,c_m, \bar c], \mathcal A) \in u(P)$).  Up-to composition
% by itself would also be a valid ($s$-compatible) up-to technique,
% however, up-to decomposition (likely) is not.

%\subsubsection{up-to shift}

% ex satisfiability-III

\begin{lemma}[Up-to shift]
  \label{thm:compat-uptoshift}
  Let $P \subseteq \Seq \times \Cond$. Then,
  $\uptoShift(P)$ is its closure under shift, that is,
  $\uptoShift(P) = \{ ([ (c;c_1), c_2, \dots ], \mathcal B) \mid
  ([c_1, c_2, \dots], \mathcal B_{\downarrow c}) \in P \}$.
  The function $\uptoShift$ is $s$-compatible.
  % It holds that $\uptoShift(s(P)) \subseteq s(\uptoShift(P))$.
  % i.e., $\uptoShift$ is $s$-compatible.
\end{lemma}

\begin{proof}
  We show that
  $([(c;c_1), c_2, \dots], \mathcal B) \in \uptoShift(s(P))$ implies
  $([(c;c_1), c_2, \dots], \mathcal B) \in s(\uptoShift(P))$.  Let the
  children of $\mathcal B$ be $\mathcal Q h_j.\mathcal B_j$ for
  $j \in J$ and $\mathcal Q \in \{\forall,\exists\}$.

  Since $([(c;c_1), c_2, \dots], \mathcal B) \in \uptoShift(s(P))$, we
  know there exists some arrow $c$ such that
  $([c_1, c_2, \dots], \mathcal B_{\downarrow c}) \in s(P)$.  We will
  call that condition $\mathcal B_{\downarrow c} = \mathcal A$ and its
  children $\mathcal Q f_i.\mathcal A_i$ for $i \in I$.  As we know it
  results from a shift ($\mathcal A = \mathcal B_{\downarrow c}$),
  hence each child $\mathcal Q f_i.\mathcal{A}_i$ of $\mathcal{A}$ is
  derived from a child $\mathcal Q h_j.\mathcal B_j$ of $\mathcal B$,
  that is $\mathcal{A}_i = (\mathcal{B}_j)_{\downarrow c_i'}$, where
  $(c_i', f_i) \in \kappa(h_j, c)$ is a representative square.  Then
  we make the following case distinction:

  \begin{itemize}
    \proofPart{$\mathcal A$ is existential ($\mathcal Q = \exists$)}
    
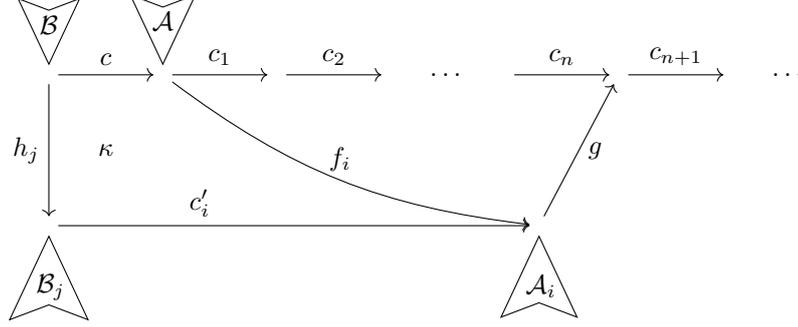
\smallskip
    \begin{tikzpicture}
      % \path[use as bounding box] (-0.33,1) rectangle (2*\sqw+0.1,-\sqh-1);

      \def\clen{1.5}
      \foreach \i in {0,1,2,3,4,5,6,7} {
        \node (c\i) at (\i*\clen,0) {};
      }

      \draw[->] (c0) -- node[above]{$c$} (c1);
      \draw[->] (c1) -- node[above]{$c_1$} (c2);
      \draw[->] (c2) -- node[above]{$c_2$} (c3);
      \node at ($(c3)!0.5!(c4)$) {$\dots$};
      \draw[->] (c4) -- node[above]{$c_n$} (c5);
      \draw[->] (c5) -- node[above]{$c_{n+1}$} (c6);
      \node at ($(c6)!0.5!(c7)$) {$\dots$};

      \node (fitarg) at (4.3*\clen,-2) {};
      \node (hjtarg) at (0,-2) {};

      \node at ($(hjtarg)!0.5!(c1)$) {$\kappa$};

      \draw[->] (c1) to[bend right=15] node[above]{$f_i$} (fitarg);
      \draw[->] (fitarg) to node[right]{$g$} (c5);
      \draw[->] (c0) to node[left]{$h_j$} (hjtarg);
      \draw[->] (hjtarg) to node[above,pos=0.3]{$c_i'$} (fitarg);

      \node[condtri,shape border rotate=270] at (c1.north) {$\mathcal A$};
      \node[condtri,shape border rotate=270] at (c0.north) {$\mathcal B$};
      \node[condtri,shape border rotate=90] at (hjtarg.south) {$\mathcal B_j$\kern-0.5mm};
      \node[condtri,shape border rotate=90] at (fitarg.south) {$\mathcal A_i$\kern-0.5mm};
    \end{tikzpicture}

    Since $([c_1,c_2,\dots], \mathcal A) \in
    s(P)$, % did you forget? A=B shift c, that's why
    by definition of $s$ there exists some child
    $\exists f_i.\mathcal A_i$ of $\mathcal{A}$ and $g,n$ such that
    $c_1;\dots;c_n = f_i ; g$ and
    $([g, c_{n+1}, \dots], \mathcal A_i) \in P$.  $(*_1)$

    Because of
    $([c_1,c_2,\dots], \mathcal A) = ([c_1,c_2,\dots], \mathcal
    B_{\downarrow c})$, we know that child $\exists f_i.\mathcal A_i$
    has been created from some specific child
    $\exists h_j.\mathcal B_j$ of $\mathcal B$ $(*_3)$ such that
    $\mathcal A_i = {\mathcal B_j}_{\downarrow c_i'}$ $(*_2)$ and
    $(c_i',f_i) \in \kappa(h_j,c)$ (which in particular means that
    $h_j;c_i' = c;f_i$). $(*_4)$

    Using that, we can show
    $([ (c;c_1),c_2,\dots ], \mathcal B) \in s(\uptoShift(P))$:
    \begin{itemize}
    \item \emph{there exists a child $\exists h_j.\mathcal B_j$ of
        $\mathcal B$, and $g'$, $n'$}: Implied by $(*_3)$.  We simply
      choose the same child $\exists h_j.\mathcal B_j$ and set
      $g' = c_i' ; g$.  Also, we simply use $n'=n$.
      \item \emph{such that $(c;c_1) ; c_2 ; \dots ; c_{n'} = h_j ; g'$}:
        \begin{alignat*}{2}
          \text{(by $(*_1)$)}
          && c_1 ; \dots ; c_n &= f_i ; g
          \\
          \text{($n=n'$)}
          && c_1 ; \dots ; c_{n'} &= f_i ; g
          \\
          &\implies\quad& c ; c_1 ; \dots ; c_{n'} &= c ; f_i ; g
          \\
          \text{(by $(*_4)$)}
          && &= h_j ; c_i' ; g
          \\
          \text{(choice of $g'$)}
          && &= h_j ; g'
        \end{alignat*}
      \item \emph{and
          $([g', c_{n'+1}, \dots], \mathcal B_j) \in \uptoShift(P)$}:
        \begin{align*}
          (*_1) \phantom{\iff}
          & ([ g, c_{n+1}, \dots], \mathcal A_i) \in P
          \\
          (*_2) \implies
          & ([ g, c_{n+1}, \dots], ({\mathcal B_j})_{\downarrow c_i'}) \in P
          \\
          \text{(Def. $\uptoShift$)} \implies
          & ([\underbrace{c'_i;g}_{g'}, \underbrace{c_{n+1}}_{=c_{n'+1}}, \dots], \mathcal B_j) \in \uptoShift(P)
        \end{align*}
        % \oldtodo{S: the $\iff$ is maybe just an implication?\\
        %   \textbf{L:} actually it would have been also correct to write $([\dots], \mathcal A_i) = ([\dots], {\mathcal B_j}_{\downarrow c_i'}) \in P$ because $\mathcal A_i = {\mathcal B_j}_{\downarrow c_i'}$.\\
        %   I put it on two lines so that I can both refer to $(*_2)$ (that equality) and to $(*_1)$ ($\dots \in P$) properly, hopefully making it less confusing what's going on. (apparently didn't completely work...)\\
        %   (so it's indeed $\iff$, because any statement $A \iff A$ is true)\\
        %   S: sure, but the proof is only needed forwards? So implication would suffice, right? no change needed though: just a matter of taste maybe.}
    \end{itemize}
  \proofPart{$\mathcal A$ is universal ($\mathcal Q = \forall$)}%
    % from here to "We will now..." - might move that to just after "Furthermore, it can be reduced to a representative square as shown above on the right" maybe? not sure
  Since $([c_1,c_2,\dots], \mathcal A) \in s(P)$, by definition of $s$
  for each child $\forall f_i.\mathcal A_i$ of $\mathcal A$ and $g,n$
  such that $c_1;\dots;c_n = f_i ; g$, we have
  $([g, c_{n+1}, \dots], \mathcal A_i) \in P$.

  By
  $([c_1,c_2,\dots], \mathcal A) = ([c_1,c_2,\dots], \mathcal
  B_{\downarrow c})$, we know that each such
  $\forall f_i.\mathcal A_i$ has been created from some specific child
  $\forall h_j.\mathcal B_j$ of $\mathcal B$ such that
  $\mathcal A_i = {\mathcal B_j}_{\downarrow c_i'}$ and
  $(c_i',f_i) \in \kappa(h_j,c)$.

  We will now show
  $([ (c;c_1),c_2,\dots ], \mathcal B) \in s(\uptoShift(P))$ as
  follows.  Let some child $\forall h_j.\mathcal B_j$ of $\mathcal B$,
  $g'$ and $n'$ be given, and assume that
  $(c;c_1) ; c_2 ; \dots ; c_{n'} = h_j ; g'$.  We need to show that
  $([g', c_{n'+1}, \dots], \mathcal B_j) \in \uptoShift(P)$.

    Consider the square shown below on the left.
    According to the assumption from above and by associativity of composition,
    it is a commuting square.

    
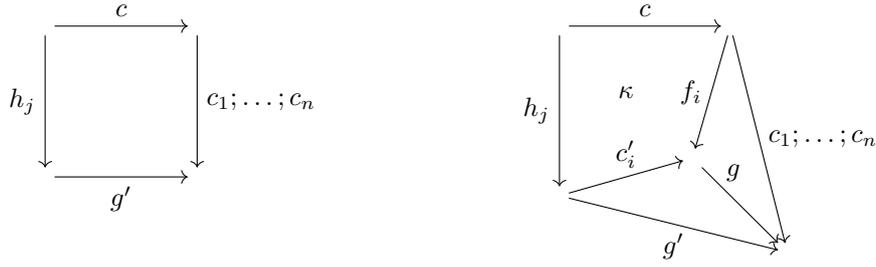
\begin{center}
      \begin{tikzpicture}[baseline={(0,-1)}]
        \node (c0) at (0,0) {};
        \node (c1) at (2,0) {};
        \node (hjtarg) at (0,-2) {};
        \node (cn) at (2,-2) {};

        \draw[->] (c0) -- node[above]{$c$} (c1);
        \draw[->] (c1) -- node[right]{$c_1;\dots;c_n$} (cn);

        \draw[->] (hjtarg) to node[below]{$g'$} (cn);
        \draw[->] (c0) to node[left]{$h_j$} (hjtarg);
      \end{tikzpicture}%
      \hskip 2.5cm%
      \begin{tikzpicture}[baseline={(0,-1)}]
        \node (c0) at (0,0) {};
        \node (c1) at (2.25,0) {};
        \node (hjtarg) at (0,-2.25) {};
        \node (cn) at (3,-3) {};
        \node (fitarg) at (1.75,-1.75) {};

        \draw[->] (c0) -- node[above]{$c$} (c1);
        \draw[->] (c1) -- node[right]{$c_1;\dots;c_n$} (cn);

        \draw[->] (hjtarg) to node[below]{$g'$} (cn);
        \draw[->] (c0) to node[left]{$h_j$} (hjtarg);
        \draw[->] (c1) -- node[left]{$f_i$} (fitarg);
        \draw[->] (hjtarg) -- node[above]{$c_i'$} (fitarg);
        \draw[->] (fitarg) -- node[above right,pos=0.3]{\kern-1mm$g$} (cn);

        \node at ($(c0)!0.5!(fitarg)$) {$\kappa$};
      \end{tikzpicture}
    \end{center}

    Furthermore, it can be reduced to a representative square as shown above on the right.

    Since the construction of the condition $\mathcal A$ involves the same set of squares $\kappa(h_j, c)$,
    we can associate to our given child $\forall h_j.\mathcal B_j$ of $\mathcal B$ and $g'$
    a particular child $\forall f_i.\mathcal A_i$ of $\mathcal A$ and $g$ as shown in the diagram above.

    We now evaluate the definition of $([c_1, c_2, \dots], \mathcal A) \in s(P)$
    for that $\forall f_i.\mathcal A_i$, $g$, and $n' = n$.
    According to that, we have:
    if $c_1;\dots;c_n = f_i ; g$, then $([g, c_{n+1}, \dots], \mathcal A_i) \in P$.
    As can be seen in the representative square diagram above,
%    \oldtodo{yes, the repr sq diagram. We need to look at that one so that we actually get $\dots\in P$ out of the initial statement ``if \dots $c_1;\dots$ then''.
%    (Everything about this proof is correct, but maybe rephrase it anyway so that it is harder to misunderstand)}
    the premise is indeed true,
    hence $([g, c_{n+1}, \dots], \mathcal A_i) \in P$ holds as well.

    By construction of $\mathcal A$, we have
    $\mathcal A_i = {\mathcal B_j}_{\downarrow c_i'}$ and therefore
    $([g, c_{n+1}, \dots], \mathcal A_i) \in P$ implies
    $([g', c_{n+1}, \dots], \mathcal B_j) = ([c_i' ; g, c_{n+1}, \dots], \mathcal B_j)
    \in \uptoShift(P)$,
    which was to be shown.
    \qedhere
  \end{itemize}
\end{proof}

% Motivation: used for witnesses of infinite models, which expects to
% see a condition that, after shifting, equals a previous
% condition. This might never actually happen.  Working with isomorphism
% classes of graphs would fix this issue, but cause others (which?).
% Hence we define the concept of isomorphic conditions and an up-to
% technique for it.

\begin{lemma}[Up-to iso]
  \label{thm:compat-uptoiso}
  Let $P \subseteq \Seq \times \Cond$. Let $\uptoIso(P)$ be
  its closure under isomorphic \linebreak conditions, i.e.,
  $\uptoIso(P) = \{ ([ (h;c_1), c_2, \dots ], \mathcal B) \mid
  \text{$([c_1, c_2, \dots], \mathcal A) \in P,\ \mathcal A \condiso
    \mathcal B$ with iso} \linebreak \text{%
    $h \colon \RO(\mathcal B) \to \RO(\mathcal A)$} \}$.
  The function $\uptoIso$ is $s$-compatible.
  % Then $\uptoIso(s(P)) \subseteq s(\uptoIso(P))$.
  % i.e., $\uptoIso$ is $s$-compatible.
\end{lemma}

\begin{proof}
  We show that
  $([(h;c_1), c_2, \dots], \mathcal B) \in \uptoIso(s(P))$
  implies
  $([(h;c_1), c_2, \dots], \mathcal B) \in s(\uptoIso(P))$.  Let
  the children of $\mathcal B$ be $\mathcal Q g_j.\mathcal B_j$ for
  $j \in J$ and $\mathcal Q \in \{\forall,\exists\}$.

  Since $([(h;c_1), c_2, \dots], \mathcal B) \in \uptoIso(s(P))$, we
  know there exists some condition $\mathcal A$, having children
  $\mathcal Q f_i.\mathcal A_i$ for $i \in I$, such that
  $\mathcal A \condiso \mathcal B$ wrt.\ $h$ and
  $([c_1, c_2, \dots], \mathcal A) \in s(P)$.

  As $\mathcal A \condiso \mathcal B$ wrt.\ $h$, for each $i \in I$
  there is $j \in J$ such that $\mathcal A_i \condiso \mathcal B_j$
  wrt.\ an iso $h_{j,i}$ and $h;f_i = g_j;h_{j,i}$; and vice versa
  (for each $j \in J$ \dots).

  Now we can show $([(h;c_1), c_2, \dots], \mathcal B) \in
  s(\uptoIso(P))$ via the following case distinction:

  \begin{itemize}
    \proofPart{$\mathcal A$ is existential ($\mathcal Q = \exists$)}
    Since $([c_1,c_2,\dots], \mathcal A) \in s(P)$, by definition of
    $s$ there exists some child $\exists f_i.\mathcal A_i$ of
    $\mathcal{A}$ and $g,n$ such that $c_1;\dots;c_n = f_i ; g$ and
    $([g, c_{n+1}, \dots], \mathcal A_i) \in P$.

    Then we can show that
    $([(h;c_1), c_2, \dots], \mathcal B) \in s(\uptoIso(P))$:
    \begin{itemize}
    \item \emph{there exists a child $\exists g_j.\mathcal B_j$ of
        $\mathcal B$, and $g', n'$}: As
      $\mathcal A \condiso \mathcal B$ wrt.\ $h$, there is some child
      $\exists g_j.\mathcal B_j$ of $\mathcal{B}$ such that
      $\mathcal B_j \condiso \mathcal A_i$ for some iso
      $h_{j,i} \colon \RO(\mathcal B_j) \to \RO(\mathcal A_i)$ such
      that $h;f_i = g_j;h_{j,i}$.

        Let $g' \defeq h_{j,i};g$ and $n' \defeq n$.
      \item \emph{such that $h;c_1;\dots;c_{n'} = g_j;g'$}:
        \[ h;c_1;\dots;c_n \stackrel{\mathclap{(1)}}{=} h;f_i;g
          \stackrel{\mathclap{(2)}}{=} g_j;h_{j,i};g = g_j;g' \]
        ($(1)$ from $c_1;\dots;c_n = f_i ; g$, $(2)$ from
        $h;f_i = g_j;h_{j,i}$).
      \item \emph{and
          $([g', c_{n+1}, \dots], \mathcal B_j) \in \uptoIso(P)$}: As
        $\mathcal A_i \condiso \mathcal B_j$ with iso $h_{j,i}$,
        hence $([g, c_{n+1}, \dots], \mathcal A_i) \in P$ implies
        $([(h_{j,i};g), c_{n+1}, \dots], \mathcal B_j) \in
        \uptoIso(P)$.
      \end{itemize}
      
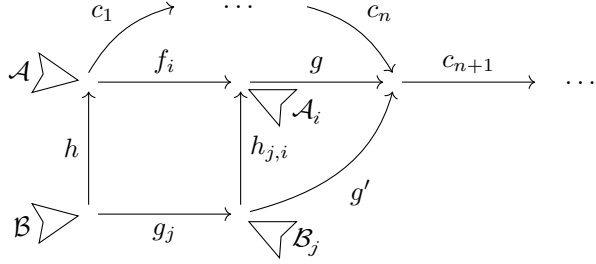
\medskip
    \begin{tikzpicture}[x=1cm,y=1cm]
      \node (ro-a) at (0,0) {};
      \node (ro-ai) at (2,0) {};
      \node (ro-b) at (0,-1.75) {};
      \node (ro-bj) at (2,-1.75) {};
      \node (c1) at (1.25,1) {};
      \node (c2) at (2.75,1) {};
      \node (cn) at (4,0) {};
      \node (cnp1) at (6,0) {};
      \node (cnp2) at (7,0) {};
      \draw[->] (ro-a) -- node[above]{$f_i$} (ro-ai);
      \draw[->] (ro-b) -- node[left]{$h$} (ro-a);
      \draw[->] (ro-bj) -- node[right]{$h_{j,i}$} (ro-ai);
      \draw[->] (ro-b) -- node[below]{$g_j$} (ro-bj);

      \draw[->] (ro-ai) -- node[above]{$g$} (cn);
      \draw[->] (ro-bj) to[out=15,in=-105] node[below right]{$g'$} (cn.south);

      \draw[->] (ro-a.north) to[out=60,in=-165] node[above left]{$c_1$} (c1.west);
      \draw[->] (c2.east) to[out=-15,in=120] node[above right]{$c_n$} (cn.north);
      \draw[->] (cn) -- node[above]{$c_{n+1}$} (cnp1);
      \node at ($(c1)!0.5!(c2)$) {$\dots$};
      \node at ($(cnp1)!0.5!(cnp2)$) {$\dots$};

      \node[condtri,shape border rotate=0, rotate around={-10:(ro-a.center)}, label={[rotate=-10,anchor=east,label distance=1pt]left:{$\mathcal A$}}] at (ro-a.west) {\kern7pt};
      \node[condtri,shape border rotate=0, rotate around={10:(ro-b.center)}, label={[rotate=10,anchor=east,label distance=1pt]left:{$\mathcal B$}}] at (ro-b.west) {\kern7pt};
      \node[condtri,shape border rotate=180, rotate around={-25:(ro-ai.center)}, label={[rotate=-10,anchor=west,label distance=1pt]right:{$\mathcal A_i$}}] at (ro-ai.-20) {\kern7pt};
      \node[condtri,shape border rotate=180, rotate around={-25:(ro-bj.center)}, label={[rotate=-10,anchor=west,label distance=1pt]right:{$\mathcal B_j$}}] at (ro-bj.-20) {\kern7pt};
    \end{tikzpicture}
    \proofPart{$\mathcal A$ is universal ($\mathcal Q = \forall$)}%
    We need to show for all children $\forall g_j.\mathcal B_j$ of
    $\mathcal{B}$ and $g',n'$ that: if $h;c_1;\dots;c_{n'} = g_j;g'$
    then $([g', c_{n+1}, \dots], \mathcal B_j) \in \uptoIso(P)$.

    Let such a child, $g',n'$ be given and let $h;c_1;\dots;c_{n'} = g_j;g'$.

    From $\mathcal A \condiso \mathcal B$ wrt.\ $h$ we know there is a
    corresponding child $\forall f_i.\mathcal A_i$ of $\mathcal{A}$,
    and isos $h \colon \RO(\mathcal B) \to \RO(\mathcal A)$ and
    $h_{j,i} \colon \RO(\mathcal B_j) \to \RO(\mathcal A_i)$ such that
    $h ; f_i = g_j ; h_{j,i}$.

    Now since $h;c_1;\dots;c_{n'} = g_j;g'$, then also
    $c_1;\dots;c_{n} = h^{-1};h;c_1;\dots;c_{n'} = h^{-1};g_j;g' =
    f_i;h_{j,i}^{-1};g'$.

    As $([c_1, c_2, \dots], \mathcal A) \in s(P)$, we know that for
    all children and all $g,n$ (in particular, for
    $g = h_{j,i}^{-1};g'$ and $n = n'$) we have: if
    $c_1;\dots;c_n = f_i;g$ then
    $([g, c_{n+1}, \dots], \mathcal A_i) \in P$.

    Finally, from
    $([(h_{j,i}^{-1};g'), c_{n+1}, \dots], \mathcal A_i) \in P$, using
    $h_{j,i}$, we obtain
    $([g', c_{n'+1}, \dots], \mathcal B_j) =
    ([(h_{j,i};h_{j,i}^{-1};g'), c_{n+1}, \dots], \mathcal B_j) \in
    \uptoIso(P)$, which implies
    $([(h;c_1), c_2, \dots], \mathcal B) \in s(\uptoIso(P))$.
    \qedhere
  \end{itemize}
\end{proof}
\end{toappendix}

\begin{theoremrep}[Up-to techniques]
  \label{thm:compat-upto}
  Let $P \subseteq \Seq \times \Cond$, i.e.\ tuples of
  potential model and condition. Then the following four up-to functions are
  $s$-compatible:
  \begin{itemize}
    \item Conjunction removal:
      We inductively define a relation
      $\uptoConjCond$ containing a pair of conditions
      $(\mathcal A, \mathcal T)$ iff $\mathcal T$ is the same as
      $\mathcal A$ but with some conjunctions removed. That is,
      $\uptoConjCond$ contains
      % \uptoConjCond processes just conditions and \uptoConj processes pairs of model+condition
      \[
      \renewcommand{\arraystretch}{1.2}%
      \setlength{\arraycolsep}{0.3pt}%
      \begin{array}{rlll}
        \big(\bigland_{i \in I} \forall f_i.\mathcal A_i,
%        &\mkern18mu \mathrel{\uptoConjCond} \mkern12mu
        &\ \bigland_{j \in J \subseteq I}
        &\forall f_j.\mathcal T_j\big)
        &\text{\quad whenever $(\mathcal A_j, \mathcal T_j) \in
          \uptoConjCond$ for all $j\in J$}
        \\
        \big(\biglor_{i \in I} \exists f_i.\mathcal A_i,
%        &\mkern18mu \mathrel{\uptoConjCond} \mkern12mu
        &\ \biglor_{i \in I}
        &\exists f_i.\mathcal T_i\big)
        &\text{\quad whenever $(\mathcal A_i, \mathcal T_i) \in
          \uptoConjCond$ for all $i\in I$}
      \end{array}\]
      Then define: \mbox{$\uptoConj(P) = \{ (\bar c, \mathcal T) \mid (\bar c, \mathcal A) \in P,\ (\mathcal A, \mathcal T) \in \uptoConjCond \}$}
    \item
      Recomposition:
      $\uptoRecomp(P) = \{
      ([b_1, \tdots, b_\ell, \bar c ], \mathcal A) \mid
      ([a_1, \tdots, a_k, \bar c], \mathcal A) \in P,\ %
      a_1;\tdots;a_k = b_1;\tdots;b_\ell
      \}$
    \item
      Shift:
      $\uptoShift(P) = \{ ([ (c;c_1), c_2, \tdots ], \mathcal B) \mid
      ([c_1, c_2, \tdots], \mathcal B_{\downarrow c}) \in P \}$
    \item
      Isomorphic condition:
      $\uptoIso(P) = \{ ([ (h;c_1), c_2, \tdots ], \mathcal B) \mid
      \text{$([c_1, c_2, \tdots], \mathcal A) \in P,\ \mathcal A \condiso
        \mathcal B$ with} \linebreak \text{%
        iso $h \colon \RO(\mathcal B) \to \RO(\mathcal A)$} \}$
  \end{itemize}
\end{theoremrep}

\begin{proof}
  See proofs of \Cref{thm:upto-conjunction-removal,thm:compat-uptorecomp,thm:compat-uptoshift,thm:compat-uptoiso} (in the appendix).
\end{proof}

\noindent%
Note however that up-to equivalence $u_\equiv$ is \emph{not} a valid
up-to technique: let $\mathcal U$ be an unsatisfiable condition and
let $P = \{ ([\id,\tdots], \forall \id. \mathcal U) \}$.  As
$\mathcal U \equiv \forall \id. \mathcal U$, then also
$([\id,\tdots], \mathcal U) \in u_\equiv(P)$ and hence
$P \subseteq s(u_\equiv(P))$.  If the technique were correct, this
would imply $id \models \mathcal U$.

%\subsubsection{Combining up-to techniques}

A convenient property of compatibility is that it is preserved by
various operations, in particular, composition, union and iteration
($f^\omega = \bigcup_{i\in\natzero} f^i$). This can be used to combine
multiple up-to techniques into a new one that also has the
compatibility property. \cite{pous:thesis}

\begin{lemma}[combining up-to techniques]\label{upto-techniques-combined}
  Let
  $u = (\uptoRecomp \cup \uptoConj \cup \uptoShift \cup
  \uptoIso)^\omega$ be the iterated application of the up-to
  techniques from \Cref{thm:compat-upto}, then $u$ is $s$-compatible.  
\end{lemma}

We are now able to prove the central theorem needed for showing
completeness.

\begin{theoremrep}[Fair branches are models]
  \label{thm:fair-branch-model}
  Let $\mathcal{A}_0$ be an alternating condition.  Let a fixed
  tableau constructed by the rules of \Cref{satcheck.new} be given.
  Let
  $\mathcal A_0 \xrightarrow{b_1} \mathcal A_1 \xrightarrow{b_2}
  \mathcal A_2 \xrightarrow{b_3} \dots$ be a branch of the tableau
  that is either not extendable and ends with a universal
  quantification (i.e., it is open), or is infinite and fair.  For such a branch, we
  define
  $ P = \{ (\bar b, \mathcal A_i) \mid i \in \natzero,\ \bar b = [
    b_{i+1}, b_{i+2}, \tdots ] \} \subseteq \Seq \times \Cond$,
  i.e., the relation $P$ pairs suffixes of the branch with the corresponding
  conditions.
  Finally, let $u$ be the combination of up-to techniques defined
  in \Cref{upto-techniques-combined}.  Then, $P \subseteq s(u(P))$,
  which implies that $P\subseteq\; \models$.  In other words, every
  such branch in a tableau of \Cref{satcheck.new} corresponds to a
  model of $\mathcal A_0$.
\end{theoremrep}

\def\varsat#1{#1}
\def\fsat#1{d_{#1}}
\def\frisat#1{r_{#1}}
\def\gsat#1{g_{#1}}
\def\Asat#1{\mathcal D_{#1}}
\def\varpf#1{\ell}
\def\fpf#1{f_{\ell}}
\def\Apf#1{\mathcal A_{\ell}}

%\begin{toappendix}
\begin{proofsketch}
  Let $([ c_1, c_2, \tdots ], \mathcal C_0) \in P$, which corresponds
  to a suffix 
  $\mathcal{C}_0 \xrightarrow{c_1} \mathcal C_1 \xrightarrow{c_2}
  \mathcal{C}_2 \xrightarrow{c_3} \dots$ of the chosen branch. We show that
  $([ c_1, c_2, \tdots ], \mathcal C_0) \in s(u(P))$:
  \begin{itemize}
    \item \textbf{$\mathcal C_0$ is existential:} The next label on the
    branch is the arrow of some child $\exists c_1.\mathcal{C}_1$ of
    $\mathcal{C}_0$ and $([ c_2, c_3, \tdots ], \mathcal C_1) \in P$ implies 
    $([ c_1, c_2, \tdots ], \mathcal C_0) \in s(P)\subseteq s(u(P))$.
    
    \item \textbf{$\mathcal C_0$ is universal and contains no
      iso:} The branch ends at this point, so the sequence
    of arrows in the tuple is empty and represents $\id$, where
    $(\id,\mathcal{C}_0)\in s(u(P))$: no $f_i$ is an iso, therefore
    $f_i ; g = \id$ is never true and hence the universal condition
    is trivially satisfied.

    \item \textbf{$\mathcal C_0$ is universal and contains at least one
      iso:} By definition of $s$, we need to be able to satisfy any
    given child $\forall \fsat0.\Asat0$, i.e., show that whenever
    $c_1;\tdots;c_n = \fsat0;g$ for some $g,n$, then
    $([g, c_{n+1}, \tdots], \Asat0) \in u(P)$.

    Fairness guarantees that an indirect successor
    $\forall \fsat q.\Asat q$ of $\forall \fsat0.\Asat0$ is pulled
    forward, which results (after up-to conjunction removal) in a
    tuple $([ c_m, \tdots ], \Asat q) \in u(P)$.  The intermediate
    steps on the branch, where other children are pulled forward
    instead, allow expressing $\Asat q$ as
    $(\Asat0)_{\downarrow \alpha_1 \downarrow \dots \downarrow
      \alpha_q}$.  Use up-to shift to transform to
    $([ \alpha_1 ; \tdots ; \alpha_q ; c_m, \tdots ], \Asat0) \in
    u(P)$, then up-to recomposition to the required
    $([\gsat0, c_{n+1}, \tdots], \Asat0) \in u(P)$.  \qedhere
  \end{itemize}
\end{proofsketch}
%\end{toappendix}

\begin{proof}
  Let $([ c_1, c_2, \dots ], \mathcal C_0) \in P$, which corresponds
  to a (suffix of the chosen) branch
  $\mathcal A = \mathcal{C}_0 \xrightarrow{c_1} \mathcal{C}_1
  \xrightarrow{c_2} \mathcal{C}_2 \xrightarrow{c_3} \dots$ be
  of the tableau. We show that
  $([ c_1, c_2, c_3, \dots ], \mathcal A) \in s(u(P))$.
\begin{proofparts}
  \proofPart{$\mathcal A$ is existential} Let
  $\mathcal A = \biglor_{i \in I} \exists f_i . \mathcal A_i$.  Hence
  the first step of our branch ($c_1$) corresponds to one of the
  transitions
  $\biglor_{i \in I} \exists f_i . \mathcal A_i \xrightarrow{f_\ell}
  \mathcal A_p = \mathcal{C}_1$, where $c_1 = f_\ell$, for some
  $\ell \in I$.  This implies that
  $([ c_2, \dots ], \mathcal A_\ell) \in P$.

  For
  $([ c_1, c_2, \dots ], \biglor_{i \in I} \exists f_i . \mathcal A_i)
  \in s(u(P))$ to hold, by definition of $s$ there must exist some
  $i \in I$ (i.e., a child $\exists f_i.\mathcal A_i$), an arrow $g$
  and a number $n$ such that $c_1 ; \dots ; c_n = f_i ; g$ and
  $([ g , c_{n+1}, \dots ], \mathcal A_i) \in u(P)$.  This definition
  can be satisfied with $n = 1$, $i=\ell$ and $g = \id$: indeed
  $c_1 = f_\ell = f_\ell ; \id$ and
  $([ c_2, \dots ], \mathcal A_\ell) \in P$ implies (up-to
  recomposition)
  $([g, c_{1+1}, \dots], \mathcal A_\ell) = ([ \id, c_2, \ \dots ],
  \mathcal A_\ell) \in u(P)$.

  Note that the first step ($c_1$) is guaranteed to exist since we
  assume that the branch ends with a universal quantification (if it
  ends at all).

  \proofPart{$\mathcal A$ is universal and has no isomorphisms} If
  $\mathcal A$ is an empty universal (i.e., $\condtrue$),
  $(\bar c, \mathcal A) \in s(u(P))$ trivially holds, since the
  definition of $s$ contains a universal quantification over the set
  of child conditions, and empty universals are always true.

  In the following, let
  $\mathcal A = \bigwedge_i \forall f_i . \mathcal A_i$.  If no $f_i$
  is an isomorphism, the tableau does not have any child nodes for
  $\mathcal A$, so the sequence $[ c_1, \dots ]$ is empty, and thus
  its composition equals $\id$.  Then,
  $([\,], \mathcal A) \in s(u(P))$ holds because there is no arrow $g$
  that makes $\id = f_i ; g$ true ($\id$ is an isomorphism and hence
  can only split in two other isomorphisms, however, $f_i$ is not an
  isomorphism).  Thus the implication is trivially true.
  % \oldtodo{\textbf{L:} this case (+ empty universal above) is
  %   (eventually) hit if a condition has a finite model (and it has
  %   just been found). Maybe add that here}

  \proofPart{$\mathcal A$ is universal and has at least one isomorphism}
  We now assume that at least one $f_i$ is an isomorphism.

  For
  $([ c_1, c_2, \dots ], \bigland_{i \in I} \forall f_i . \mathcal
  A_i) \in s(u(P))$ to hold, by definition of $s$, for all
  $\forall f_i.\mathcal A_i$, all arrows~$g$ and all $n \in \natzero$,
  we need to show: if $c_1 ; \dots ; c_n = f_i ; g$, then
  $([ g, c_{n+1}, \dots ], \mathcal A_i) \in u(P)$.

  Hence let $\forall \fsat0.\Asat0$ (where $\fsat0 = f_j$,
  $\Asat0 = \mathcal{A}_j$ for some index $j$) be such a child and
  assume some $n$, $\gsat0$ such that
  $c_1 ; \dots ; c_n = \fsat0 ; \gsat0$. We need to show that
  \mbox{$([ \gsat0, c_{n+1}, \dots ], \Asat0) \in u(P)$}.

  Since at least one child is an isomorphism, in the next step on the
  branch, some $\forall \fpf0.\Apf0$ ($\fpf0$ isomorphism) is pulled
  forward, so the next condition on the path has the shape
  $\biglor_j \exists h_j.(\mathcal B_j \land \dots)$, and it is (up-to
  conjunction removal) the same as the child
  $\Apf0 = \biglor_j \exists h_j.\mathcal B_j$ that was pulled
  forward.  If the child to be satisfied is the one that was pulled
  forward ($\forall \fsat0.\Asat0 = \forall \fpf0.\Apf0$), we would
  thus immediately obtain the desired result
  $([\dots], \Asat0) \in u(P)$.

  However, the $\forall \fpf0.\Apf0$ that is pulled forward at the
  current step might be a different child than
  $\forall \fsat0.\Asat0$.  While the fairness property would
  guarantee that every isomorphism is eventually pulled forward at
  some point in the future, $\fsat0$ might not even be an isomorphism.

  So assume that instead of $\fsat 0$, for the initial steps of the
  path, other isomorphisms are pulled forward instead.  Since the
  condition is alternating and the tableau rules preserve this
  property, we refer to $c_1, c_2, c_3, c_4, \ldots$ as
  $u_1, e_1, u_2, e_2, \ldots$, with $u_i, e_i$ corresponding to the
  labels of steps from a universal or existential condition,
  respectively (i.e., $u_1$ was pulled forward first). That is
  $u_i = c_{2i-1}$, $e_i = c_{2i}$.

  For each child not yet pulled forward (such as
  $\forall \fsat0.\Asat0$), the condition $\mathcal{C}_2$ at the next
  universal step contains successors (such as the children of
  $(\forall \fsat0.\Asat0)_{\downarrow u_1;e_1}$), which are possibly
  ``closer'' to an isomorphism than $\fsat 0$ was.
  % \oldtodo{\textbf{L:} It is already hinted here and before, that it is
  % not guaranteed that any successor of $\fsat0$ is ever pulled
  % forward, because fairness makes no guarantees about that.  It's
  % also generally \textbf{not} guaranteed that $\fsat0$ is ever
  % turned into an iso.

  % Example condition: rooted at two nodes, $\forall \id.\exists
  % b.\condtrue \land \forall a.\exists \id.\condtrue$ where $b$ adds B-edge and $a$ adds A-edge.
  % Let $\fsat0$ be the second one. The first one is pulled forward first. The result (after another $\exists$ step) is a condition without isos, so the path ends.

  % Hence: no successor of $\fsat0$ is an iso and it is not pulled
  % forward.  This is not a \textbf{big} problem, because there is no
  % $\gsat0$ such that $\fsat0;\gsat0 = c_1;\dots;c_n$ (this is only
  % possible if $\fsat0$ has an iso successor, in which case it
  % \textbf{is} pulled forward).  But currently, the proof claims that
  % it \textbf{will always} be turned into an iso, which isn't true
  % (at least in an obvious way)}
  %
  % ^ this has been fixed by now I think

  We now claim that there exists a finite sequence
  $\forall \fsat0.\Asat0, \forall \fsat1.\Asat1, \dots,
  \allowbreak\forall \fsat p.\Asat p, \dots, \allowbreak\forall \fsat
  q.\Asat q$ of children of the universal conditions
  $\mathcal C_0, \mathcal C_2, \dots$ on the branch (i.e., every
  second one), with each element of the sequence being a (direct)
  successor of the previous element, such that:

  \begin{enumerate}
    \item
      After $p$ steps, $0 \leq 2p \leq n$,
      $\fsat p$ is an isomorphism,
    \item
      for $k>p$, $\fsat k$ is an isomorphism as well,
    \item
      after $q$ steps, $p \leq q < \infty$,
      $\fsat q$ is pulled forward,
      resulting in $([ e_{q+1}, \dots ], \Asat q) \in u(P)$.
  \end{enumerate}
  Afterwards, we can transform that to $([ \gsat0, c_{n+1}, \dots ], \Asat0) \in u(P)$ (as required by the satisfaction function $s$) by applying up-to techniques,
  thereby showing that the path actually describes a model.

  Note that $p=0$ or $p=q$ are possible as well.  The initial steps
  are depicted in the following commuting diagram, with the chain
  $u_1,e_1,\dots$, $n$, $\fsat0,\Asat0$ and $\gsat0$ given. The
  remaining conditions and arrows will now be constructed.

  
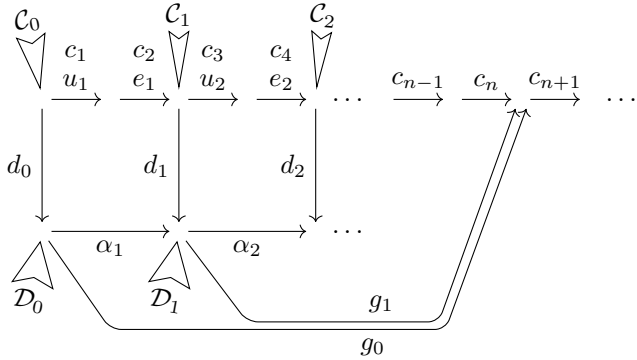
\begin{tikzpicture}[x=0.90cm]
    \foreach \i in {0,...,9} {
      \node (top\i) at (\i,0) {};
      \node (bot\i) at (\i,-1.75) {};
    }

    \draw[->] (top0) edge node[above,align=center]{$c_1$\\[-2pt]$u_1$} (top1)
              (top1) edge node[above,align=center]{$c_2$\\[-2pt]$e_1$} (top2)
              (top2) edge node[above,align=center]{$c_3$\\[-2pt]$u_2$} (top3)
              (top3) edge node[above,align=center]{$c_4$\\[-2pt]$e_2$} (top4);
    \node at ($(top4)!0.5!(top5)$) {$\dots$};
    \draw[->] (top5) edge node[above]{$c_{n-1}$} (top6)
              (top6) edge node[above]{$c_n$} (top7)
              (top7) edge node[above]{$c_{n+1}$} (top8);
    \node at ($(top8)!0.5!(top9)$) {$\dots$};

    \draw[->] (bot0) edge node[below]{$\alpha_1$} (bot2);
    \draw[->] (bot2) edge node[below]{$\alpha_2$} (bot4);
    \node at ($(bot4)!0.5!(bot5)$) {$\dots$};

    \draw[->] (top0) -- node[left]{$\fsat0$} (bot0);
    \draw[->] (top2) -- node[left]{$\fsat1$} (bot2);
    \draw[->] (top4) -- node[left]{$\fsat2$} (bot4);

    \draw[->,rounded corners=4pt] (bot0) -- +(1.0,-1.3) -- node[below,pos=0.7828]{$\gsat0$} +(5.9,-1.3) -- (top7.-60);
    \draw[->,rounded corners=4pt] (bot2) -- +(1.0,-1.2) -- node[above,pos=0.7]{$\gsat1$} +(3.8,-1.2) -- (top7.-120);

    \node[condtri,dart tip angle=30,shape border rotate=270,rotate around={10:(top0.center)}, label={[rotate=10,anchor=south,label distance=1pt]above:{$\mathcal{C}_0$}}] at (top0.north) {\kern5pt};
    \node[condtri,dart tip angle=20,shape border rotate=270,rotate around={0:(top2.center)}, label={[rotate=0,anchor=south,label distance=1pt]above:{$\mathcal{C}_1$}}] at (top2.north) {\kern3.0pt};
    \node[condtri,dart tip angle=20,shape border rotate=270,rotate around={-5:(top4.center)}, label={[rotate=-5,anchor=south,label distance=1pt]above:{$\mathcal{C}_2$}}] at (top4.north) {\kern3.0pt};
    \node[condtri,shape border rotate=90, rotate around={-10:(bot0.center)}, label={[rotate=-10,anchor=north,label distance=1pt]below:{$\Asat 0$}}] at (bot0.south) {\kern7pt};
    \node[condtri,shape border rotate=90, rotate around={-10:(bot2.center)}, label={[rotate=-10,anchor=north,label distance=1pt]below:{$\Asat 1$}}] at (bot2.south) {\kern7pt};
  \end{tikzpicture}

  The proof objective is to construct the sequence of all further
  $\forall \fsat{k+1}.\Asat{k+1}$ and associated $\alpha_{k+1}$,
  $\gsat{k+1}$.

  \begin{enumerate}
  \item We construct the aforementioned sequence by repeatedly
    choosing a successor of $\forall \fsat k.\Asat k$ as the next
    element ($\forall \fsat{k+1}.\Asat{k+1}$) of the sequence.  We
    then show that for some $p$ with $2p \leq n$, $\fsat p$ must be an
    isomorphism.

    Hence, consider the following situation:

    \medskip

      
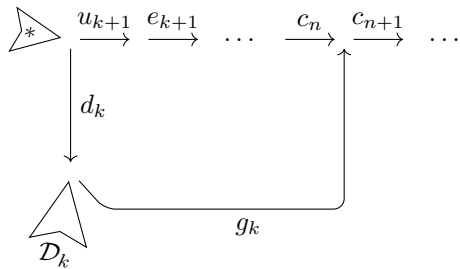
\begin{tikzpicture}[x=0.90cm]
        \foreach \i in {0,...,6} {
          \node (top\i) at (\i,0) {};
          \node (bot\i) at (\i,-1.75) {};
        }

        \draw[->] (top0) edge node[above]{$u_{k+1}$} (top1)
                  (top1) edge node[above]{$e_{k+1}$} (top2);
        \node at ($(top2)!0.5!(top3)$) {$\dots$};
        \draw[->] (top3) edge node[above]{$c_n$} (top4)
                  (top4) edge node[above]{$c_{n+1}$} (top5);
        \node at ($(top5)!0.5!(top6)$) {$\dots$};

        \draw[->] (top0) -- node[right]{$\fsat k$} (bot0);

        \draw[->,rounded corners=4pt] (bot0) -- +(0.5,-0.5) -| node[below,pos=0.3]{$\gsat{k}$} (top4);

        \node[condtri,shape border rotate=90, rotate around={-10:(bot0.center)}, label={[rotate=-10,anchor=north,label distance=1pt]below:{$\Asat k$}}] at (bot0.south) {\kern12pt};
        \node[condtri,shape border rotate=0,rotate around={-10:(top0.center)},scale=0.9] at (top0.west) {$*$};
      \end{tikzpicture}

      The condition labeled $*$ on the top left corresponds to
      $\mathcal{C}_{2k}$.  It is of the form
      $\Big( \bigwedge_{m \ne \varsat j} \forall f_m . \mathcal A_m
      \Big) \land \forall \fsat k . \Asat k$, i.e., one of its
      children is the current element
      $\forall \fsat k . \Asat k = \forall f_j.\mathcal{A}_j$ of the
      sequence, and further children $\forall f_m . \mathcal A_m$.

      If $\fsat k$ is an isomorphism, we are done, set $p = k$ and we
      do not need to choose a next element. We will show that this
      point is reached eventually and that $2p \leq n$ later in the proof.

      So assume that $\fsat k$ is not an isomorphism.
      So some other arrow $\fpf{k+1}$ was pulled forward in this step instead.
%      (Note that we are still before $n$, so the existence of $c_1 \dots c_n$ itself implies that such a step has actually taken place.)
%      \oldtodo{\textbf{L:} Since $n=0$ is allowed, it might not be the case that $c_n$ actually exists. Of course then the ``chain'' $c_1 \dots c_0$ is empty and hence $\id$, which can only split into two isos, from which it follows that $\fsat k$ is an isomorphism already. But we have to state that somewhere.}

      \begin{itemize}

        \item
          Assume $u_{k+1}$ is not $c_n$ and there is at least one more arrow $e_{k+1}$ before $c_n$ is reached.

          Let $\forall \fpf{k+1}.\Apf{k+1}$ be the child that was
          chosen in the tableau to be pulled forward now (i.e.,
          $\fpf{k+1} = u_{k+1}$ is an isomorphism).  Then, the
          condition $\mathcal{C}_{2k}$ at the current node (marked $*$
          in the diagram above) is
          $\Big( \bigwedge_{m \ne \varpf{k+1}} \forall f_m . \mathcal
          A_m \Big) \land \forall \fpf{k+1} . \Apf{k+1}$, with
          $\forall \fsat k . \Asat k$ being among the
          $\forall f_m . \mathcal A_m$, $m\ne \ell$.  Assume the
          condition that was pulled forward has the structure
          $\Apf{k+1} = \bigvee_j \exists h_j . \mathcal B_j$, which
          results in a step:

          \[
          \begin{array}{l}\displaystyle
            \Big( \bigwedge_{\kern10pt\mathclap{m \ne \varpf{k+1}}\kern10pt}
            \forall f_m . \mathcal A_m \Big)
            \land \forall \fpf{k+1} . \Apf{k+1}
%          \\\mbox{}\hspace{47mm}%FIXME
          \xrightarrow[= u_{k+1}]{\fpf{k+1}}
            \underbrace{\bigvee_j \exists h_j . \Big(
              \mathcal B_j}_{\Apf{k+1}\quad} \land
              \underbrace{
              \Big( \bigwedge_{\kern10pt\mathclap{m \ne \varpf{k+1}}\kern10pt}
              \forall f_m . \mathcal A_m \Big)%
              _{\downarrow \fpf{k+1} ; h_j}
              }_{\mathclap{\text{contains successors of $\forall \fsat k . \Asat k$}}}
            \Big)
          \end{array}
          \]

          Since the next condition is an existential, it has outgoing transitions for each $h_j$, with one of them corresponding to the $e_{k+1}$ that is next on the path:

          \[
            \bigvee_j \exists h_j . \Big(
              \mathcal B_j \land
              \Big( \bigwedge_{\kern10pt\mathclap{m \ne \varpf{k+1}}\kern10pt}
              \forall f_m . \mathcal A_m \Big)%
              _{\downarrow \fpf{k+1} ; h_j}
            \Big)
          \xrightarrow[= e_{k+1}]{h_j}
            \mathcal B_j \land
            \underbrace{
            \Big( \bigwedge_{\kern10pt\mathclap{m \ne \varpf{k+1}}\kern10pt}
            \forall f_m . \mathcal A_m \Big)%
            _{\downarrow \fpf{k+1} ; h_j}
            }_{\mathclap{\text{contains successors of $\forall \fsat k . \Asat k$}}}
          \]

          We now reinterpret the diagram from before as a square, shown below in (i).
%          Note that due to the previous two transitions, $u_{k+1}$ and $e_{k+1}$ are $\fpf{k+1}$ and $h_j$.
          The right side of square (i) is the composition of all remaining arrows until $c_n$. In case $e_{k+1}$ is actually $c_n$ and there are no further arrows left, the right side is simply $\id$.

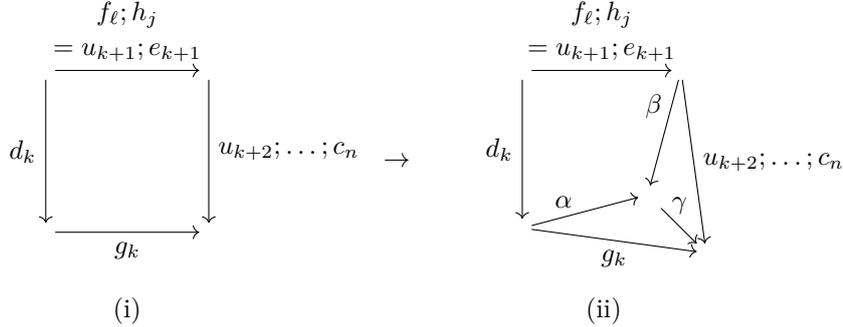
          \begin{tikzpicture}[x=1.10cm,y=1.10cm]% PAGEBREAK-ADJUST
            \def\sqw{2.2}
            \def\sqh{2.2}
            \def\sqhsp{0.3} % how much to inset C, B compared to the full D' square size
            \def\sqhspi{0.4} % how much to additionally inset D (in addition to \sqhsp)
            \def\sqhsr{0.25} % shrink of the left side. was 0 in masters thesis, but a tiny offset is acceptable (it should however not be exactly halfway between D and D', because this might suggest that D' is "split" to some smaller object D and some larger-than-D' object that is also called D'
            \def\sqhsrr{0.25} % like sqhsr but for vertical size of the first diagram

            % %%% LEFT - PRE R-STEP %%%
            \begin{scope}[shift={(0,0)}]
              \node (a) at (0,0) {};
              \node (b) at (\sqw-\sqhsr,0) {};
              \node (c) at (0,-\sqh+\sqhsrr) {};
              \node (ds) at (\sqw-\sqhsr,-\sqh+\sqhsrr) {};
              \draw[->] (a) -- node[above,align=center]{$\fpf{k+1} ; h_j$\\$= u_{k+1} ; e_{k+1}$} (b);
              \draw[->] (a) -- node[left]{$\fsat k$} (c);

              \draw[->] (b) -- node[right]{$u_{k+2} ; \dots ; c_n$} (ds);
              \draw[->] (c) -- node[below]{$\gsat k$} (ds);

              \node at (0.5*\sqw-0.5*\sqhsr,-\sqh -0.7) {(i)};
            \end{scope}

            \node at (\sqw+2,-0.5*\sqh) {$\rightarrow$};

            % %%% RIGHT %%%
            \begin{scope}[shift={(\sqw+3.5,0)}]
              \node (a) at (0,0) {};
              \node (b) at (\sqw-\sqhsp,0) {};
              \node (c) at (0,-\sqh+\sqhsp) {};
              \node (d) at (\sqw-\sqhsp-\sqhspi,-\sqh+\sqhsp+\sqhspi) {};
              \node (ds) at (\sqw,-\sqh) {};
              \draw[->] (a) -- node[above,align=center]{$\fpf{k+1} ; h_j$\\$= u_{k+1} ; e_{k+1}$} (b);
              \draw[->] (a) -- node[left]{$\fsat k$} (c);
              \draw[->] (b) -- node[left,pos=0.25]{$\beta$} (d);
              \draw[->] (c) -- node[above,pos=0.3]{$\alpha$} (d);

              \draw[->] (d.center)+(5pt,-5pt) -- (ds);
              \node at ($(d.center)+(12pt,-4pt)$) {$\gamma$};
              \draw[->] (b) -- node[right]{$u_{k+2} ; \dots ; c_n$} (ds);
              \draw[->] (c) -- node[below]{$\gsat k$} (ds);

              \node at (0.5*\sqw-0.5*\sqhsr,-\sqh -0.7) {(ii)};
            \end{scope}
          \end{tikzpicture}%

          Being a commuting square, (i) can be reduced to some representative square
          $(\alpha, \beta) \in \kappa(\fsat k, (u_{k+1} ; e_{k+1})) = \kappa(\fsat k, (\fpf{k+1};h_j))$
          (shown above in square (ii) on the right).
          Depending on $\kappa$, there might be several pairs of $\alpha,\beta$ that make (ii) commute, in which case an arbitrary one can be chosen.

          The shift that happens in the current condition (in particular that of $\forall \fsat k . \Asat k$) refers to the same set of squares:

          \[ (\forall \fsat k . \Asat k)_{\downarrow \fpf{k+1};h_j} =
          \bigwedge_{(\alpha,\beta) \in \kappa(\fsat k, (\fpf{k+1};h_j))}
          \forall \beta . ({\Asat k})_{\downarrow \alpha} \]

        This means that one of the successors of
        $\forall \fsat k . \Asat k$ actually uses the representative
        square of (ii) above.  Now let $\alpha, \beta$ be the pair of
        arrows that close that square and $\gamma$ be the mediating
        morphism.  Then rename $\alpha_{k+1} = \alpha$,\
        $\fsat{k+1} = \beta$,\ \mbox{$\gsat{k+1} = \gamma$},\
        $\Asat{k+1} = (\Asat{k})_{\downarrow \alpha_{k+1}}$ (note that
        $(\Asat k)_{\downarrow \alpha_{k+1}}$ is one of the successors
        of $\Asat k$) and extend our original diagram:

          
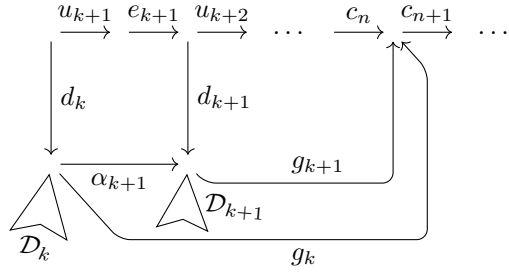
\begin{tikzpicture}[x=0.90cm]
            \foreach \i in {0,...,7} {
              \node (top\i) at (\i,0) {};
              \node (bot\i) at (\i,-1.75) {};
            }

            \draw[->] (top0) edge node[above]{$u_{k+1}$} (top1)
                      (top1) edge node[above]{$e_{k+1}$} (top2)
                      (top2) edge node[above]{$u_{k+2}$} (top3);
            \node at ($(top3)!0.5!(top4)$) {$\dots$};
            \draw[->] (top4) edge node[above]{$c_n$} (top5)
                      (top5) edge node[above]{$c_{n+1}$} (top6);
            \node at ($(top6)!0.5!(top7)$) {$\dots$};

            \draw[->] (bot0) edge node[below]{$\alpha_{k+1}$} (bot2);

            \draw[->] (top0) -- node[right]{$\fsat k$} (bot0);
            \draw[->] (top2) -- node[right]{$\fsat{k+1}$} (bot2);

            \draw[->,rounded corners=4pt] (bot0) -- +(1,-1) -| node[below,pos=0.3]{$\gsat{k}$} ($(top5)+(0.5,-0.5)$) -- (top5);
            \draw[->,rounded corners=4pt] (bot2) -- +(0.25,-0.25) -| node[above,pos=0.3]{$\gsat{k+1}$} (top5);

            \node[condtri,shape border rotate=90, rotate around={-10:(bot0.center)}, label={[rotate=-10,anchor=north,label distance=1pt]below:{$\Asat k$}}] at (bot0.south) {\kern12pt};
            \node[condtri,shape border rotate=90, rotate around={-5:(bot2.center)}, label={[rotate=-5,anchor=west,label distance=-1pt]45:{$\Asat{k+1}$}}] at (bot2.south) {\kern10pt};
          \end{tikzpicture}

          Now start over from the beginning, with $\fsat{k+1}, \gsat{k+1}$ instead of $\fsat k, \gsat k$.
          (Since the chain $u_{k+2}, \dots c_n$ is now two elements shorter than before, this process will not continue endlessly.)

        \item
          Consider the case that $c_n = u_{k+1}$: %, i.e., we have just reached $c_n$.

          
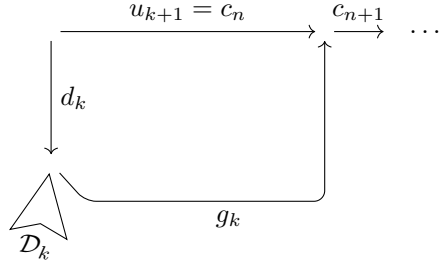
\begin{tikzpicture}[x=0.90cm]
            \foreach \i in {0,...,6} {
              \node (top\i) at (\i,0) {};
              \node (bot\i) at (\i,-1.75) {};
            }

            \draw[->] (top0) edge node[above]{$u_{k+1} = c_n$} (top4)
                      (top4) edge node[above]{$c_{n+1}$} (top5);
            \node at ($(top5)!0.5!(top6)$) {$\dots$};

            \draw[->] (top0) -- node[right]{$\fsat k$} (bot0);

            \draw[->,rounded corners=4pt] (bot0) -- +(0.5,-0.5) -| node[below,pos=0.3]{$\gsat{k}$} (top4);

            \node[condtri,shape border rotate=90, rotate around={-10:(bot0.center)}, label={[rotate=-10,anchor=north,label distance=1pt]below:{$\Asat k$}}] at (bot0.south) {\kern12pt};
          \end{tikzpicture}

          Labels of transitions from a universal condition (such as $u_{k+1}$) are always isomorphisms.
          Hence $u_{k+1} = \fsat k ; \gsat k$ can only split into isomorphisms,
          so $\fsat k$ must be an isomorphism as well.
        \item Consider the case that $c_n = e_{k+1}$ or $n=0$: In both
          cases the remaining sequence $c_{n+1},\dots,c_n$ is empty,
          representing the identity $\id$. Hence
          $\fsat{k+1};\gsat{k+1} = \id$, which implies that
          $\fsat{k+1}$ is an iso.

          (This subcase and the previous one ensure that an
          isomorphism is obtained after finitely many -- at most $p$
          with $2p\le n$ -- steps.)
      \end{itemize}

    \item Summarizing the situation so far, the original morphism
      $\fsat 0$ has been ``turned into an iso'' $\fsat q$, more
      precisely there is a successor $\fsat{p}.\Asat{p}$ of
      $\fsat{0}.\Asat{0}$ that is an iso:

      
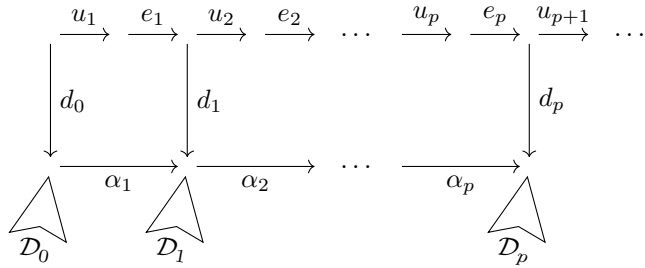
\begin{tikzpicture}[x=0.90cm]
        \foreach \i in {0,...,9} {
          \node (top\i) at (\i,0) {};
          \node (bot\i) at (\i,-1.75) {};
        }

        \draw[->] (top0) edge node[above]{$u_1$} (top1)
                  (top1) edge node[above]{$e_1$} (top2)
                  (top2) edge node[above]{$u_2$} (top3)
                  (top3) edge node[above]{$e_2$} (top4);
        \node at ($(top4)!0.5!(top5)$) {$\dots$};
        \draw[->] (top5) edge node[above]{$u_p$} (top6)
                  (top6) edge node[above]{$e_p$} (top7)
                  (top7) edge node[above]{$u_{p+1}$} (top8);
        \node at ($(top8)!0.5!(top9)$) {$\dots$};

        \draw[->] (bot0) edge node[below]{$\alpha_1$} (bot2);
        \draw[->] (bot2) edge node[below]{$\alpha_2$} (bot4);
        \node at ($(bot4)!0.5!(bot5)$) {$\dots$};
        \draw[->] (bot5) edge node[below]{$\alpha_p$} (bot7);

        \draw[->] (top0) -- node[right]{$\fsat0$} (bot0);
        \draw[->] (top2) -- node[right]{$\fsat1$} (bot2);
        \draw[->] (top7) -- node[right]{$\fsat p$} (bot7);

        \node[condtri,shape border rotate=90, rotate around={-10:(bot0.center)}, label={[rotate=-10,anchor=north,label distance=1pt]below:{$\Asat 0$}}] at (bot0.south) {\kern12pt};
        \node[condtri,shape border rotate=90, rotate around={-10:(bot2.center)}, label={[rotate=-10,anchor=north,label distance=1pt]below:{$\Asat 1$}}] at (bot2.south) {\kern12pt};
        \node[condtri,shape border rotate=90, rotate around={-10:(bot7.center)}, label={[rotate=-10,anchor=north,label distance=1pt]below:{$\Asat p$}}] at (bot7.south) {\kern12pt};
      \end{tikzpicture}

      To be able to pull an indirect successor of $\fsat p$ forward,
      it has to be an isomorphism.  Generally, for $k \geq p$, not all
      direct successors of $\forall \fsat k.\Asat k$ might have
      isomorphisms (depending on the class $\kappa$ in use).  However,
      shifting isomorphisms always results in at least one isomorphism
      again (cf.\ \Cref{rsq-preserve-sections}), so at least one
      successor of $\forall \fsat k.\Asat k$ contains an isomorphism.
      Hence for $\forall \fsat{k+1}.\Asat{k+1}$ we choose an arbitrary
      one of these, and thereby ensure that all elements
      $\forall \fsat{p+1}.\Asat{p+1}, \allowbreak\forall
      \fsat{p+2}.\Asat{p+2}, \dots$ of the sequence do in fact have
      isomorphisms.

      Furthermore by construction we have that
      $u_{i+1};e_{i+1};d_{i+1} = d_k;\alpha_{i+1}$ for all $i$ which
      means that the row of squares in the diagram above
      commutes. From the triangles in the representative square
      diagram we obtain that $g_i = \alpha_{i+1};g_{i+1}$ and
      $d_{i+1};g_{i+1} = u_{i+2};\dots;c_n$. Whenever $c_n$ is a
      morphism derived from an existential step ($c_n = e_{k+1}$), we
      have that $d_{k+1};g_{k+1} = \id$ (see above). And whenever
      $c_n$ is a morphism derived from a universal step
      ($c_n = u_{k+1}$), we know that $d_k;g_k = u_{k+1}$ and
      $\alpha_{k+1} = g_k;c_{n+1};d_{k+1}$. The latter is true since
      $d_k;\alpha_{k+1} = u_{k+1};e_{k+1};d_{k+1} =
      d_k;g_k;e_{k+1};d_{k+1} = d_k;g_k;c_{n+1};d_{k+1}$ and $d_k$ is
      an iso as argued above.

    \item Since $\fsat p$ is an isomorphism, it must eventually be
      pulled forward.  We continue the construction of the morphisms,
      potentially beyond $c_n$, and choose the representative squares
      such that the arrows $\fsat k$ with $k\ge p$ are also isos
      (cf.\ the assumption at the end of \Cref{sec:representative-squares-shift}).

      If the path is infinite, the fairness constraint guarantees the
      existence of an index $q$ where an iso $d_q$ which is a
      successor of $d_p$ is pulled forward.  If the path is finite,
      the length of the path is an upper bound for the value of $q$,
      and some indirect successor of $\fsat p$ must have been pulled
      forward by then: the only way that a branch can be unextendable
      is that there are no isomorphisms that could be pulled forward;
      however, for $k \geq p$, $\fsat k$ \emph{is} an isomorphism (as
      shown in the previous item), so the path cannot end before one
      of these $\fsat k$ has actually been pulled forward.

      This pull-forward step eventually happens at condition
      $\mathcal{C}_{2q}$, being marked by $*$ in the diagram. Both
      cases $2q \le n$ and $2q>n$ are possible. We assume that
      $\mathcal{C}_{2q} = \bigwedge_{i\in I} \forall f_i . \mathcal
      A_i$.

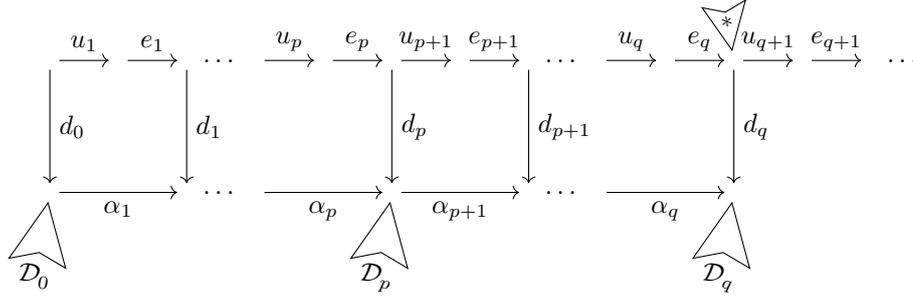
      \begin{tikzpicture}[x=0.90cm]
        \foreach \i in {0,...,13} {
          \node (top\i) at (\i,0) {};
          \node (bot\i) at (\i,-1.75) {};
        }

        \draw[->] (top0) edge node[above]{$u_1$} (top1)
                  (top1) edge node[above]{$e_1$} (top2);
        \node at ($(top2)!0.5!(top3)$) {$\dots$};
        \draw[->] (top3) edge node[above]{$u_p$} (top4)
                  (top4) edge node[above]{$e_p$} (top5)
                  (top5) edge node[above]{$u_{p+1}$} (top6)
                  (top6) edge node[above]{$e_{p+1}$} (top7);
        \node at ($(top7)!0.5!(top8)$) {$\dots$};
        \draw[->] (top8) edge node[above]{$u_q$} (top9)
                  (top9) edge node[above]{$e_q$} (top10);
        \draw[->] (top10) -- node[right]{$\fsat{q}$} (bot10);

        \draw[->] (top10) -- node[above]{$u_{q+1}$} (top11);
        \draw[->] (top11) -- node[above]{$e_{q+1}$} (top12);
        \node at ($(top12)!0.5!(top13)$) {$\dots$};

        \draw[->] (bot0) edge node[below]{$\alpha_1$} (bot2);
        \node at ($(bot2)!0.5!(bot3)$) {$\dots$};
        \draw[->] (bot3) edge node[below]{$\alpha_p$} (bot5);
        \draw[->] (bot5) edge node[below]{$\alpha_{p+1}$} (bot7);
        \node at ($(bot7)!0.5!(bot8)$) {$\dots$};
        \draw[->] (bot8) edge node[below]{$\alpha_{q}$} (bot10);

        \draw[->] (top0) -- node[right]{$\fsat0$} (bot0);
        \draw[->] (top2) -- node[right]{$\fsat1$} (bot2);
        \draw[->] (top5) -- node[right]{$\fsat p$} (bot5);
        \draw[->] (top7) -- node[right]{$\fsat{p+1}$} (bot7);

        \node[condtri,shape border rotate=90, rotate around={-10:(bot0.center)}, label={[rotate=-10,anchor=north,label distance=1pt]below:{$\Asat 0$}}] at (bot0.south) {\kern12pt};
        \node[condtri,shape border rotate=90, rotate around={-10:(bot5.center)}, label={[rotate=-10,anchor=north,label distance=1pt]below:{$\Asat p$}}] at (bot5.south) {\kern12pt};
        \node[condtri,shape border rotate=90, rotate around={-10:(bot10.center)}, label={[rotate=-10,anchor=north,label distance=1pt]below:{$\Asat q$}}] at (bot10.south) {\kern12pt};
        \node[condtri,shape border rotate=270,rotate around={10:(top10.center)},scale=0.9] at (top10.north) {$*$};
      \end{tikzpicture}

      Hence at step $q$ a successor $\forall \fsat{q}.\Asat{q}$ of
      $\forall \fsat p . \Asat p$ is finally pulled forward, i.e., the
      next label $u_{q+1}$ equals $\fsat q$. Assume that
      $\forall \fsat{q}.\Asat{q} = \forall
      f_\ell.\mathcal{A}_\ell$. Let
      $\Asat q = \bigvee_j \exists h_j \mathcal B_j$, then the next
      step is:

      \[
        \Big( \bigwedge_{\kern10pt\mathclap{m \ne \ell}\kern10pt}
        \forall f_m . \mathcal A_m \Big)
        \land \forall \fsat q . \Asat q
      \xrightarrow[= u_{q+1}]{\fsat q}
      \overbrace{
        \underbrace{\bigvee_j \exists h_j . \Big(
          \mathcal B_j}_{\Asat q} \land
          \Big( \bigwedge_{\kern10pt\mathclap{m \ne \varsat{q}}\kern10pt}
          \forall f_m . \mathcal A_m \Big)%
          _{\downarrow \fsat q ; h_j}
        \Big)
      }^{\mathcal E}
      \]

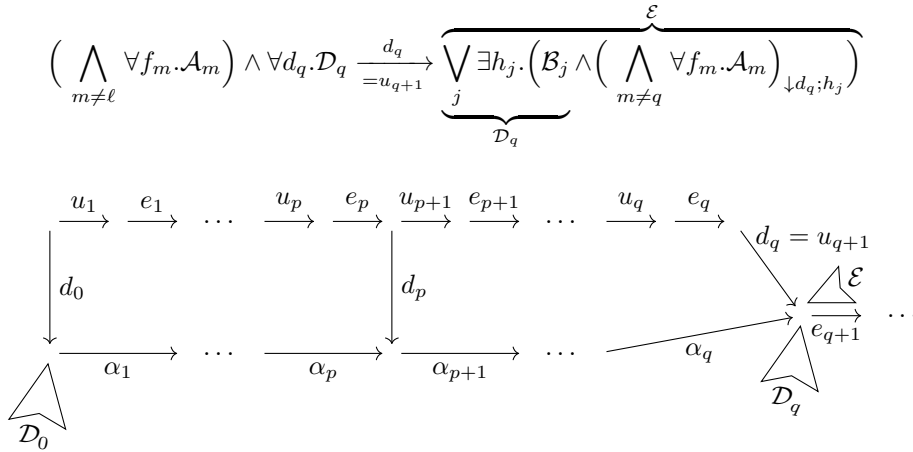
      \begin{tikzpicture}[x=0.90cm]
        \foreach \i in {0,...,10} {
          \node (top\i) at (\i,0) {};
          \node (bot\i) at (\i,-1.75) {};
        }
        \foreach \i in {11,12,13} {
          \node (tbj\i) at (\i,-1.25) {};
        }

        \draw[->] (top0) edge node[above]{$u_1$} (top1)
                  (top1) edge node[above]{$e_1$} (top2);
        \node at ($(top2)!0.5!(top3)$) {$\dots$};
        \draw[->] (top3) edge node[above]{$u_p$} (top4)
                  (top4) edge node[above]{$e_p$} (top5)
                  (top5) edge node[above]{$u_{p+1}$} (top6)
                  (top6) edge node[above]{$e_{p+1}$} (top7);
        \node at ($(top7)!0.5!(top8)$) {$\dots$};
        \draw[->] (top8) edge node[above]{$u_q$} (top9)
                  (top9) edge node[above]{$e_q$} (top10);
        \draw[->] (top10) -- node[right,pos=0.1]{$\fsat q = u_{q+1}$} (tbj11);

        \draw[->] (tbj11) -- node[below]{$e_{q+1}$} (tbj12);
        \node at ($(tbj12)!0.5!(tbj13)$) {$\dots$};

        \draw[->] (bot0) edge node[below]{$\alpha_1$} (bot2);
        \node at ($(bot2)!0.5!(bot3)$) {$\dots$};
        \draw[->] (bot3) edge node[below]{$\alpha_p$} (bot5);
        \draw[->] (bot5) edge node[below]{$\alpha_{p+1}$} (bot7);
        \node at ($(bot7)!0.5!(bot8)$) {$\dots$};
        \draw[->] (bot8) edge node[below]{$\alpha_q$} (tbj11);

        \draw[->] (top0) -- node[right]{$\fsat0$} (bot0);
        \draw[->] (top5) -- node[right]{$\fsat p$} (bot5);

        \node[condtri,shape border rotate=90, rotate around={-10:(bot0.center)}, label={[rotate=-10,anchor=north,label distance=1pt]below:{$\Asat 0$}}] at (bot0.south) {\kern12pt};
        \node[condtri,shape border rotate=90, rotate around={-10:(tbj11.center)}, label={[rotate=-10,anchor=north,label distance=1pt]below:{$\Asat q$}}] at (tbj11.south) {\kern12pt};
        \node[condtri,shape border rotate=180, rotate around={25:(tbj11.center)}, label={[rotate=0,anchor=west,label distance=-1pt]35:{$\mathcal E$}}] at (tbj11.north east) {\kern7pt};
      \end{tikzpicture}

      As a result, we have: $( [ e_{q+1},\ u_{q+2}, e_{q+2}, \dots ], \mathcal{E}) \in P$.

      Note that the subsequence $u_{q+2}, e_{q+2}, \dots$ may or may
      not be empty.  (However, $e_{q+1}$ exists because the path
      cannot end at an empty existential.)  Furthermore
      $d_q;g_q = u_{q+1};e_{q+1};\dots;c_n = d_q;e_{q+1};\dots;c_n$,
      which implies that $g_q = e_{q+1};\dots;c_n$, since $d_q$ is an
      iso. % L: I got confused by this, so don't remove

  \end{enumerate}

  We now apply our up-to techniques to obtain the originally desired
  result (given $\forall \fsat0.\Asat0$, $n$ and $\gsat0$ such that
  $c_1 ; \dots ; c_n = \fsat0 ; \gsat0$, show that
  $([ \gsat0, c_{n+1}, \dots ], \Asat0) \in u(P)$).

%      Let $\mathcal E$ be the condition that was last reached on the path of $P$.
      Observe that the root objects of $\Asat q$ and $\mathcal E$ \coincide.
      They are not the same condition,
      however, $\Asat q$ is ``contained'' in $\mathcal E$,
      in the sense that the path condition $\mathcal E$ is a (nested) conjunction of $\Asat q$ and additional conditions.
      Hence we first apply up-to conjunction removal (\Cref{thm:upto-conjunction-removal}) such that our condition is not $\mathcal E$ (``$\Asat q$ with conjunctions''), but only $\Asat q$:

      \[ \implies
      \Big( [ e_{q+1}, u_{q+1}, e_{q+2}, \dots ],
      \bigvee_j \exists h_j . \mathcal B_j
      \Big) =
      \Big( [ e_{q+1}, u_{q+2}, e_{q+2}, \dots ], \Asat q \Big) \in u(P) \]

      Since $\Asat q = (\Asat{q-1})_{\downarrow \alpha_q} = \ldots = (\Asat0)_{\downarrow \alpha_1 \downarrow \dots \downarrow \alpha_q}$,
      we can now apply a series of up-to shift operations (\Cref{thm:compat-uptoshift}) to obtain a tuple with $\Asat0$ as the condition:

      \begin{align*}
        ([ e_{q+1}, u_{q+2}, e_{q+2}, \dots ], &\ {\Asat0}_{\downarrow \alpha_1 \downarrow \dots \downarrow \alpha_q} ) \in u(P) \\
        \implies
        ([ \alpha_q ; e_{q+1}, u_{q+2}, e_{q+2}, \dots ], &\ {\Asat0}_{\downarrow \alpha_1 \downarrow \dots \downarrow \alpha_{q-1}} ) \in u(P) \\
        \implies \dots \implies
        ([ \alpha_1 ; \dots ; \alpha_q ; e_{q+1}, u_{q+2}, e_{q+2}, \dots ], &\ \Asat0) \in u(P)
      \end{align*}

      Now we use up-to recomposition (\Cref{thm:compat-uptorecomp}) to rewrite the initial part of the chain, depending on where $c_n$ is in the path, relative to $\fsat q$:

      \begin{proofparts}
      \proofPart{$2q\le n$, i.e., $c_n$ is after $\fsat q$}
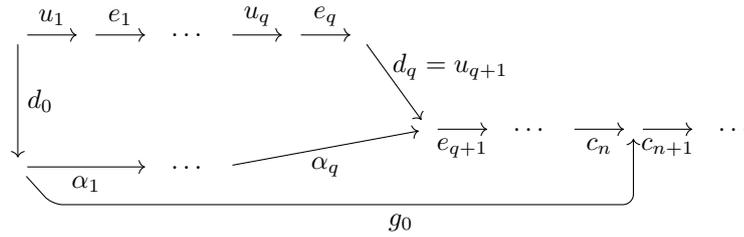
        \begin{tikzpicture}[x=0.90cm]
          \foreach \i in {0,...,5} {
            \node (top\i) at (\i,0) {};
            \node (bot\i) at (\i,-1.75) {};
          }
          \foreach \i in {6,...,11} {
            \node (tbj\i) at (\i,-1.25) {};
          }

          \draw[->] (top0) edge node[above]{$u_1$} (top1)
                    (top1) edge node[above]{$e_1$} (top2);
          \node at ($(top2)!0.5!(top3)$) {$\dots$};
          \draw[->] (top3) edge node[above]{$u_q$} (top4)
                    (top4) edge node[above]{$e_q$} (top5);
          \draw[->] (top5) -- node[right,pos=0.3]{$\fsat q = u_{q+1}$} (tbj6);

          \draw[->] (tbj6) -- node[below]{$e_{q+1}$} (tbj7);
          \node at ($(tbj7)!0.5!(tbj8)$) {$\dots$};
          \draw[->] (tbj8) -- node[below]{$c_n$} (tbj9);
          \draw[->] (tbj9) -- node[below]{$c_{n+1}$} (tbj10);
          \node at ($(tbj10)!0.5!(tbj11)$) {$\dots$};

          \draw[->] (bot0) edge node[below]{$\alpha_1$} (bot2);
          \node at ($(bot2)!0.5!(bot3)$) {$\dots$};
          \draw[->] (bot3) edge node[below]{$\alpha_q$} (tbj6);

          \draw[->] (top0) -- node[right]{$\fsat0$} (bot0);

          \draw[->,rounded corners=4pt] (bot0) -- +(0.5,-0.5) -| node[below,pos=0.3]{$\gsat 0$} (tbj9);
        \end{tikzpicture}

        As indicated in the diagram, we have constructed arrows in
        such a way that
        \[
          \gsat0 = \alpha_1;g_1 = \dots = \alpha_1;\dots;\alpha_q;g_q
          = \alpha_1 ; \dots ; \alpha_q ; e_{q+1} ; u_{q+2} ;
          e_{q+2} ; \dots ; c_n.
        \]
        Hence
        \begin{align*}
        ([ \alpha_1 ; \dots ; \alpha_q ; e_{q+1}, u_{q+2}, e_{q+2}, \dots, c_n, \ c_{n+1}, \dots ], &\ \Asat0) \in u(P) \\
        \implies
        ([ \gsat0, \ c_{n+1}, \dots ], &\ \Asat0) \in u(P)
        \end{align*}
      \proofPart{$2q > n$, i.e., $c_n$ is before $\fsat q$}
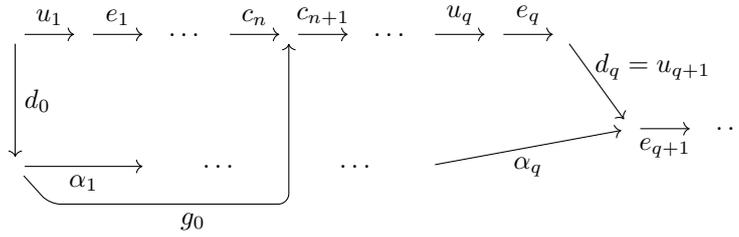
        \begin{tikzpicture}[x=0.90cm]
          \foreach \i in {0,...,8} {
            \node (top\i) at (\i,0) {};
            \node (bot\i) at (\i,-1.75) {};
          }
          \foreach \i in {9,10,11} {
            \node (tbj\i) at (\i,-1.25) {};
          }

          \draw[->] (top0) edge node[above]{$u_1$} (top1)
                    (top1) edge node[above]{$e_1$} (top2);
          \node at ($(top2)!0.5!(top3)$) {$\dots$};
          \draw[->] (top3) edge node[above]{$c_n$} (top4)
                    (top4) edge node[above]{$c_{n+1}$} (top5);
          \node at ($(top5)!0.5!(top6)$) {$\dots$};
          \draw[->] (top6) edge node[above]{$u_q$} (top7)
                    (top7) edge node[above]{$e_q$} (top8);
          \draw[->] (top8) -- node[right,pos=0.3]{$\fsat q = u_{q+1}$} (tbj9);

          \draw[->] (tbj9) -- node[below]{$e_{q+1}$} (tbj10);
          \node at ($(tbj10)!0.5!(tbj11)$) {$\dots$};

          \draw[->] (bot0) edge node[below]{$\alpha_1$} (bot2);
          \node at ($(bot2)!0.5!(bot4)$) {$\dots$};
          \node at ($(bot4)!0.5!(bot6)$) {$\dots$};
          \draw[->] (bot6) edge node[below]{$\alpha_q$} (tbj9);

          \draw[->] (top0) -- node[right]{$\fsat0$} (bot0);

          \draw[->,rounded corners=4pt] (bot0) -- +(0.5,-0.5) -| node[below,pos=0.3]{$\gsat 0$} (top4);
        \end{tikzpicture}

        As indicated in the diagram, we can show that
        $\gsat0 ; c_{n+1} ; \dots ; u_{q+1} = \alpha_1 ; \dots ;
        \alpha_q$. In order to prove this equality, we have to make a
        case distinction:
        \begin{itemize}
        \item Whenever $c_n$ originates from a universal step, that is
          $c_n = u_{k+1}$, we know that $c_{n+1}$ stems from an
          existential step, that is $c_{n+1} = e_{k+1}$. In this case
          we use the fact that the $d_\ell$ for $\ell\ge k+1$ are isos:
          \begin{align*}
            \gsat0 ; c_{n+1} ; \dots ; u_{q+1} &=
            \alpha_0;\dots;\alpha_k;g_k;c_{n+1} ; \dots ; u_{q+1} \\
            &=
            \alpha_0;\dots;\alpha_k;\alpha_{k+1};d_{k+1}^{-1};c_{n+2};
            \dots;u_{q+1} \\
            &=
            \alpha_0;\dots;\alpha_k;\alpha_{k+1};\alpha_{k+2};d_{k+2}^{-1};
            c_{n+4}; \dots;u_{q+1} = \dots \\
            &=
            \alpha_1;\dots;\alpha_q;d_q^{-1};u_{q+1} \\
            &=
            \alpha_1;\dots;\alpha_q;d_q^{-1};d_q = \alpha_1;\dots;\alpha_q
          \end{align*}
        \item Whenever $c_n$ originates from an existential step, that
          is $c_n = e_{k+1}$, we know that $c_{n+1}$ stems from a
          universal step, that is $c_n = u_{k+2}$. In this case we
          can infer from $\alpha_{k+1};g_{k+1} = g_k$ that
          $\alpha_{k+1} = \alpha_{k+1};\id =
          \alpha_{k+1};g_{k+1};d_{k+1} = g_k;d_{k+1}$. This implies:
          \begin{align*}
           \gsat0 ; c_{n+1} ; \dots ; u_{q+1} &=
            \alpha_0;\dots;\alpha_k;g_k;c_{n+1} ; \dots ; u_{q+1} \\
            &=
            \alpha_0;\dots;\alpha_k;\alpha_{k+1};d_{k+1}^{-1};c_{n+1};
            \dots;u_{q+1} \\
            &=
            \alpha_0;\dots;\alpha_k;\alpha_{k+1};\alpha_{k+2};d_{k+2}^{-1};
            c_{n+3}; \dots;u_{q+1} = \dots \\
            &=
            \alpha_1;\dots;\alpha_q;d_q^{-1};u_{q+1} \\
            &=
            \alpha_1;\dots;\alpha_q;d_q^{-1};d_q = \alpha_1;\dots;\alpha_q
          \end{align*}
        \end{itemize}
        Hence
        \begin{align*}
        ([ \alpha_1 ; \dots ; \alpha_q ; \ e_{q+1}, u_{q+2}, e_{q+2}, \dots ], &\ \Asat0) \in u(P) \\
        \implies
        ([ \gsat0, c_{n+1}, \dots, u_{q+1}, \ e_{q+1}, u_{q+2}, e_{q+2}, \dots ], &\ \Asat0) \in u(P)
        \end{align*}
      \end{proofparts}
      In both cases we have achieved the desired result. \qedhere
\end{proofparts}
\end{proof}

\subsection{Soundness and Completeness}

We are finally ready to show soundness and completeness of our
method.

As a condition is essentially equivalent to any of its tableaux, which
break it down into existential subconditions, a closed tableau
represents an unsatisfiable condition.

\begin{theoremrep}[Soundness]\label{thm-algo-sound}
  If there exists a tableau $T$ for a condition $\mathcal A$ where all
  branches are closed, then the condition $\mathcal A$ in the root
  node is unsatisfiable.
\end{theoremrep}

%\begin{toappendix}
\begin{proofsketch}
  By induction over the depth of $T$.  Base case is $\condfalse$
  (obviously unsatisfiable).  Induction step for $\exists$:
  $\bigvee_i \exists f_i.\mathcal A_i$ is unsatisfiable if all
  $\mathcal A_i$ are.  For $\forall$: by construction, the only child
  contains an equivalent condition.
\end{proofsketch}
%\end{toappendix}

\begin{proof}
  By induction over the depth of $T$.
  \begin{itemize}
  \item Let the depth of $T$ be $1$, i.e., $T$ consists of exactly one
    node (the root $\mathcal A$).  Since all branches (in this case,
    the only branch) of $T$ are closed and therefore end with an empty
    existential, this node itself, which is also the root node, is an
    empty existential.  An empty existential condition is
    unsatisfiable. Therefore, the statement holds.
    \item Assume the statement holds for tableaux of depth $d$.  Let
      $T$ be a tableau with depth $d+1$.  $T$ has several children as
      created by one of the tableau rules ($\forall$ or $\exists$).
      Each child can be seen as the root of a subtableau $T_i$.  As
      all branches of $T$ are closed, so are all branches of all
      $T_i$, and each $T_i$ has a depth of at most $d$.  Therefore,
      the root $\mathcal A_i$ of each subtableau $T_i$ is
      unsatisfiable.  Now consider the root of $T$:
      \begin{itemize}
      \item Root of $T$ is existential
        ($\mathcal A = \biglor_i \exists f_i.\mathcal A_i$): The
        children of the root of $T$ are, by construction, exactly the
        children of $\mathcal A$.  By definition of satisfaction,
        $\mathcal A$ is satisfiable if and only if at least one
        $\mathcal A_i$ is satisfiable.  However, all $\mathcal A_i$,
        being the roots of closed tableaux $T_i$ of depth $d$, are
        known to be unsatisfiable.  Therefore, $\mathcal A$ is also
        unsatisfiable.
      \item Root of $T$ is universal
        ($\mathcal A = \bigland_i \forall f_i.\mathcal A_i$): Since
        $T$ has depth $d+1$, the root has a single child node $\mathcal A'$
        that was created by pulling forward some isomorphism $f_p$
        from $\mathcal A$.  $\mathcal A'$ is the root of a closed
        subtableau of depth $d$ and therefore unsatisfiable.  Then,
        $\exists f_p.\mathcal A'$ is also unsatisfiable, which by
        \Cref{lem:pull-forward-isos} is equivalent to $\mathcal A$ and
        therefore also unsatisfiable.
        \qedhere
      \end{itemize}
  \end{itemize}
\end{proof}

We now prove completeness, which -- to a large extent -- is a
corollary of \Cref{thm:fair-branch-model}.

\begin{theorem}[Completeness]\label{thm-algo-complete}
  If a condition $\mathcal A$ is unsatisfiable, then every tableau
  constructed by obeying the fairness constraint is a
  finite tableau where all branches are closed. Furthermore,
  at least one such tableau exists.
  % If a condition $\mathcal A$ is unsatisfiable, then there exists a
  % finite tableau with root $\mathcal A$ where all branches are
  % closed.
\end{theorem}

\begin{proof}
  The contraposition follows from \Cref{thm:fair-branch-model}: If the
  constructed tableau is finite with open branches or infinite, then
  $\mathcal A$ is satisfiable. Furthermore, a fair tableau must exist
  and can be constructed by following the strategy described in
  \Cref{rem:fairness-strategies}.
  % We show the contraposition: If there exists \emph{no} tableau for
  % $\mathcal A$ where all branches are closed, then $\mathcal A$ is
  % satisfiable. --
  % If no such tableau exists, then every tableau has at least one open
  % branch.  Now consider only tableaux that can either no longer be
  % extended, or have an infinite fair branch (such tableaux must exist
  % by \Cref{rem:fairness-strategies}).  Then, by
  % \Cref{thm:fair-branch-model}, every open branch in a tableau
  % represents a model for $\mathcal A$.
\end{proof}

In the next section we will show how the open branches in a fully
expanded tableau can be interpreted as models, thus giving us a
procedure for model finding. 

\begin{example}[Proving unsatisfiability]\label{exa-prove-unsat}
  \newcommand{\Mor}[1]{[\mkern3mu\fcGraph{n1}{#1}\mkern3mu]}
  \newcommand{\Morsub}[1]{[\mkern3mu\scalebox{0.8}{\fcGraph{n1}{#1}}\mkern3mu]}
  \tikzset{
    old/.style={draw=black!38, color=black!38, line width=0.4pt},
    new/.style={line width=0.7pt},
  }
  % draw commands for the various graphs that are used in this example
  \newcommand{\gno}[2]{\node[gninlinable,old] (n#1) at (#2,0) {#1};}
  \newcommand{\gnn}[2]{\node[gninlinable,new] (n#1) at (#2,0) {#1};}
  \newcommand{\gnoZ}{\node[gninlinable,old] (n0) at (0,0) {};}
  \newcommand{\gnnZ}{\node[gninlinable,new] (n0) at (0,0) {};}
  \newcommand{\gnoP}[1]{\node[gninlinable,old] (nP) at (#1,0) {$\mathclap+$};}
  \newcommand{\gnnP}[1]{\node[gninlinable,new] (nP) at (#1,0) {$\mathclap+$};}
  \newcommand{\gnoX}[1]{\node[gninlinable,old] (nX) at (#1,0) {$\mathclap\times$};}
  \newcommand{\gnnX}[1]{\node[gninlinable,new] (nX) at (#1,0) {$\mathclap{\boldsymbol\times}$};}
  \newcommand{\gnoA}[1]{\node[gninlinable,old] (nA) at (#1,0) {A};}
  \newcommand{\gnnA}[1]{\node[gninlinable,new] (nA) at (#1,0) {A};}
  \newcommand{\gnoB}[1]{\node[gninlinable,old] (nB) at (#1,0) {B};}
  \newcommand{\gnnB}[1]{\node[gninlinable,new] (nB) at (#1,0) {B};}
  \newcommand{\gnnL}[1]{\node[gninlinable,new] (nL) at (#1,0) {};}
  \newcommand{\geo}[2]{\draw[gedge,old] (#1) to (#2);}
  \newcommand{\gen}[2]{\draw[gedge,new] (#1) to (#2);}
  \newcommand{\gebo}[2]{\draw[gedge,old] (#1) to[bend left=10] (#2);}
  \newcommand{\gebn}[2]{\draw[gedge,new] (#1) to[bend left=10] (#2);}
  \newcommand{\gebbn}[2]{\draw[gedge,new] (#1) to[bend left=20] (#2);}
  We work in $\graphfinj$.
  For this example, we use the following shorthand notation for graph morphisms:
  $\Mor{ \gno11 \gnn22 \gen{n1}{n2} }$ means $
    \fcGraph{n1}{\node[gninlinable] (n1) at (0,0) {$1$};}
    \rightarrow \fcGraph{n1}{
      \node[gninlinable] (n1) at (0,0) {$1$};
      \node[gninlinable] (n2) at (1,0) {$2$};
      \draw[gedge] (n1) to (n2);
    }
  $, i.e.,
  the morphism is the inclusion from the light-gray graph elements to the full graph.

  Consider the condition
  $\mathcal A = \forall [\emptyset] .\exists \Mor{\gnn11} .\condtrue \land \forall \Mor{\gnnX1} .\condfalse$,
  meaning (1) there exists a node and (2) no node must exist.
  It is easily seen that these contradict each other and hence $\mathcal A$ is
  unsatisfiable. We obtain a tableau with a single branch for this condition:
  \[
    \mathcal A
    \:\xrightarrow{\smash{[\emptyset]}}\:
    \exists \Mor{\gnn11} .\big( \condtrue \land (\forall \Mor{\gnnX1}.\condfalse)_{\downarrow \Morsub{\gnn11}} \big)
    \:\xrightarrow{\smash{\Morsub{\gnn11}}}\:
    \forall \Mor{\gno11} .\condfalse \land \forall \Mor{\gno11 \gnnX2} .\condfalse
    \:\xrightarrow{\smash{\Morsub{\gno11}}}\:
    \condfalse
  \]
  In the first step, $\mathcal A$ is universal with an isomorphism $\emptyset \to \emptyset$, which is pulled forward.
  Together with its (only) existential child, this results in the partial model 
  $\Mor{\gnn11}$ and another universal condition with the meaning:
  (1)~the just created node $\,\fcGraph{n1}{\gnn11}\,$ must not exist, and
  (2)~no additional node $\fcGraph{nX}{\gnnX1}$ may exist either.

  \pagebreak
  This condition includes another isomorphism ($\Mor{\gno11}$) to be
  pulled forward.  Its child
  \mbox{$\mathcal A_p = \biglor_{j \in J} \exists g_j.\mathcal B_j =
    \condfalse$} is an empty disjunction, so the tableau rule for
  universals adds an empty disjunction ($\condfalse$) as the only
  descendant.  This closes the (only) branch, hence the initial
  condition $\mathcal A$ can be recognized as unsatisfiable.
\end{example}

\subsection{Model Finding}

We will now discuss the fact that the calculus not only searches for a
logical contradiction to show unsatisfiability, but at the same time
tries to generate a (possibly infinite) model.  We can show that
\emph{every} finitely decomposable (i.e., ``finite'') model (or a
prefix thereof) can be found after finitely many steps in a
\emph{fully expanded} tableau, i.e., a tableau where all branches are
extended whenever possible, including infinite branches.  This is a
feature that distinguishes it from other known calculi for first-order
logic.

The following lemma shows that an infinite branch always makes
progress towards approximating the infinite model.

\begin{lemmarep}\label{tableau-nonisos-every-now-and-then}
  Let $\mathcal A$ be a condition and $T$ be a fully
  expanded tableau for
  $\mathcal A$.  Then, for each branch it holds that it either is
  finite, or that there always eventually is another
  \mbox{non-isomorphism} on the branch.
\end{lemmarep}

%\begin{toappendix}
\begin{proofsketch}
  We define the size of a condition and show that it decreases if an
  iso occurs on a path. This means that eventually there will
  always occur another non-iso on a path.
%  
  % We show that it is impossible for a path to begin with infinitely
  % many isos.  Assume we are starting with a universal condition
  % $\mathcal A$, which results in some iso
  % $\forall f_p.\mathcal A_p$ being pulled forward.  Consider the
  % children of the resulting existential condition
  % $\biglor_j \exists g_j.\mathcal C_j$.  Each child where $g_j$ is a
  % non-iso immediately validates the statement.
%
  % For children where $g_j$ is an iso, $\mathcal C_j$ contains
  % the result of shifting over an iso $f_p;g_j$.  With the
  % assumptions on $\kappa$ made at the end of
  % \Cref{sec:representative-squares-shift}, such a shift preserves the
  % structure of the condition and hence its ``size''.  Then show that
  % $w(\mathcal A) > w(\mathcal C_j)$ for some weight function $w$ that
  % determines the size of a condition.  As each such step reduces the
  % weight and it cannot be reduced forever, there cannot be
  % infinitely many leading isos on the branch.
\end{proofsketch}
%\end{toappendix}

\begin{proof}
  We show that every iso on the branch is followed by finitely many
  further isos and eventually by a non-iso, or the branch ends.
  Equivalently: for each subtableau $T'$ of $T$, every branch is
  finite, or begins with finitely many isos followed by a non-iso.  We
  do this by showing that it is impossible for a path to begin with
  infinitely many isos.

  Assume we are starting with a universal condition $\mathcal A$
  (for an existential condition, simply proceed with the universal conditions in its children).

  $\mathcal A$ either has no isos as children, resulting in no
  pull-forward step and hence a finite branch of length zero, which
  validates the statement.  Otherwise, assume some
  $\forall f_p.\mathcal A_p$ was pulled forward, resulting in a single
  existential child $\biglor_j \exists g_j.\mathcal C_j$.

  Of this existential condition, every child where $g_j$ is a
  non-isomorphism immediately validates the statement, as, starting
  from the root, only one iso ($f_p$) precedes a non-iso $g_j$.

  Consider a child where $g_j$ is an isomorphism, and the next node on
  the branch is hence a universal condition $\mathcal C_j$.  Being a
  result of a pull-forward step, its shape is
  $\mathcal C_j = \mathcal B_j \land \bigland_{m \ne p} (\forall
  f_m.\mathcal A_m)_{\downarrow f_p;g_j}$, where
  $\mathcal{A}_p = \biglor_{\ell} \exists g_\ell.\mathcal{B}_\ell$.
  
  We will compare the weights (respectively size) of conditions
  $\mathcal A$ and $\mathcal C_j$, by defining a weight function $w$
  on conditions as follows:
  \[\textstyle
    w\big( \biglor_i \exists f_i.\mathcal A_i \big)
    = w\big( \bigland_i \forall f_i.\mathcal A_i \big)
    = 1 + \sum_i w(\mathcal A_i)
  \]

  (Note that the base cases are conditions with no children with
  $w(\condtrue) = w(\condfalse) = 1$ according to the definition
  above.)

  Assume
  $\mathcal A = \forall f_p.\left( \biglor_{j'} \exists
    g_{j'}.\mathcal B_{j'} \right) \land \bigland_{m \ne p} \forall
  f_m.\mathcal A_m$.  Let $w_p = \sum_{m \ne p} w(\mathcal A_m)$ be
  the combined weight of the children of $\mathcal A$ that were
  \emph{not} pulled forward.  Also, for the condition $B_j$ which is
  part of the conjunction $\mathcal{C}_j$, assume
  that $\mathcal B_j = \bigland_k \forall h_k.\mathcal D_k$.  Then:
  \begin{align*}
    w(\mathcal A)
    &=        \textstyle 1 + w\left( \biglor_{j'} \exists g_{j'}.\mathcal B_{j'} \right) + w_p \\
    &=        \textstyle 1 + 1 + \left( \sum_{j'} w(\mathcal B_{j'}) \right) + w_p    &&\text{(includes weights of all $\exists$ children)} \\
    &\geq     \textstyle 1 + w(\mathcal B_j) + w_p   &&\text{(only the child on the current branch)} \\
    &= \textstyle 1 + \Big( 1 + \sum_k w(\mathcal D_k) \Big) + w_p
  \end{align*}

  Now compute the weight of $\mathcal C_j$:
  \begin{align*}
    w(\mathcal C_j)
    &= \textstyle w\left( \mathcal B_j \land \bigland_{m \ne p} (\forall f_m.\mathcal A_m)_{\downarrow f_p;g_j} \right) \\
    &= \textstyle w\left( \bigland_k \forall h_k.\mathcal D_k \land \bigland_{m \ne p} (\forall f_m.\mathcal A_m)_{\downarrow f_p;g_j} \right) \\
    &=     \textstyle 1 + \Big( \sum_k w(\mathcal D_k) \Big) + \sum_{m \ne p} w\left( {\mathcal A_m}_{\downarrow f_p;g_j} \right) \\
    &\stackrel{\mathclap{(*)}}{=} \textstyle 1 + \Big( \sum_k w(\mathcal D_k) \Big) + \sum_{m \ne p} w(\mathcal A_m) \\
    &=     \textstyle 1 + \Big( \sum_k w(\mathcal D_k) \Big) + w_p
  \end{align*}

  The equality marked with $(*)$ holds because $f_p;g_j$ is an
  isomorphism ($f_p$ because only isos can be pulled forward), and for
  an iso $i$, $w(\mathcal C_{\downarrow i}) = w(\mathcal C)$. Since we
  previously assumed that $\kappa$ has only a single representative
  square when isos are involved, then, by the assumption made at the
  end of \Cref{sec:representative-squares-shift}, when shifting over
  an isomorphism, the resulting condition has the same structure as
  the original one (we obtain exactly one square per child condition,
  and hence the new condition has exactly as many children as the
  original one; and its child conditions are shifted over identities).
  Since the weight function considers only child count and nesting
  depth, shifting over isos preserves the weight.  Clearly
  $w(\mathcal A) > w(\mathcal C_j)$ in that case.

  % \[\textstyle\thickmuskip=12mu
  %   w(\mathcal A)
  %   \geq 1 + \Big( 1 + \sum_k w(\mathcal D_k) \Big) + w_p
  %   >    1 + \Big( \sum_k w(\mathcal D_k) \Big) + w_p
  %   =    w(\mathcal C_j)
  % \]

  Hence, if a branch starts at a universal condition, and the next
  arrow on the path is an isomorphism, the weight of the condition is
  reduced by at least $1$.  This excludes having an infinite sequence
  of isomorphisms at the beginning of the branch, since the weight
  cannot be reduced infinitely many times.  Hence every branch
  eventually ends, or a non-iso occurs on it after at most
  $2 \cdot w(\mathcal A)$ leading isomorphisms.
\end{proof}

\begin{theoremrep}[Model Finding]\label{will-find-finite-model}
  Let $\mathcal A$ be a condition, $m$ a finitely decomposable arrow
  such that $m \models \mathcal A$ and let $T$ be a fully expanded
  tableau for $\mathcal A$.

  Then, there exists an open and unextendable branch with arrows
  $c_1,\tdots,c_n$ in $T$, having condition $\mathcal R$ in the leaf
  node, where $m = c_1; \tdots; c_n;r$ for some $r$ with
  $r \models \mathcal R$. Furthermore the finite prefix is itself a
  model for $\mathcal A$ (i.e.,
  $[c_1, \tdots, c_n] \models \mathcal A$).
\end{theoremrep}

\begin{proof}
  First, we observe that by the tableau construction rules of
  \Cref{satcheck.new}, if $m \models \mathcal A$, then the tableau $T$
  for $\mathcal{A}$ contains at least one branch
  $\mathcal A = \mathcal{C}_0 \xrightarrow{c_1} \mathcal C_1 \xrightarrow{c_2} \mathcal C_2 \xrightarrow{c_3} \dots$
  such for each node $\mathcal C_i$ on the branch, it holds that
  $m = c_1 ; \dots ; c_i ; r_i$ and $r_i \models \mathcal C_i$.
  (This is because the tableau contains only equivalence transformations and existential unfolding.
  The existence of at least one such branch is guaranteed for all
  models, whether they are infinite or finite.)

  This can be shown by induction: assume that we have to constructed
  the path up to index $i$. We distinguish the following cases:

  \begin{itemize}
  \item $\mathcal{C}_i = \biglor_{j\in J} \exists f_j.\mathcal{A}_j$
    is existential: since $r \models \mathcal{C}_i$, there
    exists an index $j\in J$ such that $r_i = f_j;r_{i+1}$ for some
    arrow $r_{i+1}$ with $r_{i+1} \models
    \mathcal{A}_j$. Then choose $c_{i+1} = f_j$,
    $\mathcal{C}_{i+1} = \mathcal{A}_j$. Clearly
    $m = c_1;\dots;c_i;r_i = c_1;\dots;c_i;f_j;r_{i+1} =
    c_1;\dots;c_i;c_{i+1};r_{i+1}$.
  \item $\mathcal{C}_i = \bigland_{j\in J} \forall f_j.\mathcal{A}_j$
    is universal: then there is some (non-deterministically chosen)
    tableau step reaching $\mathcal{C}_{i+1}$, where $c_{i+1}$ is an
    iso and $\mathcal{C}_i \equiv \exists c_{i+1}.\mathcal{C}_{i+1}$
    (see \Cref{lem:pull-forward-isos}). Define
    $r_{i+1} = c_{i+1}^{-1};r_i$, which gives us
    $m = c_1;\dots;c_i;r_i = c_1;\dots;c_i;c_{i+1};r_{i+1}$.
  \end{itemize}

  A result of this observation is that every model (or a prefix thereof) actually occurs in the tableau.
  We will now show that such a prefix will occur at the end of a finite branch, and hence be produced as a positive output of the algorithm.

  Consider a branch that has the property described above.
  If there are multiple branches that satisfy this property, pick an arbitrary one.
  The branch could be either finite or infinite:

  \begin{itemize}
  \item Assume the branch is infinite.  By
    \Cref{tableau-nonisos-every-now-and-then}, the branch must contain
    infinitely many non-isomorphisms $c_i$.  However, $m$ is finitely
    decomposable by assumption, so such a decomposition is impossible.

      Thus, any branch with the property given above cannot be infinite.

    \item Assume the branch is finite and ends with condition
      $\mathcal{R} = \mathcal C_\ell$ in the leaf.  Then,
      $\mathcal{R}$ is not an existential condition, since then it
      would be $\condfalse$ and hence has no model.

      Hence $\mathcal{R}$ is a universal condition that has no isomorphisms.
      Then, $\id \models \mathcal{R}$.
      \qedhere
  \end{itemize}
\end{proof}

Note that this finite branch can be found in finite time, assuming a
suitable strategy for exploration of the tableau such as breadth-first
search or parallel processing.

\begingroup% locally define the graph drawing commands for the examples
  \def\gDom#1{%
    \tikz[baseline=(a.base),opacity=0.35]{\node[rectangle,inner sep=0pt] (a) {#1};}%
  }

  % draw commands for the various graphs that are used in this example
  % \nd{1}{1} creates nodes, first argument is internal label and node name, second argument is x position (integer)
  % gAtoB creates a graph (including \fGraph) (1)-->(2)  (O,A,B,... for 0,1,2,... and X,Y for A,B)
  \def\nd#1#2{\node[gninlinable] (n#1) at (#2,0) {#1};}
  \def\ndbl#1#2{\node[gninlinable] (n#1) at (#2,0) {};}
  \def\ndsub#1#2{\node[gninlinablesub] (n#1) at (#2,0) {#1};}
  \def\gE#1#2{\draw[gedge] (#1) to (#2);}
  \def\mutualedges#1#2{\draw[gedge] (#1) to[bend left=10] (#2); \draw[gedge] (#2) to[bend left=10] (#1);}
  \newcommand\gI{             \fcGraph{n1}{\nd10}}
  \newcommand\gA{             \fcGraph{nA}{\nd A0}}
  \newcommand\gAtoB{          \fcGraph{nA}{\nd A0 \nd B1               \gE{nA}{nB}}}
  \newcommand\gIsub{          \fcGraph{n1}{\ndsub10}}
  \newcommand\gOtoI{          \fcGraph{n1}{\nd10 \ndbl0{-1}            \gE{n0}{n1}}}
  \newcommand\gItoII{         \fcGraph{n1}{\nd10 \nd21                 \gE{n1}{n2}}}
  \newcommand\gItoIIsub{      \fcGraph{n1}{\ndsub10 \ndsub21           \gE{n1}{n2}}}
  \newcommand\gItoAsub{       \fcGraph{n1}{\ndsub10 \ndsub A1          \gE{n1}{nA}}}
  \newcommand\gIA{            \fcGraph{n1}{\nd10 \nd A{0.75}}}
  \newcommand\gIAsub{         \fcGraph{n1}{\ndsub10 \ndsub A{0.75}}}
  \newcommand\gIAtoB{         \fcGraph{n1}{\nd10 \nd A{0.75} \nd B{1.75}         \gE{nA}{nB}}}
  \newcommand\gIotA{          \fcGraph{n1}{\nd10 \nd A1                \gE{nA}{n1}}}
  \newcommand\gOtoItoII{      \fcGraph{n1}{\ndbl0{-1} \nd10 \nd21      \gE{n0}{n1} \gE{n1}{n2}}}
  \newcommand\gImutII{        \fcGraph{n1}{\nd10 \nd21                 \mutualedges{n1}{n2} }}
  \newcommand\gImutIIsub{     \fcGraph{n1}{\ndsub10 \ndsub21           \mutualedges{n1}{n2} }}
  \newcommand\gItoIIA{        \fcGraph{n1}{\nd10 \nd21 \nd A{1.75}          \gE{n1}{n2}}}
  \newcommand\gItoIIAsub{     \fcGraph{n1}{\ndsub10 \ndsub21 \ndsub A{1.75} \gE{n1}{n2}}}
  \newcommand\gItoIItoIII{    \fcGraph{n1}{\nd10 \nd21 \nd32           \gE{n1}{n2} \gE{n2}{n3} }}
  \newcommand\gItoIItoIIIsub{ \fcGraph{n1}{\ndsub10 \ndsub21 \ndsub32  \gE{n1}{n2} \gE{n2}{n3} }}
  \newcommand\gItoIIAtoB{     \fcGraph{n1}{\nd10 \nd21 \nd A2 \nd B3   \gE{n1}{n2} \gE{nA}{nB} }}
  \newcommand\gItoIIotA{      \fcGraph{n1}{\nd10 \nd21 \nd A2          \gE{n1}{n2} \gE{nA}{n2} }}
  \newcommand\gItoIIAtoI{     \fcGraph{n1}{\nd10 \nd21 \nd A2          \gE{n1}{n2} \draw[gedge] (nA) to[bend left=20] (n1); }}
  \newcommand\gOtoItoIItoIII{ \fcGraph{n1}{\ndbl0{-1} \nd10 \nd21 \nd32\gE{n0}{n1} \gE{n1}{n2} \gE{n2}{n3} }}
  \newcommand\gItoIItoIIItoI{ \fcGraph{n1}{\nd10 \nd21 \nd32           \gE{n1}{n2} \gE{n2}{n3} \draw[gedge] (n3) to[bend left=20] (n1); }}
  \newcommand\gImutIItoIII{   \fcGraph{n1}{\nd10 \nd21 \nd32           \mutualedges{n1}{n2} \gE{n2}{n3} }}
  \newcommand\gItoIItoIIItoIV{\fcGraph{n1}{\nd10 \nd21 \nd32 \nd43     \gE{n1}{n2} \gE{n2}{n3} \gE{n3}{n4} }}
  \newcommand\gItoIImutIII{   \fcGraph{n1}{\nd10 \nd21 \nd32           \mutualedges{n2}{n3}  \gE{n1}{n2} }}
  \newcommand\gItoIItoIIIA{   \fcGraph{n1}{\nd10 \nd21 \nd32 \nd A{2.75} \gE{n1}{n2} \gE{n2}{n3} }}
  \newcommand\gOtoItoIImutIII{\fcGraph{n1}{\ndbl0{-1} \nd10 \nd21 \nd32\mutualedges{n2}{n3}   \gE{n1}{n2} \gE{n0}{n1} }}
  \newcommand\gItoIImutIIItoI{\fcGraph{n1}{\nd10 \nd21 \nd32           \mutualedges{n2}{n3}   \gE{n1}{n2} \draw[gedge] (n3) to[bend left=20] (n1); }}
  \newcommand\gImutIImutIII{  \fcGraph{n1}{\nd10 \nd21 \nd32           \mutualedges{n2}{n3}   \mutualedges{n2}{n1} }}
  \newcommand\gItoIImutIIIA{  \fcGraph{n1}{\nd10 \nd21 \nd32 \nd A{2.75} \mutualedges{n2}{n3}  \gE{n1}{n2} }}
  \newcommand\gImutIIA{       \fcGraph{n1}{\nd10 \nd21 \nd A{1.75}     \mutualedges{n1}{n2} }}

% (This is supposed to cover the case ``only infinite models = algorithm
% does not terminate, but produces an infinite model in the limit'')

\begin{algorithm}[Satisfiability Check]
  \label{alg:sat-check}
  Given a condition $\mathcal{A}$, we define the procedure
  $\mathit{SAT}(\mathcal{A})$ that may either produce a model
  $c\colon \RO(\mathcal A) \to C$, answer $\mathrm{unsat}$ or does not
  terminate.
  \begin{itemize}
    \item Initialize the tableau with $\mathcal A$ in the root node.
    \item While the tableau still has open branches:
    \begin{itemize}
    \item Select one of the open branches as the current branch, using
      an appropriate strategy that extends each open branch eventually.
    \item If the leaf is a universal condition without isomorphisms,
      terminate and return the labels of the current branch as model.
    \item Otherwise, extend the branch according to the rules of
      \Cref{satcheck.new}, obeying the fairness constraint.
    \end{itemize}
  \item If all branches are closed, terminate and answer $\mathrm{unsat}$.
  \end{itemize}
\end{algorithm}

\noindent%
This procedure has some similarities to the tableau-based reasoning
from \cite{lo:tableau-graph-properties}. The aspect of model
generation was in particular considered in
\cite{slo:model-generation}.
Overall, we obtain the following result:

\begin{theorem}
  There is a one-to-one correspondence between satisfiability of a
  condition ($\mathcal A$ unsatisfiable; $\mathcal A$ has a finitely
  decomposable model; $\mathcal A$ is satisfiable, but has no finitely
  decomposable model) and the output of \Cref{alg:sat-check}
  ($\mathit{SAT}(\mathcal{A})$) (terminates with $\mathrm{unsat}$,
  terminates with a model, does not terminate).
\end{theorem}
\begin{proof}
  \mbox{}
  \begin{itemize}
  \item $\mathcal A$ unsatisfiable $\iff$ algorithm outputs unsat:
    $(\Rightarrow)$ \Cref{thm-algo-complete}, $(\Leftarrow)$
    \Cref{thm-algo-sound}
  \item $\mathcal A$ has a finitely decomposable model $\iff$
    algorithm finds finitely decomposable model: $(\Rightarrow)$
    \Cref{will-find-finite-model}, $(\Leftarrow)$
    \Cref{thm:fair-branch-model}
  \item $\mathcal A$ has only models that are not finitely
    decomposable $\iff$ algorithm does not
    terminate: \\
    $(\Rightarrow)$ exclusion of other possibilities for
    non-termination,
    \Cref{tableau-nonisos-every-now-and-then} for model in the limit \\
    $(\Leftarrow)$ \Cref{thm:fair-branch-model}
    % if $\mathcal A$ had both finite and infinite models, then the case
    % ``$\mathcal A$ finitely satisfiable'' would guarantee termination
    % instead
    \qedhere
  \end{itemize}
\end{proof}

\begin{example}[Finding finite models]\label{exa-find-finite-models}
  \newcommand{\Mor}[1]{[\mkern3mu\fcGraph{n1}{#1}\mkern3mu]}
  \newcommand{\Morsub}[1]{[\mkern3mu\scalebox{0.8}{\fcGraph{n1}{#1}}\mkern3mu]}
  \tikzset{
    old/.style={draw=black!38, color=black!38, line width=0.4pt},
    new/.style={line width=0.7pt},
  }
  % draw commands for the various graphs that are used in this example
  \newcommand{\gno}[2]{\node[gninlinable,old] (n#1) at (#2,0) {#1};}
  \newcommand{\gnn}[2]{\node[gninlinable,new] (n#1) at (#2,0) {#1};}
  \newcommand{\gnoZ}{\node[gninlinable,old] (n0) at (0,0) {};}
  \newcommand{\gnnZ}{\node[gninlinable,new] (n0) at (0,0) {};}
  \newcommand{\gnoP}[1]{\node[gninlinable,old] (nP) at (#1,0) {$\mathclap+$};}
  \newcommand{\gnnP}[1]{\node[gninlinable,new] (nP) at (#1,0) {$\mathclap+$};}
  \newcommand{\gnoX}[1]{\node[gninlinable,old] (nX) at (#1,0) {$\mathclap\times$};}
  \newcommand{\gnnX}[1]{\node[gninlinable,new] (nX) at (#1,0) {$\mathclap{\boldsymbol\times}$};}
  \newcommand{\gnoA}[1]{\node[gninlinable,old] (nA) at (#1,0) {A};}
  \newcommand{\gnnA}[1]{\node[gninlinable,new] (nA) at (#1,0) {A};}
  \newcommand{\gnoB}[1]{\node[gninlinable,old] (nB) at (#1,0) {B};}
  \newcommand{\gnnB}[1]{\node[gninlinable,new] (nB) at (#1,0) {B};}
  \newcommand{\gnnL}[1]{\node[gninlinable,new] (nL) at (#1,0) {};}
  \newcommand{\geo}[2]{\draw[gedge,old] (#1) to (#2);}
  \newcommand{\gen}[2]{\draw[gedge,new] (#1) to (#2);}
  \newcommand{\gebo}[2]{\draw[gedge,old] (#1) to[bend left=10] (#2);}
  \newcommand{\gebn}[2]{\draw[gedge,new] (#1) to[bend left=10] (#2);}
  \newcommand{\gebbn}[2]{\draw[gedge,new] (#1) to[bend left=20] (#2);}
  We now work in $\graphfinj$ and use the shorthand notation introduced in \Cref{exa-prove-unsat}.
%  Let the following condition be given:
%  \begin{align*}
%    & \forall\ \emptyset \to \emptyset  \dotEx \emptyset \to \gI  \ . \condtrue
%    %\\ \land &
%    \ \land\ %
%    \forall\ \emptyset \to \gI  \dotEx \gI \to \gItoII  \ . \condtrue
%  \end{align*}
%
%  The meaning of this condition is:
%  (1) there exists a node \gI;
%  (2) and every node has an outgoing edge to some other node.
  Let the following condition be given:
  %\pagebreak
  \begin{align*}
    & \forall\ \emptyset \to \emptyset .\exists \emptyset \to \gI  \ . \condtrue
    && \text{\small(there exists a node \gI,}
    \\ \land &
    \forall\ \emptyset \to \gI  \dotEx \gI \to \gItoII  \ . \condtrue
    && \text{\small and every node has an outgoing edge to some other node)}
  \end{align*}

  This condition has finite models, the smallest being the cycle
  \gImutII.
  When running \Cref{alg:sat-check} on this condition, it obtains the model
  in the following way:

  \begin{enumerate}
  \item The given condition is universal with an iso
    $\emptyset \to \emptyset$, which is pulled forward.  Together with
    its (only) existential child, this results in the partial model
    $\Mor{\gnn11}$ and the condition
      \begin{align*}
        & \condtrue
        \land \big( \forall \Mor{\gnn11} \dotEx \Mor{\gno11 \gnn22 \gen{n1}{n2} } . \condtrue \big)_{\downarrow \Morsub{\gnn11}}
        \\ ={}
        & \forall \Mor{\gno11} \dotEx \Mor{\gno11 \gnn22 \gen{n1}{n2} }. \condtrue
        \ \land\  \forall \Mor{\gno11 \gnnA2 } . \big(\, \smash{\overbrace{
          \exists \Mor{\gno11 \gnoA2 \gnnB3 \gen{nA}{nB} } . \condtrue
          \lor
          \exists \Mor{\gno11 \gnoA2 \gen{nA}{n1} } . \condtrue
        }^{\mathcal B}} \,\big)
      \end{align*}
      meaning:
      (1) the just created node $\gI$ must have an outgoing edge;
      (2) and every other node $A$ must also have an outgoing edge
      to either another node or to the existing node.
    \item
      Pull forward iso $\Mor{\gno11}$ and extend the partial model by
      %$\gDom\gI \to \gItoII$
      $\Mor{\gno11 \gnn22 \gen{n1}{n2} }$,
      resulting in:
      \begin{align*}
        & \condtrue
        \land \big(
          \forall \Mor{\gno11 \gnnA2} . \mathcal B
        \big)_{\downarrow \Morsub{\gno11 \gnn22 \gen{n1}{n2}}}
        % \\ ={} &
        \:=\:
          \forall\Mor{\gno11 \gno22 \geo{n1}{n2} } . \mathcal B_{\downarrow \Morsub{\gno11 \gnoA2 \gen{n1}{nA} }}
          \ \land\ %
          \forall\Mor{\gno11 \gno22 \gnnA3 \geo{n1}{n2} } . \mathcal B_{\downarrow \Morsub{\gno11 \gnoA3 \gnn22 \gen{n1}{n2} }}
        \\ ={} &
          %\forall\ \gItoII \ . \big(
          %  \exists\ \gItoIItoIII \ . \condtrue
          %  \lor \exists\ \gImutII \ . \condtrue
          \forall \Mor{\gno11 \gno22 \geo{n1}{n2} } . \big(
            \exists \Mor{\gno11 \gno22 \geo{n1}{n2} \gnn33 \gen{n2}{n3} } . \condtrue
            \lor \exists \Mor{\gno11 \gno22 \gebo{n1}{n2} \gebn{n2}{n1} } . \condtrue
          \big)
          \\ & \land
          %\forall\ \gItoIIA \ . \big(
          \forall \Mor{\gno11 \gno22 \gnnA3 \geo{n1}{n2} }. \big(
            %\exists\ \gItoIIAtoB \ . \condtrue
            %\lor \exists\ \gItoIIotA \ . \condtrue
            %\lor \exists\ \gItoIIAtoI \ . \condtrue
            \exists \Mor{\gno11 \gno22 \gnoA3 \gnnB4 \geo{n1}{n2} \gen{nA}{nB} } . \condtrue
            \lor \exists \Mor{\gno11 \gno22 \gnoA3 \geo{n1}{n2} \gen{nA}{n2} } . \condtrue
            \lor \exists \Mor{\gno11 \gno22 \gnoA3 \geo{n1}{n2} \gebbn{nA}{n1} } . \condtrue
          \big)
      \end{align*}
      meaning:
      (1) the second node has an edge to a third node or to the first one;
      (2) and every other node $A$ also has an edge to either another node or to one of the existing~nodes.
    \item
      Next, we pull forward $\Mor{\gno11 \gno22 \geo{n1}{n2} }$ and extend the model by $\Mor{\gno11 \gno22 \gebo{n1}{n2} \gebn{n2}{n1} }$:
      \begin{align*}
        & \condtrue
        \land \big( \forall \Mor{\gno11 \gno22 \geo{n1}{n2} \gnnA3} . \,\dots \big)_{\downarrow \Morsub{\gno11 \gno22 \gebo{n1}{n2} \gebn{n2}{n1}} }
        %\\ ={} &
        \ =\ %
        \forall \Mor{\gno11 \gno22 \gebo{n1}{n2} \gebo{n2}{n1} \gnnA3 }. \,\dots
      \end{align*}
    \item This condition does not have any children with isos, so it
      is satisfiable by $\id$.  Hence the composition of the
      partial models so far
      ($\Mor{\gnn11 \gnn22 \gebn{n1}{n2} \gebn{n2}{n1} }$) is a
      model for the original condition.  \ExaEndHere
  \end{enumerate}
\end{example}

\endgroup% end of local scope for example graph drawing commands

\section{Witnesses for infinite models}
\label{sec:witnesses}

If there is no finitely decomposable model for a satisfiable condition
$\mathcal A$ (such as in \Cref{ex-ray-graphs} below), then the
corresponding infinite branch produces a model in the limit.
To detect such models in finite time we introduce an additional use of
\coinductive techniques based on the tableau calculus previously
introduced: We will show that under some circumstances, it is possible
to find some of these infinite models while checking for
satisfiability. Naturally, one can not detect all models, since this
would lead to a decision procedure for an undecidable problem.
We first have to generalize the notion of fairness to finite path
fragments: 

\begin{definition}
  \label{def:fair-path-witness}
  Let
  $\mathcal C_0 \xrightarrow{b_1} \mathcal C_1 \xrightarrow{b_2} \dots
  \xrightarrow{b_r} \mathcal C_r$ be a finite path (also called
  segment) in a tableau (cf.\ \Cref{satcheck.new.general}). Such a
  finite path is called \emph{fair} if for every child of
  $\mathcal{C}_0$ where the morphism is an iso, an indirect successor is pulled forward at some point in the
  path.
\end{definition}

This notion of fairness does not preclude that new isos appear. It
only states that all isos present at the beginning are pulled forward
at some point.

\begin{theoremrep}[Witnesses]\label{thm.witness}
  Let $\mathcal C_0$ be an alternating, universal condition.  Let a
  fixed tableau constructed by the rules of \Cref{satcheck.new} be
  given.  Let
  $\mathcal C_0 \xrightarrow{b_1} \mathcal C_1 \xrightarrow{b_2} \dots
  \xrightarrow{b_r} \mathcal C_r$ be a fair segment of a branch of the
  tableau where $r>0$, and let some arrow $m$ be given such that
  $\mathcal C_0 \condiso (\mathcal C_r)_{\downarrow m} $ for an iso
  $\iota \colon \RO({\mathcal C_r}_{\downarrow m}) \to \RO(\mathcal
  C_0)$.  Then,
  $[b_1, \tdots, b_r, m;\iota]^\omega := [b_1, \tdots, b_r, m;\iota,
  b_1, \tdots, b_r, m;\iota, \tdots]$ is a model for $\mathcal C_0$.
\end{theoremrep}

%\begin{toappendix}
\begin{proofsketch}
  Construct the relation
  \iffalse
  \[\setlength{\arraycolsep}{1.5pt}\begin{array}{rrlrrl}% P = { ([ | d1 ... m, | rest, ([ | ...m, | rest
    P = \{
      ([ & b_1, b_2, \tdots, b_r, m;\iota, & (b_1, \tdots, b_r, m;\iota)^\omega ], \mathcal C_0), &
      ([ & \:   b_2, \tdots, b_r, m;\iota, & (b_1, \tdots, b_r, m;\iota)^\omega ], \mathcal C_1), \\
      &  && \mathllap{\dots\:\:}
      ([ &                          m;\iota, & (b_1, \tdots, b_r, m;\iota)^\omega ], \mathcal C_r)
    \}
  \end{array}\]
  \fi
  \[\setlength{\arraycolsep}{1.5pt}\begin{array}{rrl}% P = { ([ | d1 ... m, | rest
    P = \{
      ([ & b_1, b_2, \tdots, b_r, m;\iota, & (b_1, \tdots, b_r, m;\iota)^\omega ], \mathcal C_0), \\
      ([ &      b_2, \tdots, b_r, m;\iota, & (b_1, \tdots, b_r, m;\iota)^\omega ], \mathcal C_1),
      \:\dots,
      ([ \quad                    m;\iota,   (b_1, \tdots, b_r, m;\iota)^\omega ], \mathcal C_r)
    \}
  \end{array}\] and show that $P \subseteq s(u(P))$, using an approach
similar to that of \Cref{thm:fair-branch-model}.  Steps $b_1, \dots, b_r$ are handled in the same way as in
\Cref{thm:fair-branch-model}.  For the newly introduced step based on
$m;\iota$, the next element of the sequence of representative squares
and the successor $\fsat i$ are chosen from a child of
${\mathcal C_r}_{\downarrow m}$ instead of from a successor of
$\Asat i$.
%
% The remaining parts of the proof (finding a ``successor'' that is
% pulled forward, and application of the up-to techniques) can be done
% as in \Cref{thm:fair-branch-model}.
\end{proofsketch}
%\end{toappendix}

\begin{proof}
  Let a fair segment
  $\mathcal C_0 \xrightarrow{u_1} \mathcal C_0' \xrightarrow{e_1}
  \mathcal C_1 \xrightarrow{u_2} \mathcal C_1' \xrightarrow{e_2} \dots
  \xrightarrow{e_r} \mathcal C_r$ be given, where $u_i, e_i$
  correspond to labels of universal and existential steps,
  respectively. That is $u_i = b_{2i-1}$, $e_i = b_{2i}$.

Also let an arrow $m$ be given such that $\mathcal C_0  \condiso
{\mathcal C_r}_{\downarrow m}$ wrt.\ an iso $\iota$.
Note that $\mathcal C_0$ is universal, the condition is alternating,
and ${\mathcal C_r}_{\downarrow m} \condiso \mathcal C_0$ implies that
$\mathcal C_r$ is also universal. Hence $r$ must be even.
We define the relation
  \[\def\arraystretch{1.2}\setlength{\arraycolsep}{1.5pt}\begin{array}{rrl}% P = { ([ | d1 ... m, | rest
    P = \{
      ([ & b_1, b_2, \dots, b_r, m;\iota, & (b_1, \dots, b_r, m;\iota)^\omega ], \mathcal C_0), \\
      ([ &      b_2, \dots, b_r, m;\iota, & (b_1, \dots, b_r, m;\iota)^\omega ], \mathcal C_1), \\
      & \vdots \\
      ([ &                          m;\iota, & (b_1, \dots, b_r, m;\iota)^\omega ], \mathcal C_r)
      % TODO figure out if we need an additional element to account for the upto iso part
    \}
  \end{array}\]
and show that $P \subseteq s(u(P))$:
Let $([ c_1, c_2, c_3, \dots ], \mathcal A) \in P$ for some condition $\mathcal A$ on the segment.
We show that $([ c_1, c_2, c_3, \dots ], \mathcal A) \in s(u(P))$.
The general structure of the proof matches that of \Cref{thm:fair-branch-model}. Hence we highlight the differences and refer to \Cref{thm:fair-branch-model} for detailed explanations of the unchanged parts.

Compared to \Cref{thm:fair-branch-model}, arrows on the paths of $P$
can have a step based on $m;\iota$ that does not result from an application of a tableau rule, but serves to close the loop from $\mathcal C_r$ to $\mathcal C_0$ via the shift ${\mathcal C_r}_{\downarrow m}$ (and an isomorphism $\iota$).
This requires changes to the selection of successors of $\forall \fsat0.\Asat0$.
Furthermore, the fairness constraint is guaranteed for the repeated segment, which slightly changes the reasoning for when a successor of $\forall \fsat p.\Asat p$ is pulled forward.
\begin{proofparts}
  \proofPart{$\mathcal A$ is existential}
  Existential conditions only occur in the middle of the segment,
  since the step for $m;\iota$ originates from $\mathcal C_r$ which is universal.
  Hence we can handle this case in exactly the same way as in the proof of \Cref{thm:fair-branch-model}, which refers to the tuple in $P$ that corresponds to the next condition on the path, and satisfies the definition of $s$ using $n=1$ and $g = \id$.

  \proofPart{$\mathcal A$ is universal}
  Let $\mathcal A = \bigland_{i \in I} \forall f_i.\mathcal A_i$.
  % \oldtodo{\textbf{L:} confirm/mention that all conditions on the segment must have some iso. \\
  % (The last one need not really have an iso, right? Since shifting it by $m$ might indeed produce a condition $\mathcal C_0$ that does have an iso. Figure out if that somehow conflicts with the conjecture that the branch cannot produce finite models)}
  For $([ c_1, c_2, \dots ], \bigland_{i \in I} \forall f_i . \mathcal A_i) \in s(u(P))$ to hold,
  by definition of $s$, for all children $\forall f_i.\mathcal A_i$, all arrows~$g$ and all $n \in \natzero$,
  we need to show: if $c_1 ; \dots ; c_n = f_i ; g$, then $([ g, c_{n+1}, \dots ], \mathcal A_i) \in u(P)$.

  Now choose a particular child $\forall \fsat0.\Asat0$ of
  $\mathcal A$ that should be satisfied, some $n$, and some $\gsat0$
  such that $c_1 ; \dots ; c_n = \fsat0 ; \gsat0$, be given, for which
  we now need to show that
  \mbox{$([ \gsat0, c_{n+1}, \dots ], \Asat0) \in u(P)$}.  The
  sequence $c_1, c_2, \dots$ starts at some position $2\ell$ within the
  segment
  ($\mathcal C_{2\ell} \xrightarrow{u_{\ell+1}}
  \xrightarrow{e_{\ell+1}} \mathcal C_{2\ell+2} \dots$ for some
  $\ell$), and is followed by $m;\iota$ and the full segment repeated
  ad infinitum. The morphism $c_n$ can be any of the arrows of this chain, after
  an arbitrary number of repetitions of the segment.  Graphically, the
  situation can be depicted like this:

  
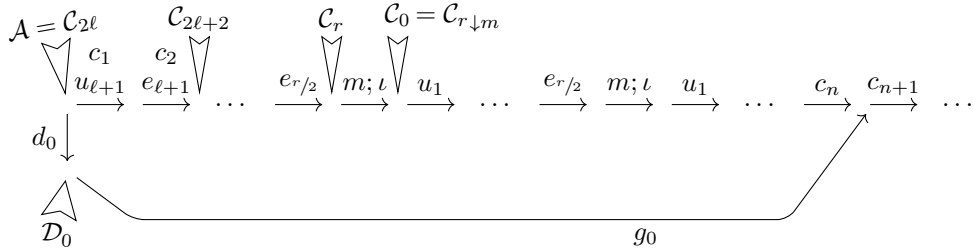
\begin{tikzpicture}[x=0.87cm,y=0.5cm]
    \foreach \i in {0,...,14} {
      \node (top\i) at (\i,0) {};
      \node (bot\i) at (\i,-1.75) {};
    }

    \draw[->] (top0) edge node[above,align=center]{$c_1$\\[-2pt]$u_{\ell+1}$} (top1)
              (top1) edge node[above,align=center]{$c_2$\\[-2pt]$e_{\ell+1}$} (top2)
              (top3) edge node[above]{$e_{\sfrac{r}{2}}$} (top4)
              (top4) edge node[above]{$m;\iota$} (top5)
              (top5) edge node[above]{$u_1$} (top6)
              (top7) edge node[above]{$e_{\sfrac{r}{2}}$} (top8)
              (top8) edge node[above]{$m;\iota$} (top9)
              (top9) edge node[above]{$u_1$} (top10)
              (top11) edge node[above]{$c_n$} (top12)
              (top12) edge node[above]{$c_{n+1}$} (top13);
    \node at ($(top2)!0.5!(top3)$) {$\dots$};
    \node at ($(top6)!0.5!(top7)$) {$\dots$};
    \node at ($(top10)!0.5!(top11)$) {$\dots$};
    \node at ($(top13)!0.5!(top14)$) {$\dots$};

    \draw[->] (top0) -- node[left]{$\fsat0$} (bot0);

    \draw[->,rounded corners=4pt] (bot0) -- +(1.0,-1.3) -- node[below,pos=0.7828]{$\gsat0$} +(10.9,-1.3) -- (top12.-60);

    \node[condtri,dart tip angle=30,shape border rotate=270,rotate around={10:(top0.center)}, label={[rotate=10,anchor=south,label distance=1pt]above:{$\mathcal A = \mathcal C_{2\ell}$}}] at (top0.north) {\kern5pt};
    \node[condtri,dart tip angle=20,shape border rotate=270,rotate around={0:(top2.center)}, label={[rotate=0,anchor=south,label distance=1pt]above:{$\mathcal C_{2\ell+2}$}}] at (top2.north) {\kern3.0pt};
    \node[condtri,dart tip angle=20,shape border rotate=270,rotate around={0:(top4.center)}, label={[rotate=0,anchor=south,label distance=1pt]above:{$\mathcal C_r$}}] at (top4.north) {\kern3.0pt};
    \node[condtri,dart tip angle=20,shape border rotate=270,rotate around={0:(top5.center)}, label={[rotate=0,anchor=south,label distance=1pt]above:{$\mathcal C_0 \mathrlap{\, = {\mathcal C_r}_{\downarrow m}}$}}] at (top5.north) {\kern3.0pt};
    \node[condtri,shape border rotate=90, rotate around={-10:(bot0.center)}, label={[rotate=-5,anchor=north,label distance=1pt]below:{$\Asat 0$}}] at (bot0.south) {\kern5pt};
  \end{tikzpicture}

%  Note that $c_n = e_{\frac{r}{2}}$.
  
  For the purposes of this proof, we need to extend the definition of
  a direct successor to the step based on $m;\iota$, as follows:
  First, shifting $\mathcal C_r$ results in an intermediate condition
  $\mathcal{H} \defeq {\mathcal C_r}_{\downarrow m} = \bigland_{t \in
    T} \forall h_t.\mathcal{H}_t$, with each child of $\mathcal C_r$
  being related to its shifted counterparts of $\mathcal{H}$.  As
  $\mathcal C_0 \condiso {\mathcal C_r}_{\downarrow m} = \mathcal{H}$
  wrt.\ $\iota$, each child $\forall h_t.\mathcal{H}_t$ is again
  related to (at least) one child of $\mathcal C_0$, together with an
  iso $\iota_t$ relating the children.  Then, each child of
  $\mathcal C_0$ is a successor of the corresponding original child of
  $\mathcal C_r$.

  We proceed to construct, as in \Cref{thm:fair-branch-model}, a
  sequence
  $\forall \fsat0.\Asat0, \forall \fsat1.\Asat1, \dots,
  \allowbreak\forall \fsat p.\Asat p, \dots, \allowbreak\forall \fsat
  q.\Asat q$ of children of the universal conditions
  $\mathcal{C}_{2i}$ on the infinite chain of arrows (i.e., the
  conditions where $m$ and all $u_i$ originate), with each element of
  the sequence being a (direct) successor of the previous element,
  such that:

  \begin{enumerate}
  \item After $p$ steps, $0 \leq 2p \leq n$, $\fsat p$ is an
    isomorphism,
  \item for $k>p$, $\fsat k$ is an isomorphism as well,
  \item after $q$ steps, $p \leq q < \infty$, $\fsat q$ is pulled
    forward, resulting in
    $([ e_{q+1}, \dots ], \Asat q) \in u(P)$.
  \end{enumerate}
  Afterwards, we can transform that to
  $([ \gsat0, c_{n+1}, \dots ], \Asat0) \in u(P)$ (as required by the
  satisfaction function $s$) by applying up-to techniques, thereby
  showing that the segment actually describes an infinite model.

%  Note that $p=0$ or $p=q$ are possible as well.
  The initial steps are depicted in the following diagram, with the
  chain $u_{\ell+1},e_{\ell+1},\dots$, $n$, $\fsat0,\Asat0$ and
  $\gsat0$ given:

  
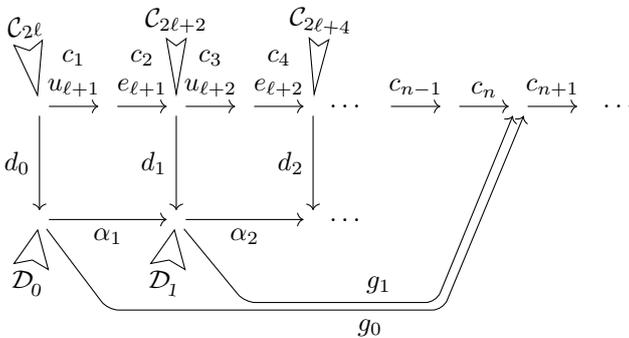
\begin{tikzpicture}[x=0.90cm]
    \foreach \i in {0,...,9} {
      \node (top\i) at (\i,0) {};
      \node (bot\i) at (\i,-1.5) {};
    }

    \draw[->] (top0) edge node[above,align=center]{$c_1$\\[-2pt]$u_{\ell+1}$} (top1)
              (top1) edge node[above,align=center]{$c_2$\\[-2pt]$e_{\ell+1}$} (top2)
              (top2) edge node[above,align=center]{$c_3$\\[-2pt]$u_{\ell+2}$} (top3)
              (top3) edge node[above,align=center]{$c_4$\\[-2pt]$e_{\ell+2}$} (top4);
    \node at ($(top4)!0.5!(top5)$) {$\dots$};
    \draw[->] (top5) edge node[above]{$c_{n-1}$} (top6)
              (top6) edge node[above]{$c_n$} (top7)
              (top7) edge node[above]{$c_{n+1}$} (top8);
    \node at ($(top8)!0.5!(top9)$) {$\dots$};

    \draw[->] (bot0) edge node[below]{$\alpha_1$} (bot2);
    \draw[->] (bot2) edge node[below]{$\alpha_2$} (bot4);
    \node at ($(bot4)!0.5!(bot5)$) {$\dots$};

    \draw[->] (top0) -- node[left]{$\fsat0$} (bot0);
    \draw[->] (top2) -- node[left]{$\fsat1$} (bot2);
    \draw[->] (top4) -- node[left]{$\fsat2$} (bot4);

    \draw[->,rounded corners=4pt] (bot0) -- +(1.0,-1.2) -- node[below,pos=0.7828]{$\gsat0$} +(5.9,-1.2) -- (top7.-60);
    \draw[->,rounded corners=4pt] (bot2) -- +(1.0,-1.1) -- node[above,pos=0.7]{$\gsat1$} +(3.8,-1.1) -- (top7.-120);

    \node[condtri,dart tip angle=30,shape border rotate=270,rotate around={10:(top0.center)}, label={[rotate=10,anchor=south,label distance=1pt]above:{$\mathcal C_{2\ell}$}}] at (top0.north) {\kern5pt};
    \node[condtri,dart tip angle=20,shape border rotate=270,rotate around={0:(top2.center)}, label={[rotate=0,anchor=south,label distance=1pt]above:{$\mathcal C_{2\ell+2}$}}] at (top2.north) {\kern3.0pt};
    \node[condtri,dart tip angle=20,shape border rotate=270,rotate around={-5:(top4.center)}, label={[rotate=-5,anchor=south,label distance=1pt]above:{$\mathcal C_{2\ell+4}$}}] at (top4.north) {\kern3.0pt};
    \node[condtri,shape border rotate=90, rotate around={-10:(bot0.center)}, label={[rotate=-10,anchor=north,label distance=1pt]below:{$\Asat 0$}}] at (bot0.south) {\kern5pt};
    \node[condtri,shape border rotate=90, rotate around={-10:(bot2.center)}, label={[rotate=-10,anchor=north,label distance=1pt]below:{$\Asat 1$}}] at (bot2.south) {\kern5pt};
  \end{tikzpicture}

  The proof objective is now to construct the sequence of all further $\forall \fsat{k+1}.\Asat{k+1}$ and associated $\alpha_{k+1}$ (and, for $k < p$, $\gsat{k+1}$).

  \begin{enumerate}
    \item
      We construct the aforementioned sequence by repeatedly choosing a successor of $\forall \fsat k.\Asat k$ as the next element $\forall \fsat{k+1}.\Asat{k+1}$ of the sequence.
      We then show that an isomorphism is always obtained after finitely many steps (at most $n$) by showing that for some $p$ with $2p \leq n$, $\fsat p$ must be an isomorphism.

      If $\fsat k$ is an isomorphism, we are done, set $p = k$ and we do not need to choose a next element. We will show that $2p \leq n$ later in the proof.

      So assume that $\fsat k$ is not an isomorphism.
      We consider the next arrow on the chain, which is either $u_{k+1}$ or $m$:

      \begin{itemize}
      \item Assume the next arrow on the chain results from a tableau
        step (i.e., $u_{k+1}$) and neither $u_{k+1}$ nor $e_{k+1}$
        equal $c_n$.

          Then construct the representative square and pick the
          successor as in \Cref{thm:fair-branch-model}.  Afterwards,
          start over again, with $\fsat{k+1}, \gsat{k+1}$
          instead of $\fsat k, \gsat k$.  (Since the chain
          $u_{k+2}, \dots c_n$ is now two elements shorter than
          before, this process will not continue endlessly.)

        \item Consider the case where $c_n = u_{k+1}$ or
          $c_n=e_{k+1}$ or $n=0$.

          As explained in \Cref{thm:fair-branch-model}, in the first case
          $c_n = u_{k+1}$ is an isomorphism, so
          $u_{k+1} = \fsat k ; \gsat k$ can only split into
          isomorphisms, so $\fsat k$ must be an isomorphism as well.
          In the second and third case, we can conclude that
          $\fsat{k+1} ; \gsat{k+1} = \id$, since $\id$ corresponds to
          an empty sequence of arrows and again $\fsat{k+1}$ is an iso.

        \item Consider the case that the next arrow is $m;\iota$, and
          it may or may not be $c_n$.

          To obtain a successor for $\forall \fsat k.\Asat k$, first
          reduce the left square below to a representative square
          (below center), to obtain a corresponding child of
          $\mathcal{H}$, which we call $\forall h_t.\mathcal{H}_t$
          with $h_t = \beta$.  Then, pick a matching isomorphic child
          $\forall \beta'.\mathcal{H}_t'$ of $\mathcal C_0$ and a
          relating iso $\iota_t$ to obtain the diagram on the right.

          
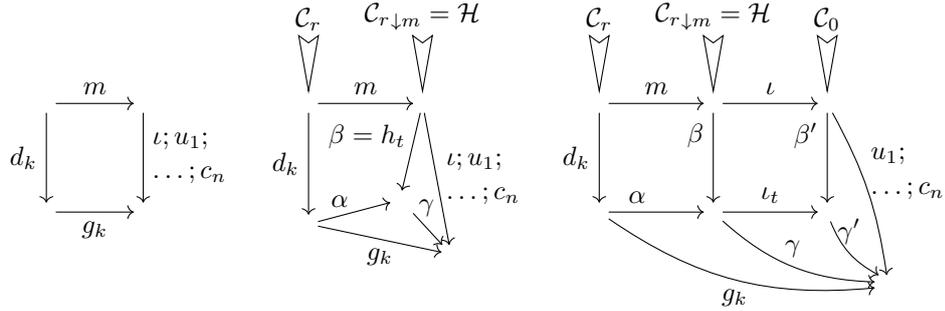
\begin{tikzpicture}[x=0.75cm,y=0.8cm]% PAGEBREAK-ADJUST
            % %%% LEFT - just the m part of the step
            \begin{scope}[shift={(0,0)}]
              \node (a) at (0,0) {};
              \node (b) at (1.7,0) {};
              \node (c) at (0,-1.8) {};
              \node (ds) at (1.7,-1.8) {};
              \draw[->] (a) -- node[above]{$m$} (b);
              \draw[->] (a) -- node[left]{$\fsat k$} (c);

              \draw[->] (b) -- node[right,align=left]{$\iota ; u_1 ;$\\$\dots ; c_n$} (ds);
              \draw[->] (c) -- node[below]{$\gsat k$} (ds);
            \end{scope}

            %\node at (3.5,-1) {$\rightarrow$};

            % %%% CENTER - a representative square, corresponding to the shift, has been inscribed
            \begin{scope}[shift={(4.6,0)}]
              \node (a) at (0,0) {};
              \node (b) at (2,0) {};
              \node (c) at (0,-2) {};
              \node (d) at (1.6,-1.6) {};
              \node (ds) at (2.5,-2.5) {};
              \draw[->] (a) -- node[above]{$m$} (b);
              \draw[->] (a) -- node[left]{$\fsat k$} (c);
              \draw[->] (b) -- node[left,pos=0.3]{$\beta = h_t$} (d);
              \draw[->] (c) -- node[above,pos=0.3]{$\alpha$} (d);

              \draw[->] (d.center)+(5pt,-5pt) -- (ds);
              \node at ($(d.center)+(10pt,-3pt)$) {$\gamma$};
              \draw[->] (b) -- node[right,align=left]{$\iota ; u_1 ;$\\$ \dots ; c_n$} (ds);
              \draw[->] (c) -- node[below]{$\gsat k$} (ds);

              \node[condtri,dart tip angle=20,shape border
              rotate=270,rotate around={0:(a.center)},
              label={[rotate=0,anchor=south,label
                distance=1pt]above:{$\mathcal C_r$}}] at (a.north)
              {\kern3.0pt}; \node[condtri,dart tip angle=20,shape
              border rotate=270,rotate around={0:(b.center)},
              label={[rotate=0,anchor=south,label
                distance=1pt]above:{${\mathcal C_r}_{\downarrow m} =
                  \mathcal{H}$}}] at (b.north) {\kern3.0pt};
            \end{scope}
            % %%% RIGHT - extended by a second square to account for condition isomorphism
            \begin{scope}[shift={(9.7,0)}]
              \node (ro-cr) at (0,0) {};
              \node (ro-cr-m) at (2,0) {};
              \node (ro-c0) at (4,0) {};
              \node (ro-asatr) at (0,-1.8) {};
              \node (ro-asatr-b) at (2,-1.8) {};
              \node (ro-asat0) at (4,-1.8) {};
              \node (cntarg) at (5,-3) {};
              \draw[->] (ro-cr) -- node[above]{$m$} (ro-cr-m);
              \draw[->] (ro-cr-m) -- node[above]{$\iota$} (ro-c0);
              \draw[->] (ro-asatr-b) -- node[above]{$\iota_t$} (ro-asat0);
              \draw[->] (ro-cr) -- node[left]{$\fsat k$} (ro-asatr);
              \draw[->] (ro-cr-m) -- node[left,pos=0.25]{$\beta$} (ro-asatr-b);
              \draw[->] (ro-asatr) -- node[above,pos=0.3]{$\alpha$} (ro-asatr-b);
              \draw[->] (ro-c0) -- node[left,pos=0.25]{$\beta'$} (ro-asat0);

              \draw[->] (ro-asatr) to[bend right=20] node[below]{$\gsat k$} (cntarg.200);
              \draw[->] (ro-asatr-b) to[bend right=20] node[above]{$\gamma$} (cntarg.165);
              \draw[->] (ro-asat0) to[bend right=20] node[above]{$\gamma'$} (cntarg.120);
              \draw[->] (ro-c0) to[bend left=10] node[right,align=left,pos=0.38]{$u_1 ;$\\$ \dots ; c_n$} (cntarg.75);

              \node[condtri,dart tip angle=20,shape border
              rotate=270,rotate around={0:(ro-cr.center)},
              label={[rotate=0,anchor=south,label
                distance=1pt]above:{$\mathcal C_r$}}] at (ro-cr.north)
              {\kern3.0pt}; \node[condtri,dart tip angle=20,shape
              border rotate=270,rotate around={0:(ro-cr-m.center)},
              label={[rotate=0,anchor=south,label
                distance=1pt]above:{${\mathcal C_r}_{\downarrow m} =
                  \mathcal{H}$}}] at (ro-cr-m.north) {\kern3.0pt};
              \node[condtri,dart tip angle=20,shape border
              rotate=270,rotate around={0:(ro-c0.center)},
              label={[rotate=0,anchor=south,label
                distance=1pt]above:{$\mathcal C_0$}}] at (ro-c0.north)
              {\kern3.0pt};
            \end{scope}
          \end{tikzpicture}%

          Hence the successor of $\forall \fsat k.\Asat k$ is a child
          $\forall \fsat{k+1}.\Asat{k+1} \defeq \forall
          \beta'.\mathcal{H}_t'$ of $\mathcal C_0$, and
          $\gsat{k+1} \defeq \gamma' = (\iota_t)^{-1};\gamma$. We note
          that the two triangles in the lower right of the diagram
          above commute: first obviously $\iota_t;\gamma' =
          \iota_t;(\iota_t)^{-1};\gamma = \gamma$ and second
          $\beta';\gamma' = \beta';(\iota_t)^{-1};\gamma =
          \iota^{-1};\beta;\gamma = \iota^{-1};\iota;u_1;\dots;c_n
          = u_1;\dots;c_n$.

          If $m;\iota = c_n$, then $\fsat{k+1};\gsat{k+1} = \id$, and
          hence $\fsat{k+1}$ must be an isomorphism and we have found
          the iso $\fsat p$.  Otherwise, proceed with the remainder of
          the chain ($u_1, e_1, \dots$).
      \end{itemize}

    \item
      Summarizing the situation so far, $\fsat 0$ has been turned into an isomorphism $\fsat p$ by now:

      
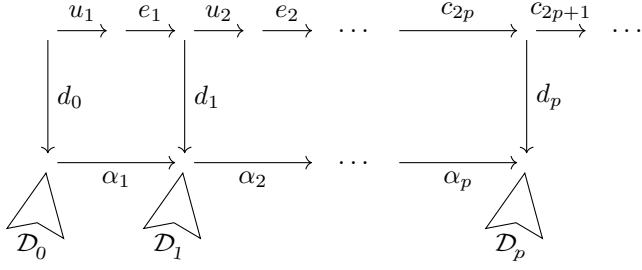
\begin{tikzpicture}[x=0.90cm]
        \foreach \i in {0,...,9} {
          \node (top\i) at (\i,0) {};
          \node (bot\i) at (\i,-1.75) {};
        }

        \draw[->] (top0) edge node[above]{$u_1$} (top1)
                  (top1) edge node[above]{$e_1$} (top2)
                  (top2) edge node[above]{$u_2$} (top3)
                  (top3) edge node[above]{$e_2$} (top4);
        \node at ($(top4)!0.5!(top5)$) {$\dots$};
        \draw[->] (top5) edge node[above]{$c_{2p}$} (top7)
                  (top7) edge node[above]{$c_{2p+1}$} (top8);
        \node at ($(top8)!0.5!(top9)$) {$\dots$};

        \draw[->] (bot0) edge node[below]{$\alpha_1$} (bot2);
        \draw[->] (bot2) edge node[below]{$\alpha_2$} (bot4);
        \node at ($(bot4)!0.5!(bot5)$) {$\dots$};
        \draw[->] (bot5) edge node[below]{$\alpha_p$} (bot7);

        \draw[->] (top0) -- node[right]{$\fsat0$} (bot0);
        \draw[->] (top2) -- node[right]{$\fsat1$} (bot2);
        \draw[->] (top7) -- node[right]{$\fsat p$} (bot7);

        \node[condtri,shape border rotate=90, rotate around={-10:(bot0.center)}, label={[rotate=-10,anchor=north,label distance=1pt]below:{$\Asat 0$}}] at (bot0.south) {\kern12pt};
        \node[condtri,shape border rotate=90, rotate around={-10:(bot2.center)}, label={[rotate=-10,anchor=north,label distance=1pt]below:{$\Asat 1$}}] at (bot2.south) {\kern12pt};
        \node[condtri,shape border rotate=90, rotate around={-10:(bot7.center)}, label={[rotate=-10,anchor=north,label distance=1pt]below:{$\Asat p$}}] at (bot7.south) {\kern12pt};
      \end{tikzpicture}

%      \oldtodo[inline]{\textbf{L:} this is like the 65\textsuperscript{th} time the pairing of $u_i, e_i$ is being annoying: here, what I labeled as $c_p$ might be either $u_p, e_p$ or it might be $m$.
%      we could define $P$ (\emph{not} the tableau rules) to compose $u_i, e_i$, which could also resolve some issues in the other proofs, \textbf{comments?}}

      To be able to pull an indirect successor of $\fsat p$ forward, it has to be an isomorphism.
      See the corresponding part of the proof of \Cref{thm:fair-branch-model} on how to ensure that all elements $\forall \fsat{p+1}.\Asat{p+1}, \allowbreak\forall \fsat{p+2}.\Asat{p+2}, \dots$ of the sequence do in fact have isomorphisms. (The argument can be applied to the shift over $m;\iota$ as well.)

    \item
      Since $\fsat p$ is an isomorphism and all its indirect successors are isomorphisms as well, there is an indirect successor that is a direct child of $\mathcal C_0$ and also an iso.
      As the segment is fair, this means that some successor of it must be pulled forward after finitely many further steps (before reaching $\mathcal C_r$ the next time).

      By the same reasoning as in \Cref{thm:fair-branch-model}, for
      some $q\ge p$, we reach a condition
      $\mathcal{C}_t = \bigland_{i\in I} \forall f_i.\mathcal{A}_i$
      that contains
      $\forall f_\ell.\mathcal{A}_\ell = \forall \fsat q.\Asat q$,
      which is an indirect successor of $\forall \fsat p.\Asat p$ and
      which is pulled forward. We assume that
      $\Asat q = \bigvee_j \exists g_j . \mathcal B_j$. Hence we obtain the
      following tuple in $P$:
      \[ \Big( [ e_{q+1},\ u_{q+2}, e_{q+2}, \dots ],
      \biglor\nolimits_j \exists g_j . \Big( \mathcal B_j \land
      \Big( \bigland_{\kern10pt\mathclap{m \ne \ell}\kern10pt} \forall f_m . \mathcal A_m \Big)_{\downarrow \fsat q ; g_j}
      \Big)
      \Big) \in P \]
  \end{enumerate}
  Here $u_{q+1} = f_\ell = d_q$.

  We apply our up-to techniques to obtain the originally desired result
  (which was to satisfy a given $\forall \fsat0.\Asat0$ by showing for a given $n$ and $\gsat0$ such that $c_1 ; \dots ; c_n = \fsat0 ; \gsat0$
  that $([ \gsat0, c_{n+1}, \dots ], \Asat0) \in u(P)$ holds).

  As before, $\Asat q$ is ``contained'' in the previously obtained path condition, and up-to conjunction removal (\Cref{thm:upto-conjunction-removal}) can be used to obtain a tuple with only $\Asat q$:
  \[ \implies
  \Big( [ e_{q+1}, u_{q+1}, e_{q+2}, \dots ],
  \biglor\nolimits_j \exists g_j . \mathcal B_j
  \Big) =
  \Big( [ e_{q+1}, u_{q+2}, e_{q+2}, \dots ], \Asat q \Big) \in u(P) \]

  $\Asat q$ has been derived from $\Asat 0$ by a series of tableau steps (which involve a shift), or $m;\iota$ step (which involves a shift and a condition isomorphism).
  Hence, for $k$ ranging from $q$ to $0$, we apply a series of steps:
  \begin{itemize}
    \item
      If $\Asat{k+1} = {\Asat{k}}_{\downarrow \alpha_{k+1}}$ (tableau step),
      apply up-to shift (\Cref{thm:compat-uptoshift}) to obtain:
      \begin{align*}
        ([ \alpha_{k+2} ; \dots ; \alpha_q ; e_{q+1}, u_{q+2}, e_{q+2}, \dots ], &\ \Asat{k+1}) \in u(P) \\
        \implies ([ \alpha_{k+1} ; \alpha_{k+2} ; \dots ; \alpha_q ;
        e_{q+1}, u_{q+2}, e_{q+2}, \dots ], &\ \Asat{k}) \in u(P)
      \end{align*}
    \item If $\Asat{k+1}$ is the result of a step based on $m;\iota$,
      we know that $\forall \fsat{k}.\Asat{k}$ is a child of
      $\mathcal{C}_r$, while $\forall \fsat{k+1}.\Asat{k+1}$ -- its
      successor -- is a child of $\mathcal{C}_0$.
      
      Now first apply up-to condition isomorphism (\Cref{thm:compat-uptoiso}),
      then apply up-to shift (\Cref{thm:compat-uptoshift}) to undo the shift:
      \begin{align*}
        ([ \alpha_{k+2} ; \dots ; \alpha_q ; e_{q+1}, u_{q+2}, e_{q+2}, \dots ], &\ \Asat{k+1}) \in u(P) \\
        \implies
        ([ \iota_t ; \alpha_{k+2} ; \dots ; \alpha_q ; e_{q+1}, u_{q+2}, e_{q+2}, \dots ], &\ \mathcal{H}_t) \in u(P) \\
        \implies
        ([ \alpha ; \iota ; \alpha_{k+1} ; \dots ; \alpha_q ; e_{q+1}, u_{q+2}, e_{q+2}, \dots ], &\ \Asat{k}) \in u(P)
      \end{align*}
      Note that ${\Asat{k}}_{\downarrow \alpha} = \mathcal{H}_t$.
    \end{itemize}

    We eventually obtain a tuple with $\Asat0$ as the condition.  Now
    we use up-to recomposition (\Cref{thm:compat-uptorecomp}) to
    rewrite the initial part of the chain (see
    \Cref{thm:fair-branch-model} for details). For the step based on
    $m;\iota$ we require commutativity of the two triangles that was
    shown previously. Finally we obtain:
  \[
    ([ \gsat0, \ c_{n+1}, \dots ], \ \Asat0) \in u(P)
    \qedhere
  \]
\end{proofparts}
\end{proof}

\begin{example}[Finding witnesses]
  \label{ex-ray-graphs}
  Consider the following condition:
  %\pagebreak
  \begin{align*}
    &
      \forall\ \emptyset \to \emptyset
      .\exists \emptyset \to \fcGraph{n1}{\node[gninlinable] (n1) at (0,0) {$1$};}
      \dotAll %
        \fcGraph{n1}{\node[gninlinable] (n1) at (0,0) {$1$};}
        \to
        \fcGraph{n1}{
          \node[gninlinable] (n1) at (0,0) {$1$};
          \node[gninlinable] (n0) at (-1,0) {};
          \draw[gedge] (n0) to (n1); }
      \dotFalse
    & \text{\small(there is a node \fcGraph{n1}{\node[gninlinable] (n1) at (0,0) {$1$};} without an incoming edge}  \\
    \land &
      \forall\ \emptyset \to \fcGraph{n1}{\node[gninlinable] (n1) at (0,0) {$\vphantom1$};}
      \dotEx
        \fcGraph{n1}{\node[gninlinable] (n1) at (0,0) {$\vphantom1$};}
        \to
        \fcGraph{n1}{\node[gninlinable] (n1) at (0,0) {$\vphantom1$}; \node[gninlinable] (n2) at (1,0) {$\mathclap+$}; \draw[gedge] (n1) to (n2);}
      \dotTrue
    & \text{\kern-35pt\small and every node has an outgoing edge to some other node} \\
    \land &
      \forall\ \emptyset \to
        \fcGraph{n1}{
          \node[gninlinable] (n0) at (0,0) {};
          \node[gninlinable] (n1) at (-1,0) {$1$};
          \node[gninlinable] (n2) at (1,0) {$2$};
          \draw[gedge] (n1) to (n0);
          \draw[gedge] (n2) to (n0); }
      \dotFalse
     & \text{\small and no node has two incoming edges)}
  \end{align*}

  This condition has an infinite model, namely an infinite path
  (\fcGraph{n1}{
    \node[gninlinable] (n1) at (0.0,0) {$1$};
    \node[gninlinable] (n2) at (1,0) {$2$};
    \node[gninlinable] (n3) at (2,0) {$3$};
    \node[gninlinablefont,inner sep=0.5pt] (d) at (3.1,0) {$\cdots$};
    \draw[gedge] (n1) to (n2);
    \draw[gedge] (n2) to (n3);
    \draw[gedge] (n3) to (d);
  }). It does not have any finite
  model.

  \newcommand{\cospan}[1]{\llbracket\mkern3mu\fcGraph{n1}{#1}\mkern3mu\rrbracket}
  \tikzset{
    old/.style={draw=black!38, color=black!38, line width=0.4pt},
    new/.style={line width=0.7pt},
  }
  % draw commands for the various graphs that are used in this example
  \newcommand{\gno}[2]{\node[gninlinable,old] (n#1) at (#2,0) {#1};}
  \newcommand{\gnn}[2]{\node[gninlinable,new] (n#1) at (#2,0) {#1};}
  \newcommand{\gnoZ}{\node[gninlinable,old] (n0) at (0,0) {};}
  \newcommand{\gnnZ}{\node[gninlinable,new] (n0) at (0,0) {};}
  \newcommand{\gnoP}[1]{\node[gninlinable,old] (nP) at (#1,0)
    {$\mathclap+$};} \newcommand{\gnnP}[1]{\node[gninlinable,new] (nP)
    at (#1,0) {$\mathclap+$};}
  \newcommand{\gnoX}[1]{\node[gninlinable,old] (nX) at (#1,0)
    {$\mathclap\times$};} \newcommand{\gnnX}[1]{\node[gninlinable,new]
    (nX) at (#1,0) {$\mathclap{\boldsymbol\times}$};}
  \newcommand{\gnnA}[1]{\node[gninlinable,new] (nA) at (#1,0) {A};}
  \newcommand{\gnnB}[1]{\node[gninlinable,new] (nB) at (#1,0) {B};}
  \newcommand{\gnnL}[1]{\node[gninlinable,new] (nL) at (#1,0) {};}
  \newcommand{\geo}[2]{\draw[gedge,old] (#1) to (#2);}
  \newcommand{\gen}[2]{\draw[gedge,new] (#1) to (#2);}
  \newcommand{\gebo}[2]{\draw[gedge,old] (#1) to[bend left=10] (#2);}
  \newcommand{\gebn}[2]{\draw[gedge,new] (#1) to[bend left=10] (#2);}
  \def\mutualedges#1#2{\draw[gedge] (#1) to[bend left=10] (#2);
    \draw[gedge] (#2) to[bend left=10] (#1);} In order to display a
  witness for this model, we need to consider the condition in the
  category of cospans $\ILC(\graphfinj)$, into which $\graphfinj$ can
  be embedded. We use the following shorthand
  notation for cospans: $\cospan{ \gno11 \gnn22 \gen{n1}{n2} }$ means
  $ \fcGraph{n1}{\node[gninlinable] (n1) at (0,0)
    {$1$};} \rightarrow \fcGraph{n1}{ \node[gninlinable] (n1) at (0,0)
    {$1$}; \node[gninlinable] (n2) at (1,0)
    {$2$}; \draw[gedge] (n1) to (n2); } \leftarrow \fcGraph{n1}{
    \node[gninlinable] (n1) at (0,0)
    {$1$}; \node[gninlinable] (n2) at (1,0)
    {$2$}; \draw[gedge] (n1) to (n2); } $, i.e., the left object
  consists of only the light-gray graph elements, the center and right
  objects consist of the full graph, the left leg is the inclusion and
  the right leg is always the identity.

  If we execute our algorithm on the condition, after one step we obtain a condition $\mathcal C_0$ that is rooted at $\fcGraph{n1}{
    \node[gninlinable] (n1) at (0,0) {$1$};
  }$, and spells out the requirements of the original condition for node \fcGraph{n1}{\node[gninlinable] (n1) at (0,0) {$1$};} and all other nodes separately (\fcGraph{n1}{\node[gninlinable] (n1) at (0,0) {$1$};} has a successor, and all other nodes have successors, and so on).

  After another step, we obtain $\mathcal C_1$, which is rooted at $\fcGraph{n1}{
    \node[gninlinable] (n1) at (0,0) {$1$};
    \node[gninlinable] (n2) at (1,0) {$2$};
    \draw[gedge] (n1) to (n2);}$,
  and does the same for node \fcGraph{n1}{\node[gninlinable] (n1) at (0,0) {$2$};} separately as well.
  (These steps are displayed more concretely in the appendix, \Cref{table-witness-steps}.)

  Now let
  $m = \fcGraph{n1}{\node[gninlinable] (n1) at (0,0)
    {$1$}; \node[gninlinable] (n2) at (1,0)
    {$2$}; \draw[gedge] (n1) to (n2);} \rightarrow
  \fcGraph{n1}{\node[gninlinable] (n1) at (0,0)
    {$1$}; \node[gninlinable] (n2) at (1,0)
    {$2$}; \draw[gedge] (n1) to (n2);} \leftarrow
  \fcGraph{n2}{\node[gninlinable] (n2) at (0,0) {$2$};} $. Then, we
  can compute $(\mathcal C_1)_{\downarrow m}$, which
  essentially ``forgets'' node \fcGraph{n1}{\node[gninlinable] (n1) at
    (0,0) {$1$};} from all subconditions of $\mathcal C_1$.
  (Subconditions that contain edges from or to this node disappear entirely, which is a
  consequence of the way borrowed context diagrams are constructed
  (via pushout complements).)
  Then, $(\mathcal C_1)_{\downarrow m}$ is similar in structure to $\mathcal C_0$, and in fact, using a renaming isomorphism $\iota =
    \fcGraph{n1}{\node[gninlinable] (n1) at (0,0) {$2$};}
    \rightarrow \fcGraph{n1}{\node[gninlinable] (n1) at (0,0) {$\vphantom1$};}
    \leftarrow \fcGraph{n1}{\node[gninlinable] (n1) at (0,0) {$1$};}
  $, it holds that $(\mathcal C_1)_{\downarrow m} \condiso \mathcal C_0$ wrt.\ $\iota$.
\end{example}

\begin{toappendix}
  \newcommand{\cospan}[1]{\llbracket\mkern3mu\fcGraph{n1}{#1}\mkern3mu\rrbracket}
  \tikzset{
    old/.style={draw=black!38, color=black!38, line width=0.4pt},
    new/.style={line width=0.7pt},
  }
  % draw commands for the various graphs that are used in this example
  \newcommand{\gno}[2]{\node[gninlinable,old] (n#1) at (#2,0) {#1};}
  \newcommand{\gnn}[2]{\node[gninlinable,new] (n#1) at (#2,0) {#1};}
  \newcommand{\gnoZ}{\node[gninlinable,old] (n0) at (0,0) {};}
  \newcommand{\gnnZ}{\node[gninlinable,new] (n0) at (0,0) {};}
  \newcommand{\gnoP}[1]{\node[gninlinable,old] (nP) at (#1,0)
    {$\mathclap+$};} \newcommand{\gnnP}[1]{\node[gninlinable,new] (nP)
    at (#1,0) {$\mathclap+$};}
  \newcommand{\gnoX}[1]{\node[gninlinable,old] (nX) at (#1,0)
    {$\mathclap\times$};} \newcommand{\gnnX}[1]{\node[gninlinable,new]
    (nX) at (#1,0) {$\mathclap{\boldsymbol\times}$};}
  \newcommand{\gnnA}[1]{\node[gninlinable,new] (nA) at (#1,0) {A};}
  \newcommand{\gnnB}[1]{\node[gninlinable,new] (nB) at (#1,0) {B};}
  \newcommand{\gnnL}[1]{\node[gninlinable,new] (nL) at (#1,0) {};}
  \newcommand{\geo}[2]{\draw[gedge,old] (#1) to (#2);}
  \newcommand{\gen}[2]{\draw[gedge,new] (#1) to (#2);}
  \newcommand{\gebo}[2]{\draw[gedge,old] (#1) to[bend left=10] (#2);}
  \newcommand{\gebn}[2]{\draw[gedge,new] (#1) to[bend left=10] (#2);}
  \def\mutualedges#1#2{\draw[gedge] (#1) to[bend left=10] (#2);
    \draw[gedge] (#2) to[bend left=10] (#1);}

  \begin{table}[h]
  \newcommand{\dFalse}{.\condfalse}
  \newcommand{\dTrue}{.\condtrue}
  \newcommand{\AllC}[1]{\forall\cospan{#1}}
  \newcommand{\ExC}[1]{\exists\cospan{#1}}
  \vspace*{-\baselineskip}
  \[\def\arraystretch{1.0}\setlength{\arraycolsep}{3pt}\begin{array}{rl|rl|rl}
    \multicolumn{2}{l|}{\mathcal C_0}                                      &\multicolumn{2}{l|}{\mathcal C_1}                                                                              &\multicolumn{2}{l}{(\mathcal C_1)_{\downarrow m}} \\ \hline
         &\AllC{ \gnnZ \gno11 \gen{n0}{n1} }\dFalse                        &     &\AllC{ \gnnZ \gno11 \gno22 \geo{n1}{n2} \gen{n0}{n1} }\dFalse \\
         &                                                                 &\land&\AllC{\gno11\gno2{2.3} \gebo{n1}{n2} \gebn{n2}{n1} }\dFalse \\
    \land&\AllC{\gno11\gnnP2}.\big(                                        &\land&\AllC{\gno11\gno22\geo{n1}{n2}\gnnP{2.8}} .\big(                                                         &     &\AllC{\gno22\gnnP{2.8}} .\big(  \\
         &        \mkern30mu    \ExC{\gno11\gnoP2\gnnL3\gen{nP}{nL}}\dTrue &     &    \mkern30mu    \ExC{\gno11\gno22\geo{n1}{n2}\gnoP{2.8}\gnnL{3.8}\gen{nP}{nL} }\dTrue                  &     &    \mkern30mu    \ExC{\gno22\gnoP{2.8}\gnnL{3.8}\gen{nP}{nL}}\dTrue \\
         &        \mkern18mu\lor\ExC{\gno11\gnoP2\gen{nP}{n1}}\dTrue\big)  &     &    \mkern18mu\lor\ExC{\gno11\gno22\geo{n1}{n2}\gnoP{3}  \gen{nP}{n2} }\dTrue                            &     &    \mkern18mu\lor\ExC{\gno22\gnoP{3}\gen{nP}{n2}}\dTrue \big) \\
         &                                                                 &     &    \mkern18mu\lor\ExC{\gno11\gno22\geo{n1}{n2}\gnoP{0}  \gen{nP}{n1} }\dTrue \big) \\
    \land&\AllC{\gno11}.\ExC{\gno11\gnn22\gen{n1}{n2} }\dTrue              &\land&\AllC{\gno11\gno22\geo{n1}{n2}} .\big( \ExC{\gno11 \gno22 \gnn33 \geo{n1}{n2} \gen{n2}{n3}}\dTrue        &\land&\AllC{\gno22} . \ExC{\gno22 \gnn33 \gen{n2}{n3}}\dTrue \\
         &                                                                 &     &                         \mkern83mu\lor\ExC{\gno11 \gno2{2.3} \gebo{n1}{n2} \gebn{n2}{n1} }\dTrue \big) \\
    \land&\AllC{\gno11\gnnA2\gnnX3\gnnB4\gen{nA}{nX}\gen{nB}{nX}}\dFalse   &\land&\AllC{ \gno11 \gno22 \gnnA3 \gnnX4 \gnnB5 \geo{n1}{n2} \gen{nA}{nX} \gen{nB}{nX} }\dFalse                &\land&\AllC{ \gno22 \gnnA3 \gnnX4 \gnnB5 \gen{nA}{nX} \gen{nB}{nX} }\dFalse \\
    \land&\AllC{\gno11\gnnX2\gnnB3\gen{n1}{nX}\gen{nB}{nX}}\dFalse         &\land&\AllC{ \gno11 \gno22  \gnnX3 \gnnB4 \geo{n1}{n2} \gen{n2}{nX} \gen{nB}{nX} }\dFalse                      &\land&\AllC{ \gno22  \gnnX3 \gnnB4 \gen{n2}{nX} \gen{nB}{nX} }\dFalse \\
         &                                                                 &\land&\AllC{ \gno11 \gno22  \gnnB3 \geo{n1}{n2} \gen{nB}{n2} }\dFalse                                          &\land&\AllC{ \gno22  \gnnB3 \gen{nB}{n2} }\dFalse \\
    \land&\ldots                                                           &\land&\ldots                                                                                                   &\land&\ldots
  \end{array}\]
  \caption{Steps for the condition of \Cref{ex-ray-graphs}, showing that a repeating infinite model exists.}%
  \label{table-witness-steps}%
  \vspace*{-\baselineskip}
  \end{table}
\end{toappendix}

In general this witness construction will almost never be applicable
for simple graph categories, we need to work in other categories, such
as cospan categories.

\section{Satisfiability in the General Case}
\label{sec:general-case}

\Cref{sec:satisfiability} heavily depends on the fact that all
sections are isos, i.e., only isos have a right inverse. However, in
the general case we might have conditions of the form
$\forall s.\mathcal{A}$ where $s$ has a right inverse $r$ with
$s;r = \id$ (note that $r$ need not be unique). This would invalidate
our reasoning in the previous sections, since the identity is not
necessarily a model of $\forall s.\mathcal{A}$.

% Note that every section is always a mono, but the reverse is not true,
% not even in $\mathbf{Set}$ where functions of type
% $\emptyset \to A\neq\emptyset$ are injective (hence monos), but not
% sections. In the category of graphs and graph morphisms there are
% considerably more counterexamples.

\begin{example}\label{ex-section-motivation}
  We work in the category $\graphf$. Consider the following condition
  $\mathcal A = \forall\ %
  \fGraph{n1}{\node[gninlinable] (n1) at (0,0)
    {$1$};} \to \fGraph{n1}{ \node[gninlinable] (n1) at (0,0)
    {$1$}; \node[gninlinable] (n2) at (0.5,0)
    {$2$}; } \dotEx \fGraph{n1}{ \node[gninlinable] (n1) at (0,0)
    {$1$}; \node[gninlinable] (n2) at (0.5,0)
    {$2$}; } \to \fGraph{n1}{ \node[gninlinable] (n1) at (0,0)
    {$1$}; \node[gninlinable] (n2) at (0.75,0)
    {$2$}; \draw[gedge] (n1) to (n2); } \ . \condtrue $, defined over
  a single node
  $\fcGraph{n1}{\node[gninlinable] (n1) at (0,0) {$1$};}$ as root
  object, which states that the distinguished node has an edge to
  every other node -- including itself, since a non-injective match
  may merge the two nodes. The first morphism of $\mathcal{A}$ is a
  section, while the second is injective, but not a section.

  The identity on the single node is not a model of $\mathcal{A}$, but
  $ \fGraph{n1}{\node[gninlinable] (n1) at (0,0)
    {$1$};} \to \fGraph{n1}{ \node[gninlinable] (n1) at (0,0)
    {$1$}; \draw[gedge] (n1) to[loop right] (); } $ is. However, the
  condition $ \mathcal{A} \land \forall\ %
  \fGraph{n1}{\node[gninlinable] (n1) at (0,0)
    {$1$};} \to \fGraph{n1}{ \node[gninlinable] (n1) at (0,0)
    {$1$}; \draw[gedge] (n1) to[loop right] (); } \ . \condfalse $ is
  unsatisfiable, a fact that would not be detected by
  \Cref{alg:sat-check}, since neither of the universal quantifiers
  contains an iso.
\end{example}

% While the previous pull-forward rule (for isos) is still valid, only
% conditions without sections have the identity as model.  For
% conditions with sections, we need to find a new way to pull
% subconditions forward.

\begin{toappendix}
  We now show some straightforward facts on sections in the context of
  conditions.

  \begin{lemma}
    \label{lem:all-section-implies-exists-section}
    Let $s \colon A \to B$ be a section, where $r \colon B \to A$ is a
    right inverse of $s$, i.e., $s ; r = \id_A$.  Then
    $\forall s.\mathcal A \models \mathcal A_{\downarrow r} \models
    \exists s.\mathcal A$.  However, in general,
    $\exists s.\mathcal A \notmodels \mathcal A_{\downarrow r}
    \notmodels \forall s.\mathcal A$.
  \end{lemma}

\begin{proof}
  Note that
  $\RO(\forall s.\mathcal A) = \RO(\exists s.\mathcal A) =
  \RO(\mathcal A_{\downarrow r}) = A$ and $\RO(\mathcal A) = B$.
%  \oldtodo{\textbf{L:} in \cite[Proposition
%    17]{bchk:conditional-reactive-systems} a slightly different
%    strategy
%    was used that doesn't even restrict to sections. \\
%    \textbf{B:} these are the adjunction laws. Are they really related? \\
%    \textbf{L:} They could be applied to prove the first two items:
%    $\mathcal A \models \forall r.(\mathcal A_{\downarrow r})$, then
%    $\forall s.\mathcal A \models \forall s.\forall r.(\mathcal
%    A_{\downarrow r}) \equiv \forall s;r.(\mathcal A_{\downarrow r}) =
%    \forall \id.(\mathcal A_{\downarrow r}) \equiv \mathcal
%    A_{\downarrow r}$.  (Some of these steps only work if $s;r=\id$.)
%  }
  In the following let $\bar{c} = [c_1,c_2,\dots]$ be a composable
  sequence of arrows.

\begin{proofparts}
  \proofPartNoNewline{$\forall s . \mathcal A \models \mathcal
    A_{\downarrow r}$}
  \begin{align*}
    & \bar{c}\models \forall s.\mathcal{A}
    \\
    \text{(Def.\ \ref{def:satisfaction} satisfaction)} \iff& \forall
    n.\forall g.( c_1;\dots;c_n=s;g \implies [g,c_{n+1},\dots] \models
    \mathcal A )
    \\
    \text{(choose $g = r ; c_1 ; \dots c_n$)} \implies & \forall n.
    (c_1;\dots;c_n=s;r;c_1;\dots;c_n \\&\qquad \implies
    [r;c_1;\dots;c_n, c_{n+1}, \dots] \models \mathcal A)
    \\
    \text{(antecedent is true)} \implies& \forall n. [r;c_1;\dots;c_n,
    c_{n+1}, \dots] \models \mathcal A
    \\
    \text{(Def.\ \ref{prop:shift} shift)} \iff&
    [c_1;\dots;c_n,c_{n+1},\dots] \models \mathcal
    A_{\downarrow r} \\
    \text{(recompose)} \iff& \bar{c} \models \mathcal A_{\downarrow r}
  \end{align*}
  Note that the last step is true since $\uptoRecomp(\models)$ equals
  $\models$ (see \Cref{sec:coinductive}) and hence satisfaction is
  preserved by recomposition.
  \proofPartNoNewline{$\mathcal A_{\downarrow r} \models \exists s
    . \mathcal A$}
  \begin{align*}
    & \bar{c} \models \mathcal A_{\downarrow r}
    \\
    \text{(Prop.\ \ref{prop:shift} shift)} \iff& [r;c_1,c_2,\dots]
    \models \mathcal A
    \\
    \text{(ex. intro. using $g = r;c_1$)} \implies& \exists g.( c_1 =
    s;g \land [g,c_2,\dots] \models \mathcal A )
    \\
    \text{(Def.\ \ref{def:satisfaction} satisfaction)} \iff & \bar{c}
    \models \exists s . \mathcal A
  \end{align*}
  \proofPart{$\exists s . \mathcal A \notmodels \mathcal A_{\downarrow
      r} \notmodels \forall s . \mathcal A$} We show this via a
  counterexample using the category $\graphf$ of finite graphs and
  morphisms. Representative squares are given by pushouts.
    
  Here we consider only finite composable sequences, given by a single
  morphism. Let the following graph morphisms and conditions be given,
  where the mappings of the morphisms are indicated in green:

  \def\oneNodeLeft{ \node[gninlinable] (l) at (0,0) {}; }
  \def\twoNodesLeft{ \node[gninlinable] (l1) at (0,0.25) {};
    \node[gninlinable] (l2) at (0,-0.25) {}; } \def\oneNodeRight{
    \node[gninlinable] (r) at (1,0) {}; }
  \def\twoNodesRight{\node[gninlinable] (r1) at (1,0.25) {};
    \node[gninlinable] (r2) at (1,-0.25) {}; } \def\graphLabel#1{
    \node[anchor=mid east] at (-0.2,0) {$#1 = {}$}; }
  \def\toprightLoopAt#1{\draw[gedge] (#1)
    to[loop,out=40,in=-20,distance=4mm] (#1); }
  \def\smoltoprightLoopAt#1{\draw[gedge] (#1)
    to[loop,out=50,in=-10,distance=3mm] (#1); }
  \def\smolbotrightLoopAt#1{\draw[gedge] (#1)
    to[loop,out=-70,in=-20,distance=3mm] (#1); }

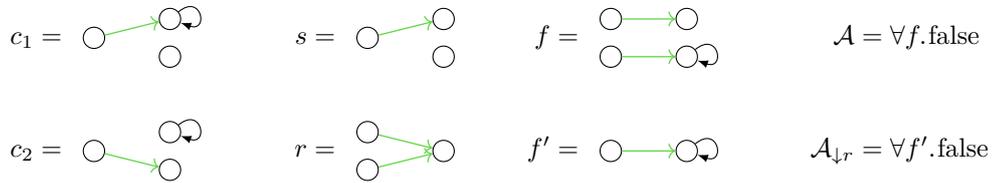
    \begin{tikzpicture}
      \pgfmathsetmacro\tabygrid{1.5}

      \node at (-1.55,0)
      {}; % "indent" the picture just as if it was an actual formula

      \begin{scope}[shift={(0,0)}]
        \begin{scope}[shift={(0,0)}]
          \graphLabel{c_1} \oneNodeLeft \twoNodesRight
          \toprightLoopAt{r1} \draw[graphmorcol] (l) to (r1);
        \end{scope}
        \begin{scope}[shift={(0,-1*\tabygrid)}]
          \graphLabel{c_2} \oneNodeLeft \twoNodesRight
          \toprightLoopAt{r1} \draw[graphmorcol] (l) to (r2);
        \end{scope}
      \end{scope}

      \begin{scope}[shift={(3.6,0)}]
        \begin{scope}[shift={(0,0)}]
          \graphLabel{s} \oneNodeLeft \twoNodesRight
          \draw[graphmorcol] (l) to (r1);
        \end{scope}
        \begin{scope}[shift={(0,-1*\tabygrid)}]
          \graphLabel{r} \twoNodesLeft \oneNodeRight
          \draw[graphmorcol] (l1) to (r); \draw[graphmorcol] (l2) to
          (r);
        \end{scope}
      \end{scope}

      \begin{scope}[shift={(3.6+3.2,0)}]
        \begin{scope}[shift={(0,0)}]
          \graphLabel{f} \twoNodesLeft \twoNodesRight
          \toprightLoopAt{r2} \draw[graphmorcol] (l1) to (r1);
          \draw[graphmorcol] (l2) to (r2);
        \end{scope}
        \begin{scope}[shift={(0,-1*\tabygrid)}]
          \graphLabel{f'} \oneNodeLeft \oneNodeRight
          \toprightLoopAt{r} \draw[graphmorcol] (l) to (r);
        \end{scope}
      \end{scope}

      \begin{scope}[shift={(3.6+3.2+3.5+0.3,0)}]
        \node[anchor=mid east] at (0, 0*\tabygrid)
        {$\mathcal A = {}\mathrlap{\forall f.\condfalse}$};
        \node[anchor=mid east] at (0,-1*\tabygrid)
        {$\mathcal A_{\downarrow r} = {}\mathrlap{\forall
            f'.\condfalse}$};
      \end{scope}
    \end{tikzpicture}

    \begin{itemize}
    \item
      $c_1 \models \exists s.\mathcal A$ because $\exists q_1 \colon c_1 = s;q_1$ and $q_1 \models \forall f.\condfalse$ \\
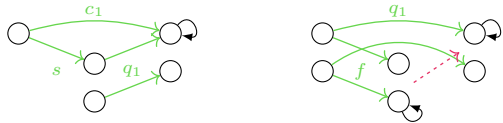
      \begin{tikzpicture}
        \begin{scope}[every node/.append style=gninlinable]
          \node (al) at (0,0) {}; \node (am1) at (1,-0.4) {}; \node
          (am2) at (1,-0.9) {}; \node (ar1) at (2,0) {}; \node (ar2)
          at (2,-0.5) {};

          \node (bl1) at (4,0) {}; \node (bl2) at (4,-0.5) {}; \node
          (bm1) at (5,-0.4) {}; \node (bm2) at (5,-0.9) {}; \node
          (br1) at (6,0) {}; \node (br2) at (6,-0.5) {};
        \end{scope}
        \smoltoprightLoopAt{ar1} \smolbotrightLoopAt{bm2}
        \smoltoprightLoopAt{br1}

          \begin{scope}[graphmorcol]
            \draw (al) to (am1); \draw (am1) to (ar1); \draw (am2) to
            (ar2); \draw (bl1) to (bm1); \draw (bl2) to (bm2); \draw
            (al) to[bend left=15] (ar1); \draw (bl1) to[bend left=15]
            (br1); \draw (bl2) to[bend left=35] (br2);
          \end{scope}
          \begin{scope}[color=graphmorphismgreen,font=\scriptsize]
            \node at (1,0.3) {$c_1$}; \node at (0.5,-0.5) {$s$}; \node
            at (1.5,-0.5) {$q_1$}; \node at (5,0.3) {$q_1$}; \node at
            (4.5,-0.5) {$f$};
          \end{scope}
          \draw[graphmorcolfail] (5.2,-0.65) to (5.8,-0.25);
        \end{tikzpicture}
      \item
        $c_1 \notmodels \mathcal A_{\downarrow r} = \forall
        f'.\condfalse$
        because $\exists q_2 \colon c_1 = f';q_2$ but $q_2 \notmodels \condfalse$ \\
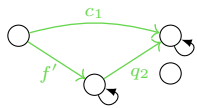
        \begin{tikzpicture}
          \begin{scope}[every node/.append style=gninlinable]
            \node (al) at (0,0) {}; \node (am) at (1,-0.65) {}; \node
            (ar1) at (2,0) {}; \node (ar2) at (2,-0.5) {};
          \end{scope}
          \smolbotrightLoopAt{am} \smolbotrightLoopAt{ar1}

          \begin{scope}[graphmorcol]
            \draw (al) to (am); \draw (am) to (ar1); \draw (al)
            to[bend left=15] (ar1);
          \end{scope}
          \begin{scope}[color=graphmorphismgreen,font=\scriptsize]
            \node at (1,0.3) {$c_1$}; \node at (0.4,-0.5) {$f'$};
            \node at (1.6,-0.5) {$q_2$};
          \end{scope}
        \end{tikzpicture}
      \item
        $c_2 \models \mathcal A_{\downarrow r}$ because $\nexists q_3 \colon c_2 = f';q_3$ \\
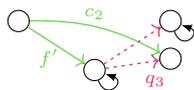
        \begin{tikzpicture}
          \begin{scope}[every node/.append style=gninlinable]
            \node (al) at (0,0) {}; \node (am) at (1,-0.65) {}; \node
            (ar1) at (2,0) {}; \node (ar2) at (2,-0.5) {};
          \end{scope}
          \smolbotrightLoopAt{am} \smolbotrightLoopAt{ar1}

          \begin{scope}[graphmorcol]
            \draw (al) to (am); \draw (al) to[bend left=15] (ar2);
          \end{scope}
          \begin{scope}[color=graphmorphismgreen,font=\scriptsize]
            \node at (1,0.1) {$c_2$}; \node at (0.4,-0.5) {$f'$};
            \node[color=graphmorphismfailred] at (1.8,-0.8) {$q_3$};
          \end{scope}
          \draw[graphmorcolfail] (am) to (ar1); \draw[graphmorcolfail]
          (am) to (ar2);
        \end{tikzpicture}
      \item
        $c_2 \notmodels \forall s.\mathcal A$ because $\exists q_4 \colon c_2 = s;q_4$ but $q_4 \notmodels \forall f.\condfalse$ \\
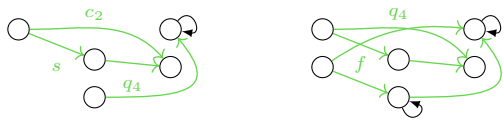
        \begin{tikzpicture}
          \begin{scope}[every node/.append style=gninlinable]
            \node (al) at (0,0) {}; \node (am1) at (1,-0.4) {}; \node
            (am2) at (1,-0.9) {}; \node (ar1) at (2,0) {}; \node (ar2)
            at (2,-0.5) {};

            \node (bl1) at (4,0) {}; \node (bl2) at (4,-0.5) {}; \node
            (bm1) at (5,-0.4) {}; \node (bm2) at (5,-0.9) {}; \node
            (br1) at (6,0) {}; \node (br2) at (6,-0.5) {};
          \end{scope}
          \smoltoprightLoopAt{ar1} \smolbotrightLoopAt{bm2}
          \smoltoprightLoopAt{br1}

          \begin{scope}[graphmorcol]
            \draw (al) to (am1); \draw (am1) to (ar2); \draw (am2)
            to[out=0,in=-50,looseness=2] (ar1); \draw (bl1) to (bm1);
            \draw (bl2) to (bm2); \draw (al) to[out=0,in=130] (ar2);
            \draw (bl1) to[out=0,in=130] (br2); \draw (bl2) to[bend
            left=20] (br1); \draw (bm1) to (br2); \draw (bm2)
            to[out=0,in=-50,looseness=2] (br1);
          \end{scope}
          \begin{scope}[color=graphmorphismgreen,font=\scriptsize]
            \node at (1,0.2) {$c_2$}; \node at (0.5,-0.5) {$s$}; \node
            at (1.5,-0.75) {$q_4$}; \node at (5,0.2) {$q_4$}; \node at
            (4.5,-0.47) {$f$};
          \end{scope}
        \end{tikzpicture}
        \qedhere
      \end{itemize}
    \end{proofparts}
  \end{proof}
\end{toappendix}

Pulling forward isos is still sound even in the general case as the equivalence of \Cref{lem:pull-forward-isos} still holds, but it is not sufficient for completeness.
Hence we will now adapt the tableau calculus to deal with sections.

\begin{lemmarep}[Pulling forward sections]\label{pull-forward-sections}
  Let
  \raisebox{0pt}[\height][\depth+2pt]{$\bigland_{i \in I} \forall
    f_i.\mathcal A_i$} be a universal condition and assume that $f_p$,
  $p\in I$, is a section, and $r_p$ is a right inverse of $f_p$ (i.e.,
  $f_p ; r_p = \id$). Furthermore let
  \raisebox{0pt}[\height][\depth+2pt]{${\mathcal A_p}_{\downarrow r_p} = \biglor_{j \in J} \exists
  h_j.\mathcal H_j$} be the result of shifting the $p$-th child over
  the right inverse. Then $f_p$ can be \emph{pulled forward}:
  \[
    \bigland_{i \in I} \forall f_i.\mathcal A_i
    \:\equiv\:
      \biglor_{j \in J} \exists h_j.\smash{\Bigg(}
        \mathcal H_j \land
        \Big( \bigland_{i \in I} \forall f_i.\mathcal A_i \Big)_{\downarrow h_j}
      \smash{\Bigg)}
  \]
\end{lemmarep}

\begin{proof}

  \def\landIinI{\bigland_{i \in I}}
  \def\landMneP{\bigland_{m \in I \setminus \{p\}}}
  \def\allFiAi{\forall f_i.\mathcal A_i}
  \def\allFmAm{\forall f_m.\mathcal A_m}
  \def\lorJinJ{\biglor_{j \in J}}
  \begin{align*}
    \landIinI\allFiAi
      &=
    \landMneP \allFmAm \land \forall f_p.\mathcal A_p
      \\ (*) &\equiv
    \landIinI\allFiAi \land {\mathcal A_p}_{\downarrow r_p}
      \\&=
    \landIinI\allFiAi \land \lorJinJ \exists h_j.\mathcal H_j
      \\ \text{(\Cref{lem:exists-and-shift} and distributivity)}&\equiv
    \lorJinJ \exists h_j.\Big(
      \mathcal H_j \land
      \landIinI (\allFiAi)_{\downarrow h_j}
    \Big)
  \end{align*}

  $(*)$: from \Cref{lem:all-section-implies-exists-section} we know
  that
  $\forall f_p.\mathcal A_p \models {\mathcal A_p}_{\downarrow
    r_p}$. Hence, since $\mathcal{C}\models \mathcal{D}$ implies
  $\mathcal{C} \equiv \mathcal{C}\land \mathcal{D}$, we have that 
  $\forall f_p.\mathcal A_p \equiv \forall f_p.\mathcal A_p \land
  {\mathcal A_p}_{\downarrow r_p}$.
  % : $\Rightarrow$
  % \Cref{lem:all-section-implies-exists-section}, $\Leftarrow$ because
  % $\mathcal C \land \mathcal D \models \mathcal C$.
\end{proof}

As for the analogous \Cref{lem:pull-forward-isos}, pulling forward
sections produces an equivalent condition.  However, here the child
being pulled forward is still included in the children shifted by
$h_j$. Hence the condition will increase in size, unlike for the
special case.  This is necessary, since $f_p$ might have other
inverses which can be used in pulling forward and might lead to new
results. This leads to the following adapted rules.

\begin{definition}[SatCheck rules, general case]
  \label{satcheck.new.general}
  Let $\mathcal{A}$ be an alternating condition. We can construct a
  tableau for $\mathcal A$ by extending it at its leaf nodes
  as follows:
  \[\begin{array}{ll|ll}
    \text{For every $p \in I$:}
    &&&
    \parbox[t]{0.54\textwidth}{\raggedright
      For \emph{one} $p \in I$ such that $f_p$ is section, $f_p ; r_p = \id$, and
      $\biglor\nolimits_{j \in J} \exists h_j . \mathcal H_j = {\mathcal A_p}_{\downarrow r_p}$:
    }
    \\
    \displaystyle\biglor_{i \in I} \exists f_i . \mathcal A_i
    \:\xrightarrow{f_p}\:
    \mathcal A_p
    &&&
    \displaystyle\bigland_{i \in I} \forall f_i . \mathcal A_i
    \:\xrightarrow{\id}\:
    \biglor_{j \in J} \exists h_j . \bigg(
      \mathcal H_j
      \land
      \Big( \bigland_{i \in I} \forall f_i . \mathcal A_i \Big)_{\downarrow h_j}
    \bigg)
  \end{array}\]
  \begin{itemize}
  \item For existential conditions, for \emph{each}(!) child condition
    $\exists f_p.\mathcal A_p$, add a new descendant.
%, i.e., for every
%    $p \in I$:
%    \[
%      \biglor_{i \in I} \exists f_i . \mathcal A_i
%      \:\xrightarrow{f_p}\:
%      \mathcal A_p
%    \]
  \item For universal conditions, pick \emph{one}(!) child condition
    $\forall f_p.\mathcal A_p$ that can be pulled forward in the sense
    of \Cref{pull-forward-sections} and add the result as its (only)
    descendant.
% I.e., For \emph{one} $p \in I$ such that $f_p$ is section,
%    $f_p ; r_p = \id$, and
%    $\biglor\nolimits_{j \in J} \exists h_j . \mathcal H_j = {\mathcal
%      A_p}_{\downarrow r_p}$:
%  \[
%    \bigland_{i \in I} \forall f_i . \mathcal A_i
%    \:\xrightarrow{\id}\:
%    \biglor_{j \in J} \exists h_j . \biggg(
%      \mathcal H_j
%      \land
%      \Big( \bigland_{i \in I} \forall f_i . \mathcal A_i \Big)_{\downarrow h_j}
%    \biggg)
%  \]
  \end{itemize}
%  \noindent% avoid misleading indentation, do not remove
%  As with \Cref{thm:fair-branch-model}, with each combination of $\forall$-step followed by $\exists$-step,
%  we associate a successor relation for top-level child conditions,
%  and additionally a relation that tracks right-inverses of top-level children with sections,
%  as detailed in the following \Cref{lts-section-successor-relation}.
\end{definition}

\noindent% avoid misleading indentation, do not remove
The rules are similar to those of the specialized case
(\Cref{satcheck.new}) and as in the special case, we need to define a
successor relation on children with sections, that in
addition tracks the corresponding right-inverses. The
definition of the successor relation is slightly more complex than in
the previous case (\Cref{lts-asai-successor-relation}).

\begin{definition}[Successor relation and tracking right-inverses]
  \label{lts-section-successor-relation}
  We define a relation on pairs of $(f_i,r_i)$, where $f_i$ is a
  morphism of a child of a universal condition and $r_i$ is one of its
  right-inverses.  Assume that in the construction of a tableau we
  have a path
  \[ \bigland_{i\in I} \forall f_i.\mathcal A_i \:\xrightarrow{\id}\:
    \mathcal{C} \:\xrightarrow{h_j}\: \mathcal H_j \land
    \Big( \bigland_{i\in I} \forall f_i.\mathcal A_i \Big)_{\downarrow h_j} =
    \mathcal H_j \land \bigland_{i\in I} \bigland_{(\alpha,\beta) \in
      \kappa(f_i,h_j )} \forall \beta.({\mathcal A_i})_{\downarrow
      \alpha} \]
  \floatingpicspaceright{3.2cm}
  \begin{floatingpicright}{2.9cm}%
    \hfill\begin{tikzpicture}[x=0.9cm,y=0.9cm]
      \node (a) at (0,0) {};
      \node (b) at (1.8,0) {};
      \node (c) at (0,-1.7) {};
      \node (d) at (1.4,-1.3) {};
      \node (ds) at (2.2,-2.0) {};
      \draw[->] (a) -- node[above]{$\smash{h_j}$} (b);
      \draw[->] (a) -- node[left]{$f_i$} (c);

      \draw[->,fancydotted] (b) -- node[left,pos=0.3]{$\beta$} (d);
      \draw[->,fancydotted] (c) -- node[above,pos=0.35]{$\alpha$} (d);

      \draw[->,fancydotted] (d) -- (ds);
      \node at ($(d.center)+(12pt,-4pt)$) {$r_\beta$};
      \draw[->] (b) -- node[right]{$\id$} (ds);
      \draw[->] (c) -- node[below,overlay]{$r_i;h_j$} (ds);
    \end{tikzpicture}%
  \end{floatingpicright}
  where $\mathcal{C}$ is the existential condition given in
  \Cref{pull-forward-sections}. Given $(f_i,r_i)$ (where
  $f_i;r_i = \id$), we can conclude that the outer square on the right
  commutes and hence there exists an inner representative square
  \mbox{$(\alpha,\beta) \in \kappa(f_i,h_j )$} and $r_\beta$ such that
  the diagram to the right commutes (in particular $\beta$ is a
  section and $r_\beta$ is a retraction). In this situation, we say
  that $(\beta,r_\beta)$ is a \emph{(retraction) successor} of
  $(f_i,r_i)$.%

  % The left figure depicts the situation for one $(\alpha,\beta) \in \kappa(f_i,h_j)$.
  % In each such case, the successor relation relates $\forall f_i.\mathcal A_i$ to $\forall \beta.({\mathcal A_i}_{\downarrow \alpha})$.

  % Moreover, with each such pair of steps, we also associate another relation that tracks right-inverses (retractions) of top-level children with sections:
  % let $\forall f_i.\mathcal A_i$ be a top-level child, where $f_i$ is a section, and $r_i$ a right-inverse such that $f_i;r_i = \id$.
  % In the diagram on the right, reduce the outer commuting square to some (inner) representative square,
  % such that \mbox{$\beta ; r_\beta = \id$} (i.e., $\beta$ is a section and $r_\beta$ is a retraction)
  % and $\alpha ; r_\beta = r_i ; h_j$.
  % In this situation, we say that $r_\beta$ is the \emph{retraction successor} of $r_i$.

%  Note that depending on the class $\kappa$, there might be additional $(\alpha',\beta') \in \kappa(f_i,h_j)$ that do not satisfy the equalities given above.
%  However, each retraction is guaranteed to have at least one retraction successor $r_\beta$ that guarantees $\alpha ; r_\beta = r_i ; h_j$
%  (a property that will be used in the proof when applying up-to techniques).
\end{definition}

We extend the definition of fairness to cover sections (not just isomorphisms) and also require that all right-inverses of each section are eventually used in a pull-forward step:

\begin{definition}[Fairness in the general case]
  A branch of a tableau is \emph{fair} if for each universal condition
  $\mathcal A$ on the branch, each child condition
  $\forall f_i.\mathcal A_i$ of $\mathcal A$ where $f_i$ is a section,
  and each right-inverse $r_i$ of $f_i$, there is $n \in \natzero$
  such that in the $n$-th next step, for some indirect successor
  $(f'_i,r'_i)$ of $(f_i,r_i)$ it holds that $f'_i$ is pulled forward
  using the right inverse $r'_i$.  (Every universally-quantified
  section is eventually pulled forward with every inverse.)
\end{definition}

For this definition to be effective, we need to require that every
section has only finitely many right-inverses (this is true for e.g.~$\graphf$).
Given that property, %similar to the specialized case,
one way to implement a fairness strategy is to use a queue, to which
child conditions (and the corresponding right-inverses) are added.
This queue has to be arranged in such a way that for each
section/retraction pair a successor is processed eventually.

We now show how to adapt the corresponding results of the previous
section (\Cref{thm:fair-branch-model,thm-algo-sound,thm-algo-complete})
and in particular show that infinite and fair
branches are always models, from which we can infer soundness and
completeness.

\begin{theoremrep}[Fair branches are models (general case)]
  \label{thm:fair-branch-model-general}
  Let $\mathcal A_0$ be an alternating condition.  Let a fixed tableau
  constructed by the rules of \Cref{satcheck.new.general} be given.
  Let
  $\mathcal A_0 \xrightarrow{b_1} \mathcal A_1 \xrightarrow{b_2}
  \mathcal A_2 \xrightarrow{b_3} \dots$ be a branch of the tableau
  that is either unextendable and ends with a universal
  quantification, or is infinite and fair.  For such a branch, we
  define:
  $ P = \{ (\bar b, \mathcal A_i) \mid i \in \natzero,\allowbreak \bar c = [
    b_{i+1}, b_{i+2}, \tdots ] \} \subseteq \Seq \times \Cond$.
  Then, $P \subseteq s(u(P))$.  (Every open and unextendable
  or infinite and fair branch in a tableau of
  \Cref{satcheck.new.general} corresponds to a model.)
\end{theoremrep}

\begin{proof}
% macros \fsat, \Asat etc. defined above the special variant of this theorem

The general structure of the proof is similar to the proof of \Cref{thm:fair-branch-model}.
For detailed explanations, we refer to the proof of the special case,
and only highlight the changes that are necessary for the general
case. Note that once a condition contains a section (on the top
level), the branch will always be infinite, since sections are never removed.

Let $\mathcal{A} =  \mathcal{C}_0 \xrightarrow{c_1} \mathcal{C}_1 \xrightarrow{c_2} \mathcal{C}_2 \xrightarrow{c_3} \dots$
be a (suffix of a) branch of a tableau that results in a given
$([ c_1, c_2, c_3, \dots ], \mathcal A) \in P$. We show that $([ c_1, c_2, c_3, \dots ], \mathcal A) \in s(u(P))$.
\begin{proofparts}
  \proofPart{$\mathcal A$ is existential}
  Exactly as in the special case.

  \proofPart{$\mathcal A$ is universal and contains no sections}
  Assume that $\mathcal{A} = \bigland_{i\in I} \forall
  f_i.\mathcal{A}_i$. The argument is similar to the special case:
  If $\mathcal A$ contains no sections, the tableau has no child nodes for $\mathcal A$ and the path ends, i.e., we have $([\,], \mathcal A)$ ($[\,]$ representing $\id$),
  which is in $s(u(P))$ because there is no arrow $g$ that makes $\id = f_i ; g$ true (if there were such an arrow, then $f_i$ would be a section).

  \proofPart{$\mathcal A$ is universal and has at least one section}
  We now assume that at least one $f_i$ is a section.

  For $([ c_1, c_2, \dots ], \bigland_{i \in I} \forall f_i . \mathcal A_i) \in s(u(P))$ to hold,
  by definition of $s$, for all children $\forall f_i.\mathcal A_i$, all arrows~$g$ and all $n \in \natzero$,
  we need to show: if $c_1 ; \dots ; c_n = f_i ; g$, then $([ g, c_{n+1}, \dots ], \mathcal A_i) \in u(P)$.

  Hence let a particular child
  $\forall \fsat0.\Asat0 = \forall f_i.\mathcal{A}_i$ to be satisfied,
  some $n$, and some $\gsat0$ such that
  $c_1 ; \dots ; c_n = \fsat0 ; \gsat0$, be given, for which we
  need to show that
  \mbox{$([ \gsat0, c_{n+1}, \dots ], \Asat0) \in u(P)$}.

  Since at least one child is a section, in the next step on the
  branch, some $\forall \fpf0.\Apf0$ ($\fpf0$ section) is pulled
  forward, so the next condition on the path has the shape
  $\biglor_j \exists h_j.(\mathcal H_j \land \dots)$.
%  \oldtodo{which is somehow ``related'' to $\mathcal A_p$ but isn't even ``contained'' in it, as it was the case in the special case}

  Since the condition is alternating and the tableau rules preserve
  this property, the sequence of $c_1, c_2, \ldots$ on the path
  alternates between $\id$ labels (from a universal condition) and
  $e_i$ (from an existential condition).
  In particular $c_{2i} = e_i$.
  Without loss of generality we can assume that $n$ is even and
  $c_n = e_{\sfrac{n}{2}}$. (Otherwise $c_n = \id$ and we know that
  $c_1 ; \dots ; c_{n-1} = \fsat0 ; \gsat0$, where $c_{n-1}$ a
  morphism corresponding to the existential step. Then we can infer
  $([\gsat0, c_n, c_{n+1}, \dots ], \Asat0) \in u(P)$ and finally
  $([\gsat0, c_{n+1}, \dots ], \Asat0) \in u(P)$ via up-to
  recomposition.) We set $m = \sfrac{n}{2}$.

  % In the following, we use $e_n$ to refer to the label of the
  % existential step at or immediately before $c_n$, i.e., if $c_n$ is
  % the label of an existential step then $e_n = c_n$, otherwise
  % $e_n = c_{n-1}$.
  % \oldtodo{\textbf{L:} not sure if it's guaranteed that
  %   any successor of $\fsat0$ is ever pulled forward. It might be that
  %   it is never turned into a section before the path ends. (see
  %   remark in special case for details) Proof currently just assumes
  %   that the path is actually infinite.
  %   Note, in the following point 3, it is explained why it has to be pulled forward if it already is a section, but this does not solve the problem of it not turning into a section at all \\
  %   ...if it's never pulled forward, the path will end because no
  %   sections, in which case it is trivially satisfiable.  maybe we
  %   could reword the previous bullet point to cover not just paths
  %   that end right here, but also at some later point?}

  We now claim that there exists a finite sequence
  $\forall \fsat0.\Asat0, \forall \fsat1.\Asat1, \dots,
  \allowbreak\forall \fsat m.\Asat m, \dots, \allowbreak\forall \fsat
  q.\Asat q$ of children of the universal conditions
  $\mathcal C_0, \mathcal C_2, \dots$ on the branch (i.e., every
  second one), with each element of the sequence being a (direct)
  successor of the previous element, such that:

  \begin{enumerate}
    \item
      After $m$ steps,
      $\fsat m$ is a section, having a right-inverse $\frisat m$,
    \item for $k>m$, $\fsat k$ is a section as well, with a
      right-inverse $\frisat k$ such that $(d_k,r_k)$ is a retraction successor of
      $(d_m,r_m)$.
    \item after $q$ steps, $m \leq q < \infty$, $\fsat q$ is pulled
      forward with $\frisat q$ (where $(\fsat q,\frisat q)$ is a
      retraction successor of $(\fsat m,\frisat m)$), resulting in
      $([ e_{q+1}, \dots ], {\Asat q}_{\downarrow \frisat q}) \in u(P)$.
  \end{enumerate}
  Afterwards, we can transform that to $([ \gsat0, c_{n+1}, \dots ], \Asat0) \in u(P)$ (as required by the satisfaction function $s$) by applying up-to techniques,
  thereby showing that the path actually describes a model.

  Note that $n=0$ or $n=q$ are possible as well.
  The initial steps are depicted in the following diagram, with the chain $(\id,)e_1,(\id,)e_2,\dots$, $n$, $\fsat0,\Asat0$ and $\gsat0$ given
  ($e_1, e_2, \dots$ are labels of existential steps and hence
  correspond to the morphisms $h_j$ of the tableau construction rules):

  
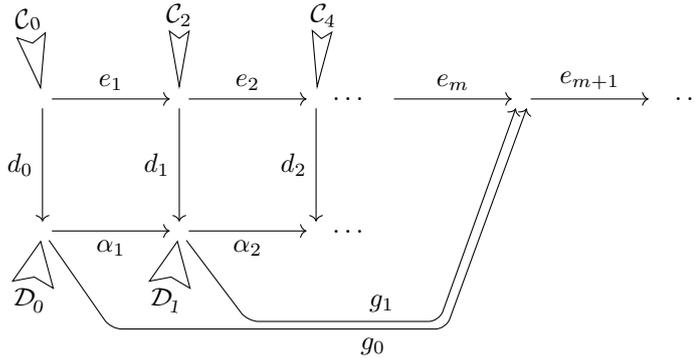
\begin{tikzpicture}[x=0.90cm]
    \foreach \i in {0,...,10} {
      \node (top\i) at (\i,0) {};
      \node (bot\i) at (\i,-1.75) {};
    }

%    \draw[->] (top0) edge node[above,align=center]{$(c_1)$\\[-2pt]$u_1$} (top1)
%              (top1) edge node[above,align=center]{$(c_2)$\\[-2pt]$e_1$} (top2)
%              (top2) edge node[above,align=center]{$(c_3)$\\[-2pt]$u_2$} (top3)
%              (top3) edge node[above,align=center]{$(c_4)$\\[-2pt]$e_2$} (top4);
    \draw[->] (top0) edge node[above,align=center]{$e_1$} (top2)
              (top2) edge node[above,align=center]{$e_2$} (top4);
    \node at ($(top4)!0.5!(top5)$) {$\dots$};
    \draw[->] (top5) edge node[above]{$e_m$} (top7)
              (top7) edge node[above]{$e_{m+1}$} (top9);
    \node at ($(top9)!0.5!(top10)$) {$\dots$};

    \draw[->] (bot0) edge node[below]{$\alpha_1$} (bot2);
    \draw[->] (bot2) edge node[below]{$\alpha_2$} (bot4);
    \node at ($(bot4)!0.5!(bot5)$) {$\dots$};

    \draw[->] (top0) -- node[left]{$\fsat0$} (bot0);
    \draw[->] (top2) -- node[left]{$\fsat1$} (bot2);
    \draw[->] (top4) -- node[left]{$\fsat2$} (bot4);

    \draw[->,rounded corners=4pt] (bot0) -- +(1.0,-1.3) -- node[below,pos=0.7828]{$\gsat0$} +(5.9,-1.3) -- (top7.-60);
    \draw[->,rounded corners=4pt] (bot2) -- +(1.0,-1.2) -- node[above,pos=0.7]{$\gsat1$} +(3.8,-1.2) -- (top7.-120);

    \node[condtri,dart tip angle=30,shape border rotate=270,rotate around={10:(top0.center)}, label={[rotate=10,anchor=south,label distance=1pt]above:{$\mathcal C_0$}}] at (top0.north) {\kern5pt};
    \node[condtri,dart tip angle=20,shape border rotate=270,rotate around={0:(top2.center)}, label={[rotate=0,anchor=south,label distance=1pt]above:{$\mathcal C_2$}}] at (top2.north) {\kern3.0pt};
    \node[condtri,dart tip angle=20,shape border rotate=270,rotate around={-5:(top4.center)}, label={[rotate=-5,anchor=south,label distance=1pt]above:{$\mathcal C_4$}}] at (top4.north) {\kern3.0pt};
    \node[condtri,shape border rotate=90, rotate around={-10:(bot0.center)}, label={[rotate=-10,anchor=north,label distance=1pt]below:{$\Asat 0$}}] at (bot0.south) {\kern7pt};
    \node[condtri,shape border rotate=90, rotate around={-10:(bot2.center)}, label={[rotate=-10,anchor=north,label distance=1pt]below:{$\Asat 1$}}] at (bot2.south) {\kern7pt};
  \end{tikzpicture}

  The proof objective is now to construct the sequence of all further
  $\forall \fsat{k+1}.\Asat{k+1}$, the associated $\alpha_{k+1}$,
  and $\gsat{k+1}$ (for $k<m$) or $\frisat{k}$ (for $k \geq m$).

  \begin{enumerate}
    \item
      We construct the aforementioned sequence by repeatedly choosing a successor of $\forall \fsat k.\Asat k$ as the next element ($\forall \fsat{k+1}.\Asat{k+1}$) of the sequence,
      then show that $\fsat m$ is a section.

      See the proof of the special case, replacing ``isomorphism'' with ``section''.
      Also instead of stopping when a section is first obtained, continue until $\fsat m$ is reached.
      Additionally, $\fsat m; \gsat m = \id$ (with $\id$ being the
      empty sequence of the $c_i$), so $\gsat m$ is a right-inverse of $\fsat m$,
      which we also choose as the starting point of the sequence of retractions (i.e., $\frisat m \defeq \gsat m$).

      Note that in the special case, we considered the first
      occurrence of an isomorphism successor as the $p$-th step (with
      $p$ possibly being less than $m$, equivalently $2p\le n$).
      Now we consider exactly the first $m$ steps, even if a section is obtained earlier already.
      As the general pull-forward rule does not ``consume'' successors (unlike the special rule) even if they are pulled forward earlier, we can still ensure that after $m$ steps, we still have a successor that can be pulled forward.
      % notes about this^: in special case we couldn't use "after n steps" because an iso might be pulled forward immediately, and then there would not be a successor after n steps that could still be pulled forward, because the rule only adds the *other* children on the right side.
      % in general case, it will be re-added again, so even if it has been pulled forward earlier, there will still be a successor after n steps.
      % note: this argument will break if we ever modify the definition of the LTS such that it doesn't re-add the same child with the same inverse again.

    \item
      Summarizing the situation so far, $\fsat 0$ has been turned into a section $\fsat m$ and a chosen right-inverse $\frisat m$ (such that $\fsat m ; \frisat m = \id$) by now.

      We use the retraction successor relation as described in \Cref{lts-section-successor-relation}
      to obtain, for $n \leq k \leq q$, the sequence of subsequent
      $\alpha_{k+1}, \fsat{k+1}, \frisat{k+1}$ (named $\alpha, \beta,
      r_\beta$, respectively, in \Cref{lts-section-successor-relation}).

      Note that \Cref{rsq-preserve-sections} guarantees the existence
      of a section $\fsat{k+1}$ in the next step. Furthermore we also
      obtain retraction successors, namely retractions $\frisat{k}$,
      $\frisat{k+1}$ such that
      $\alpha_{k+1} ; \frisat{k+1} = \frisat{k} ; e_{k+1}$ (see
      diagram below).

            
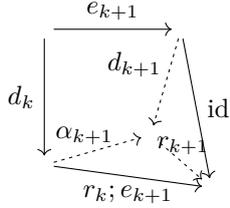
\begin{tikzpicture}
    \begin{scope}[shift={(3,0)}]
      \node (a) at (0,0) {};
      \node (b) at (1.8,0) {};
      \node (c) at (0,-1.8) {};
      \node (d) at (1.4,-1.4) {};
      \node (ds) at (2.2,-2.1) {};
      \draw[->] (a) -- node[above]{$e_{k+1}$} (b);
      \draw[->] (a) -- node[left]{$d_k$} (c);

      \draw[->,fancydotted] (b) -- node[left,pos=0.3]{$d_{k+1}$} (d);
      \draw[->,fancydotted] (c) -- node[above,pos=0.35]{$\alpha_{k+1}$} (d);

      \draw[->,fancydotted] (d) -- (ds);
      \node at ($(d.center)+(12pt,-4pt)$) {$r_{k+1}$};
      \draw[->] (b) -- node[right]{$\id$} (ds);
      \draw[->] (c) -- node[below]{$r_k;e_{k+1}$} (ds);
    \end{scope}
  \end{tikzpicture}

\item Since $\fsat m$ is a section with retraction $\frisat m$, one of
  its retraction successor must eventually pulled forward if we make
  an appropriate choice of representative squares
  above.
     From the arguments above, we know that the path is infinite and the fairness constraint guarantees the existence of a finite $q$.
%      If the path is finite, the length of the path is an upper bound for the value of $q$ (see the corresponding proof part of the special case).
      % \oldtodo{\textbf{L:} we also have the retraction successors now,
      % though. Being finite just means that all sections appearing in
      % the conditions must be processed somehow before the path can
      % end, but it might be the wrong inverse.  (It might be the case
      % that a path that ends up containing sections will always be
      % infinite, since the LTS keeps re-adding essentially the same
      % condition -- unlike in the special case.)}

     We end up with a situation depicted below, the condition
     $\mathcal{C}_q$ at the current node being marked by $*$. Note
     that the upper row of squares commutes by construction, in
     particular, $d_i;\alpha_{i+1} = e_{i+1};d_{i+1}$, while this is
     not necessarily true for the lower row of
     squares.

      
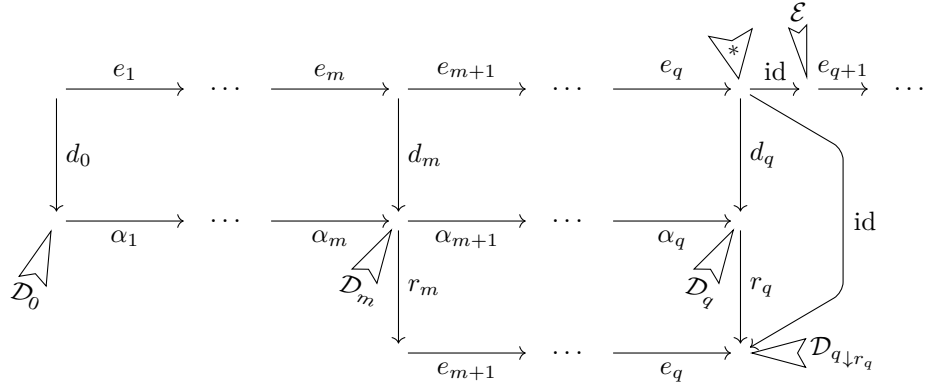
\begin{tikzpicture}[x=0.90cm]
        \foreach \i in {0,...,13} {
          \node (top\i) at (\i,0) {};
          \node (bot\i) at (\i,-1.75) {};
          \node (bbot\i) at (\i,-3.5) {};
        }

        \draw[->] (top0) edge node[above]{$e_1$} (top2);
        \node at ($(top2)!0.5!(top3)$) {$\dots$};
        \draw[->] (top3) edge node[above]{$e_m$} (top5)
                  (top5) edge node[above]{$e_{m+1}$} (top7);
        \node at ($(top7)!0.5!(top8)$) {$\dots$};
        \draw[->] (top8) edge node[above]{$e_q$} (top10);
        \draw[->] (top10) -- node[right]{$\fsat{q}$} (bot10);

        \draw[->] (top10) -- node[above]{$\id$} (top11);
        \draw[->] (top11) -- node[above]{$e_{q+1}$} (top12);
        \node at ($(top12)!0.5!(top13)$) {$\dots$};

        \draw[->] (bot0) edge node[below]{$\alpha_1$} (bot2);
        \node at ($(bot2)!0.5!(bot3)$) {$\dots$};
        \draw[->] (bot3) edge node[below]{$\alpha_m$} (bot5);
        \draw[->] (bot5) edge node[below]{$\alpha_{m+1}$} (bot7);
        \node at ($(bot7)!0.5!(bot8)$) {$\dots$};
        \draw[->] (bot8) edge node[below]{$\alpha_{q}$} (bot10);

        \draw[->] (top0) -- node[right]{$\fsat0$} (bot0);
        \draw[->] (top5) -- node[right]{$\fsat m$} (bot5);

        \draw[->] (bot5) -- node[right]{$\frisat m$} (bbot5);
        \draw[->] (bot10) -- node[right]{$\frisat q$} (bbot10);

        \draw[->] (bbot5) edge node[below]{$e_{m+1}$} (bbot7);
        \node at ($(bbot7)!0.5!(bbot8)$) {$\dots$};
        \draw[->] (bbot8) edge node[below]{$e_q$} (bbot10);

        \draw[->,rounded corners=4pt] (top10) -- +(1.5,-0.8) -- node[right,pos=0.5]{$\id$} ($(bbot10)+(1.5,0.8)$) -- (bbot10);

        \node[condtri,dart tip angle=30,shape border rotate=90, rotate around={-20:(bot0.center)}, label={[rotate=-20,anchor=north,label distance=1pt]below:{$\Asat 0$}}] at (bot0.-100) {\kern5pt};
        \node[condtri,dart tip angle=30,shape border rotate=90, rotate around={-30:(bot5.center)}, label={[rotate=-20,anchor=north,label distance=1pt]below:{$\Asat m$}}] at (bot5.-100) {\kern5pt};
        \node[condtri,dart tip angle=30,shape border rotate=90, rotate around={-30:(bot10.center)}, label={[rotate=-20,anchor=north,label distance=1pt]below:{$\Asat q$}}] at (bot10.-100) {\kern5pt};
        \node[condtri,shape border rotate=270,rotate around={10:(top10.center)},scale=0.9] at (top10.north) {$*$};
        \node[condtri,dart tip angle=20,shape border rotate=90, rotate around={190:(top11.center)}, label={[rotate=0,anchor=south,label distance=1pt]below:{$\mathcal{E}$}}] at (top11.south) {\kern2.5pt};
        \node[condtri,dart tip angle=30,shape border rotate=90, rotate around={90:(bbot10.center)}, label={[rotate=0,anchor=west,label distance=1pt]below:{${\Asat q}_{\downarrow \frisat q}$}}] at (bbot10.south) {\kern5pt};
      \end{tikzpicture}

      Now let $q$ be the step at which a successor $\fsat q$ of
      $\forall \fsat m . \Asat m$ is finally pulled forward using
      inverse $\frisat q$.  Let
      $\mathcal{C}_q = \bigland_\ell \forall f_\ell.\mathcal A_\ell$
      and $\forall \fsat q.\Asat q = \forall f_t.\mathcal{A}_t$ for
      some index $t$. Assume that
      $(\Asat q)_{\downarrow \frisat q} = \biglor_j \exists
      h_j. \mathcal H_j$, then the next step in the tableau is:

      \[
        \mathcal{C}_q = \Big( \bigland_{\kern10pt\mathclap{\ell \ne
            t}\kern10pt} \forall f_\ell . \mathcal A_\ell \Big) \land
        \forall \fsat q . \Asat q \xrightarrow{\id} \overbrace{
          \underbrace{\biglor_j \exists h_j . \Big( \mathcal
            H_j}_{{\Asat q}_{\downarrow \frisat q}} \land \Big(
          \bigland_{i} \forall f_i . \mathcal A_i \Big)%
          _{\downarrow h_j} \Big) }^{\mathcal E}
      \]

      As a result, we have:
      \[ ( [ e_{q+1},\ \id, e_{q+2}, \dots ], \mathcal{E}
      % \biglor_j \exists h_j . \Big( \mathcal H_j \land
      % \Big( \bigland_{i} \forall f_i . \mathcal A_i \Big)_{\downarrow h_j}
      % \Big)
      ) \in P \]

  \end{enumerate}

  We now apply our up-to techniques to obtain the originally desired
  result (which was to satisfy a given $\forall \fsat0.\Asat0$ by
  showing for a given $n$ and $\gsat0$ such that
  $c_1 ; \dots ; c_n = \fsat0 ; \gsat0$ that
  $([ \gsat0, c_{n+1}, \dots ], \Asat0) \in u(P)$ holds).

%      Let $\mathcal E$ be the condition that was last reached on the path of $P$.
      Observe that ${\Asat q}_{\downarrow \frisat q}$ is ``contained'' in $\mathcal E$,
      in the sense that the path condition $\mathcal E$ is a (nested) conjunction of ${\Asat q}_{\downarrow \frisat q}$ and additional conditions.
      First, we apply up-to conjunction removal (\Cref{thm:upto-conjunction-removal})  to turn our condition $\mathcal E$
      (``${\Asat q}_{\downarrow \frisat q}$ with conjunctions'')
      into just ${\Asat q}_{\downarrow \frisat q}$:

      \[ \implies
      \Big( [ e_{q+1}, \id, e_{q+2}, \dots ],
      \biglor_j \exists h_j . \mathcal H_j
      \Big) =
      \Big( [ e_{q+1}, \id, e_{q+2}, \dots ], {\Asat q}_{\downarrow \frisat q} \Big) \in u(P) \]

      Since $\Asat q = {\Asat{q-1}}_{\downarrow \alpha_q} = \ldots = {\Asat0}_{\downarrow \alpha_1 \downarrow \dots \downarrow \alpha_q}$,
      we can now apply a series of up-to shift operations (\Cref{thm:compat-uptoshift}) to obtain a tuple with $\Asat0$ as the condition:

      \begin{align*}
        ([ e_{q+1}, \dots ], &\ {\Asat q}_{\downarrow \frisat q} ) \in u(P) \\
        \implies
        ([ \frisat q ; e_{q+1}, \dots ], &\ \Asat q ) = \\
        ([ \frisat q ; e_{q+1}, \dots ], &\ {\Asat0}_{\downarrow \alpha_1 \downarrow \dots \downarrow \alpha_q} ) \in u(P) \\
        \implies
        ([ \alpha_q ; \frisat q ; e_{q+1}, \dots ], &\ {\Asat0}_{\downarrow \alpha_1 \downarrow \dots \downarrow \alpha_{q-1}} ) \in u(P) \\
        \implies \dots \implies
        ([ \alpha_1 ; \dots ; \alpha_q ; \frisat q ; e_{q+1}, \dots ], &\ \Asat0) \in u(P)
      \end{align*}

      Now apply previous results to rewrite the initial part of the
      chain.  Remember that we obtain
      $\alpha_{k+1} ; \frisat{k+1} = \frisat{k} ; e_{k+1}$ from the
      choice of retraction successors.
      
  Hence 
      \[ \alpha_{n+1};\dots;\alpha_q;\frisat q =
        \alpha_{n+1};\dots;\alpha_{q-1};\frisat {q-1}; e_q = \dots = \frisat
      m;e_{m+1};\dots;e_q, \]
      therefore
      \[ ([\alpha_1 ; \dots ; \alpha_m ; \frisat m ; e_{m+1} ; \dots ; e_q ; e_{q+1}, \dots], \Asat0) \in u(P) \]

      Next, we exploit the way in which the $\gsat i$ were constructed
      (with the last one being $\gsat m = \frisat m$). In particular
      $\alpha_{i+1};g_{i+1} = g_i$, hence:
      \[
        \alpha_1 ; \dots ; \alpha_n ; \frisat m
        = \alpha_1 ; \dots ; \alpha_m ; \gsat m
        = \alpha_1 ; \dots ; \alpha_{m-1} ; \gsat{m-1}
        = \ldots = \gsat0
      \]

      As a result, we have
      \[
        ([ \gsat0 ; e_{m+1} ; \dots ; e_q ; e_{q+1}, \dots ], \Asat0) \in u(P)
      \]

      With up-to recomposition (\Cref{thm:compat-uptorecomp}), we can adjust this to
      \[
        ([ \gsat0, \id, e_{m+1}, \id, \dots, e_q, \id, e_{q+1}, \dots ], \Asat0) \in u(P)
      \]
      which corresponds, as desired, to
      \[
        ([ \gsat0, c_{n+1}, c_n, \dots ], \Asat0) \in u(P).
        \qedhere
      \]
\end{proofparts}
\end{proof}

\begin{theorem}[Soundness and
  Completeness]\label{thm-algo-sound-complete-general}
  ~
  \begin{itemize}
  \item If all branches in a tableau are closed, then the
    condition in the root node is unsatisfiable.
  \item If a condition $\mathcal A$ is unsatisfiable, then in every
    tableau constructed by obeying the fairness constraint all
    branches are closed.
  \end{itemize}
\end{theorem}

\begin{proof}
  Use the proof strategies of \Cref{thm-algo-complete,thm-algo-sound}.
  Tableaux constructed by the rules of \Cref{satcheck.new.general} have all properties
  that are required for the proofs (in particular, universal steps
  lead to an equivalent condition).  Use
  \Cref{thm:fair-branch-model-general} to obtain models for open
  branches.
\end{proof}

% With this method one easily obtains an unsatisfiability proof for the second condition from
% \Cref{ex-section-motivation}.
% basically, by adapting exa-prove-unsat
% \oldtodo{Add an example that it is really possible to find an unsat proof
%   for an unsat condition that contains sections}

In the general case, although soundness and completeness still hold,
we are no longer able to find all finite models as before. This is
intuitively due to the fact that a condition might have a finite
model, but what we find is a seemingly infinite model that always has
the potential to collapse to the finite model.
%\oldtodo{\textbf{B:}
%  give a concrete example? I do not think that this is directly
%  connected to the fact that we do not remove subconditions when
%  pulling forward.}

On the other hand, the weaker requirements on the category imply that more instantiations are possible, such as to $\graphf$ (graphs and arbitrary morphisms). Here, pushouts can be used for representative squares, which allows for a more efficient shift operation that avoids the blowup of the size of the conditions that is associated with jointly epi squares.

% \paragraph*{Possible optimizations}

% If a section to be pulled forward is also an iso, it would be possible
% to do a step as in the specialized case (cf.\ %
% \Cref{lem:pull-forward-isos} which is valid even if not all sections
% are isos).  This might help to slightly reduce the size of the
% resulting condition.  However, since many interesting morphisms are
% sections, a rule to pull forward sections that are not isos is still
% required.

% As many interesting conditions have sections that are not isos, for
% most finitely-satisfiable conditions there are no finite open tableaux
% that would prove the existence of a model.  Still, we can re-use
% existing results and obtain finite models in some cases: Whenever the
% tableau has been extended, check if $\id$ satisfies the condition in
% the leaf, and if so, return the labels on the branch as a model.

% This does not always allow finding finite models, though: consider the
% example condition: there exists a node and every node has an
%   outgoing edge to another node.  It is finitely satisfiable:
% directed circles of length at least two (if all sections are isos) and
% in the general case also by a single node with a loop.  The general
% algorithm will not find either of these models.

\section{Conclusion}
\label{sec:conclusion}

% We have introduced satisfiability checks for nested conditions based
% on \coinductive techniques. Naturally, there is a large choice of
% methods for first-order logic \cite{f:fol-theorem-proving}, but it is
% not always clear how to generalize such methods to general graph-like
% structures or -- as done in this paper -- to arbitrary categories.
% %
% Our paper uses concepts from
% \cite{lo:tableau-graph-properties,slo:model-generation,nopl:navigational-logic}
% and extends them from graphs to a generic categorical setting, stating
% the \coinductive methods explicitly.

We introduced a semi-decision procedure for checking satisfiability of
nested conditions at the general categorical level.  The correctness
of this tableau-based procedure has been established using a novel
combination of \coinductive (up-to) techniques. In the restricted case
we also considered the generation of finite models and witnesses for
(some) infinite models.  Our procedure thereby generalizes prior work
\cite{lo:tableau-graph-properties,slo:model-generation,nopl:navigational-logic}
on nested graph conditions that are equivalent to first-order logic
\cite{f:fol-theorem-proving}.  
As a result, we can also handle cospan categories over adhesive
categories (using borrowed context diagrams for representative
squares) and other categories, such as Lawvere theories.  

There is a notion of Q-trees \cite{fs:categories-allegories}
reminiscent of the nested conditions studied in this paper, but to our
knowledge no generic satisfiability procedures have been derived for
Q-trees.

We plan to transfer the technique of counterexample-guided abstraction
refinement (CEGAR) \cite{hjmm:abstractions-proofs}, a program analysis
technique based on abstract interpretation and predicate abstraction,
to graph transformation and reactive systems. The computation of
weakest preconditions and strongest postconditions for nested
conditions is fairly straightforward
\cite{bchk:conditional-reactive-systems} and satisfiability checks
give us the necessary machinery to detect and eliminate spurious
counterexamples. One still has to work around undecidability issues
and understand whether there is a generalization of Craig
interpolation, used to simplify conditions.

One further direction for future work is to understand the mechanism
for witness generation in more detail. In particular, since it is
known that FOL satisfiability for graphs of bounded treewidth is
decidable \cite{c:mso-graphs-I}, the question arises whether we can
find witnesses for all models of bounded treewidth (or a suitable
categorical generalization of this
notion).

Another direction is to further explore instantiation with a Lawvere theory
\cite{l:functorial-semantics-algebraic-theories}, where arrows are
$n$-tuples of $m$-ary terms. In this setting representative squares
are closely related to unification.
\Cref{appendix-terms} contains some initial results.

Finally we plan to complete development of a tool that implements the
satisfiability check and explore the potential for optimizations
regarding its runtime.

\bibliographystyle{plain}
\bibliography{references}

\begin{thebibliography}{10}

\bibitem{bchk:conditional-reactive-systems}
{H.J.}~Sander Bruggink, Rapha{\"e}l Cauderlier, Mathias H{\"u}lsbusch, and
  Barbara K{\"o}nig.
\newblock Conditional reactive systems.
\newblock In {\em Proc. of FSTTCS~'11}, volume~13 of {\em {LIPIcs}}. Schloss
  Dagstuhl -- Leibniz Center for Informatics, 2011.

\bibitem{c:mso-graphs-I}
Bruno Courcelle.
\newblock The monadic second-order logic of graphs~{I}. {R}ecognizable sets of
  finite graphs.
\newblock {\em Information and Computation}, 85(1):12--75, 1990.

\bibitem{eept:fundamentals-agt}
Hartmut Ehrig, Karsten Ehrig, Ulrike Prange, and Gabriele Taentzer.
\newblock {\em Fundamentals of Algebraic Graph Transformation}.
\newblock Monographs in Theoretical Computer Science. Springer, 2006.

\bibitem{ek:congruence-dpo-journal}
Hartmut Ehrig and Barbara K{\"o}nig.
\newblock Deriving bisimulation congruences in the {DPO} approach to graph
  rewriting with borrowed contexts.
\newblock {\em MSCS}, 16(6):1133--1163, 2006.

\bibitem{eps:gragra-algebraic}
Hartmut Ehrig, Michael Pfender, and Hans~Jürgen Schneider.
\newblock Graph grammars: An algebraic approach.
\newblock In {\em Proc. 14th IEEE Symp. on Switching and Automata Theory},
  pages 167--180, 1973.

\bibitem{f:fol-theorem-proving}
Melvin Fitting.
\newblock {\em First-Order Logic and Automated Theorem Proving}.
\newblock Springer, 1996.

\bibitem{fs:categories-allegories}
Peter~J. Freyd and Andre Scedrov.
\newblock {\em Categories, Allegories}.
\newblock North-Holland, 1990.

\bibitem{hp:correctness-nested-conditions}
Annegret Habel and Karl-Heinz Pennemann.
\newblock Correctness of high-level transformation systems relative to nested
  conditions.
\newblock {\em Mathematical Stuctures in Computer Science}, 19(2):245--296,
  2009.

\bibitem{hjmm:abstractions-proofs}
Thomas~A. Henzinger, Ranjit Jhala, Rupak Majumdar, and Kenneth~L. McMillan.
\newblock Abstractions from proofs.
\newblock In {\em Proc. of POPL '04}, pages 232--244. ACM, 2004.

\bibitem{DBC-CRS}
Mathias H\"{u}lsbusch and Barbara K\"{o}nig.
\newblock Deriving bisimulation congruences for conditional reactive systems.
\newblock In {\em Proc. of FOSSACS '12}, pages 361--375. Springer, 2012.
\newblock {LNCS/ARCoSS} 7213.

\bibitem{lmcs:8951}
Mathias Hülsbusch, Barbara König, Sebastian Küpper, and Lara Stoltenow.
\newblock {Conditional Bisimilarity for Reactive Systems}.
\newblock {\em {Logical Methods in Computer Science}}, {Volume 18, Issue 1},
  January 2022.

\bibitem{ls:adhesive-journal}
Stephen Lack and Pawe{\l} Soboci\'{n}ski.
\newblock Adhesive and quasiadhesive categories.
\newblock {\em RAIRO -- Theoretical Informatics and Applications}, 39(3), 2005.

\bibitem{lo:tableau-graph-properties}
Leen Lambers and Fernando Orejas.
\newblock Tableau-based reasoning for graph properties.
\newblock In {\em Proc. of ICGT '14}, pages 17--32. Springer, 2014.
\newblock {LNCS} 8571.

\bibitem{l:functorial-semantics-algebraic-theories}
William Lawvere.
\newblock {\em Functorial Semantics of Algebraic Theories}.
\newblock PhD thesis, Columbia University, 1963.

\bibitem{l:congruences-reactive}
James~J. Leifer.
\newblock {\em Operational congruences for reactive systems}.
\newblock PhD thesis, University of Cambridge Computer Laboratory, September
  2001.

\bibitem{lm:derive-bisimulation}
James~J. Leifer and Robin Milner.
\newblock Deriving bisimulation congruences for reactive systems.
\newblock In {\em Proc. of CONCUR 2000}, pages 243--258. Springer, 2000.
\newblock {LNCS} 1877.

\bibitem{nopl:navigational-logic}
Marisa Navarro, Fernando Orejas, Elvira Pino, and Leen Lambers.
\newblock A navigational logic for reasoning about graph properties.
\newblock {\em Journal of Logical and Algebraic Methods in Programming},
  118:100616, 2021.

\bibitem{p:algorithm-approximating-satisfiability}
Karl-Heinz Pennemann.
\newblock An algorithm for approximating the satisfiability problem of
  high-level conditions.
\newblock In {\em GT-VC@CONCUR}, volume 213.1 of {\em Electronic Notes in
  Theoretical Computer Science}, pages 75--94. Elsevier, 2007.

\bibitem{p:development-correct-gts}
Karl-Heinz Pennemann.
\newblock {\em Development of Correct Graph Transformation Systems}.
\newblock PhD thesis, Universit\"at Oldenburg, May 2009.

\bibitem{p:inductive-generalization}
Gordon~D. Plotkin.
\newblock A note on inductive generalization.
\newblock {\em Machine Intelligence}, (5):153--163, 1970.

\bibitem{p:complete-lattices-up-to}
Damien Pous.
\newblock Complete lattices and up-to techniques.
\newblock In {\em Proc. of APLAS '07}, pages 351--366. Springer, 2007.
\newblock {LNCS} 4807.

\bibitem{pous:thesis}
Damien Pous.
\newblock {\em {Techniques modulo pour les bisimulations}}.
\newblock PhD thesis, {ENS Lyon}, February 2008.

\bibitem{ps:enhancements-coinductive}
Damien Pous and Davide Sangiorgi.
\newblock Enhancements of the coinductive proof method.
\newblock In Davide Sangiorgi and Jan Rutten, editors, {\em Advanced Topics in
  Bisimulation and Coinduction}. Cambridge University Press, 2011.

\bibitem{r:representing-fol}
Arend Rensink.
\newblock Representing first-order logic using graphs.
\newblock In {\em Proc. of ICGT '04}, pages 319--335. Springer, 2004.
\newblock {LNCS} 3256.

\bibitem{s:bisimulation-coinduction}
Davide Sangiorgi.
\newblock {\em Introduction to Bisimulation and Coinduction}.
\newblock Cambridge University Press, 2011.

\bibitem{ss:reactive-cospans}
Vladimiro Sassone and Pawe{\l} Soboci\'{n}ski.
\newblock Reactive systems over cospans.
\newblock In {\em Proc. of LICS '05}, pages 311--320. IEEE, 2005.

\bibitem{slo:model-generation}
Sven Schneider, Leen Lambers, and Fernando Orejas.
\newblock Symbolic model generation for graph properties.
\newblock In {\em Proc. of FASE '17}, pages 226--243. Springer, 2017.
\newblock {LNCS} 10202.

\bibitem{s:deriving-congruences}
Pawe{\l} Soboci{\'{n}}ski.
\newblock {\em Deriving process congruences from reaction rules}.
\newblock PhD thesis, Department of Computer Science, University of Aarhus,
  2004.

\bibitem{sksclo:coinductive-satisfiability-nested-conditions}
Lara Stoltenow, Barbara K\"onig, Sven Schneider, Andrea Corradini, Leen
  Lambers, and Fernando Orejas.
\newblock Coinductive techniques for checking satisfiability of generalized
  nested conditions.
\newblock In {\em Proc. of CONCUR '24}, {LIPIcs}. Schloss Dagstuhl -- Leibniz
  Center for Informatics, 2024.
\newblock to appear.

\bibitem{t:lattice-fixed-point}
Alfred Tarski.
\newblock A lattice-theoretical fixpoint theorem and its applications.
\newblock {\em Pacific Journal of Mathematics}, 5:285--309, 1955.

\end{thebibliography}
\renewcommand{\appendixbibliographystyle}{plain}

%\appendix

\nosectionappendix
\begin{toappendix}

  \section{Instantiation: Lawvere Theories}
  \label{appendix-terms}
  
  %  \subsection{Weak pushouts in Lawvere Theories}
  
  \newcommand{\ol}[1]{\overline{#1}}
  \renewcommand{\ol}[1]{{#1}}
  
  As a further instance of our generic procedure for checking
  (un)satisfiability of nested conditions we consider Lawvere theories
  \cite{l:functorial-semantics-algebraic-theories}, which provide a
  syntactic model of first-order algebraic signatures. After the basic
  definitions we show that the theory of a signature has weak pushouts
  by providing an explicit construction. Such squares can be taken as
  the representative squares necessary for our procedure. We conclude
  by presenting a simple nested condition, and showing that it is not
  satisfiable by using the general version of our procedure (see
  \Cref{sec:general-case}): this is needed because, in general,
  sections are not isomorphisms in these categories.

  % We briefly recall the definition of algebraic theory over a signature, introduced in
  % Lawvere’s thesis~\cite{l:functorial-semantics-algebraic-theories}. As Kock
  % and Reyes summarize, the right way of conceiving the totality of operations for an
  % equational theory was found by Lawvere, who realized that substitution should be
  % viewed as the composition of arrows on a certain kind of category (A. Kock, G.E. Reyes, 
  % Doctrines in categorical logic, in: J. Barwise (Ed.), Handbook of Mathematical
  % Logic, North-Holland, Amsterdam, 1977, pp. 283-313.).
  
  A (non-equational, one-sorted) signature $\Sigma = \{\Sigma_n\}_{n\in \natzero}$ is a finite ranked set of function symbols.
  A $\Sigma$-term $t$ over variables $\{x_1, x_2, \ldots, x_n\}$, with $n \geq 0$, (written $t \in T_{\Sigma_n}$) is either a variable $x_i$ with $i \in \{1, \ldots, n\}$, or $t = f(t_1, \ldots, t_k)$, where $f \in \Sigma_k$ and $t_i \in T_{\Sigma_n}$ for each $i \in \{1, \ldots, k\}$. 
  A substitution $\sigma$ from $n$ to $m$ is a function from  $\{x_1, x_2, \ldots, x_n\}$ to $T_{\Sigma_m}$.
  If $t \in T_{\Sigma_n}$ and $\sigma$ is a substitution from $n$ to $m$, then $\sigma(t) \in T_{\Sigma_m}$ is defined in the usual way.

  \begin{definition}[cartesian category]
      A category is cartesian if it has the (categorical) product for each pair of objects. 
      Given $A,B \in Obj(C)$, their product is a triple $\langle \pi_A, A\times B, \pi_B\rangle$ such that for all objects $C$ and arrows $f: C \to A, g:C \to B$, there is a unique arrow $\langle f,g\rangle: C \to A \times B$ such that $\pi_A \circ \langle f,g\rangle = f$ and  $\pi_B \circ \langle f,g\rangle = g$. 
  \end{definition}
  
  \begin{definition}[Lawvere Theory]
   The (Lawvere) theory of a signature $\Sigma$ is the cartesian category $\mathit{Th}_{\Sigma}$ defined equivalently as: 
   \begin{itemize}
      \item The free cartesian category having natural numbers as objects (where the categorical product is addition), generated by the set of arrows $\bigcup_{n\in \natzero}\{f \colon n \to 1 \mid f \in \Sigma_n\}$. Object $0$ is the terminal object.
      
      Note that by the properties of categorical products there is a bijection between arrows from $n$ to $m$ in $\mathit{Th}_{\Sigma}$ and $m$-tuples of arrows from $n$ to $1$, since $m = 1\times \ldots \times 1$. 
  
      \item The category having natural numbers as objects, and where each arrow $\ol{t} \colon n \to m$ is an $m$-tuple of terms $\ol{t}= \langle t_1, \ldots, t_m\rangle$ in $T_{\Sigma_n}$, i.e.~with variables in $\{x_1, \ldots, x_n\}$.
      
      Given arrows $\ol{s} = \langle s_1, \ldots, s_k\rangle \colon n \to k$ and $\ol{t} = \langle t_1, \ldots, t_m\rangle \colon k \to m$, their composition is $\ol{t}\circ \ol{s} = \langle \sigma(t_1), \ldots, \sigma(t_m) \rangle \colon n \to m$, where $\sigma$ is the substitution defined as $\sigma(x_i) = s_i$ for each $i \in \{1,\ldots,k\}$.
  
      For each $n \in \natzero$, the only arrow to $0$ is the empty
      tuple $\langle\rangle$, often denoted $!$. For each pair
      $n,m \in \natzero$, their product is $n+m$ with projections
      $\pi_n = \langle x_1,\ldots, x_n\rangle\colon n+m\to n$ and
      $\pi_m = \langle x_{n+1},\ldots, x_m\rangle\colon n+m\to m$.
   \end{itemize}
  \end{definition}
  
  For the construction of the weak pushout of two arrows below we need the following technical lemma. 
  
  \begin{lemma}[finite factorization of arrows]
    \label{lem:fin_fac_in_Law_theory}
  \floatingpicspaceright{3.2cm}
    \begin{floatingpicright}{2.9cm}%
      \hfill\begin{tikzpicture}[baseline={(0,-0.2)}]
        \path[use as bounding box] (-0.2,0) rectangle (1.6,-0.75);
        \node (sqtl) at (0,0) {$k$};
        \node (sqtr) at (1.25,0) {$n$};
        \node (sqbl) at (0,-0.9) {$m$};
       % \node (sqbr) at (1.25,-0.9) {m};
  
        \draw[->] (sqtl) -- node[overlay,above]{$\ol{t}$} (sqtr);
        \draw[->] (sqtl) -- node[left]{$\ol{s}$} (sqbl);
        \draw[->,dashed] (sqbl) -- node[below]{$\ol{f}$} (sqtr);
       % \draw[->] (sqbl) -- node[overlay,below]{$\beta$} (sqbr);
      \end{tikzpicture}
    \end{floatingpicright}
    Given arrows $\ol{t} \colon k \to n$ and $\ol{s} \colon k \to m$ in $\mathit{Th}_\Sigma$, there is at most a finite number of arrows $\ol{f}: m \to n$ such that $\ol{f} \circ \ol{s} = \ol{t}$.
  \end{lemma}
  
  \begin{proof}
    Since arrows to $p$ are $p$-tuples of arrows to $1$, it is sufficient to prove the statement for $n = 1$, thus for $\ol{t}\in T_{\Sigma_k}$. We proceed by induction on the structure of $\ol{t}$.
    \begin{enumerate}
      \item $\ol{t} = x_i$, with $i \in \{1,\ldots,k\}$. From $x_i = \ol{t} =  \ol{f} \circ \ol{s}$ we infer that $\ol{f}$ must be a variable in $\{x_1,\ldots,x_m\}$, which is a finite set.
      
      \item  $\ol{t} = h(t_1,\ldots,t_r)$, with $t_1,\ldots,t_r \in T_{\Sigma_m}$. In this case a decomposition of $\ol{t}$ as $\ol{f} \circ \ol{s}$ can be obtained in two ways:
      \begin{enumerate}
        \item  Either $\ol{f}$ is a variable in the finite set $\{x_1,\ldots,x_m\}$ (and $\ol{t}$ is a term appearing in the $m$-tuple $\ol{s}$), or 
        \item $\ol{f} = h \circ \langle f_1, \ldots,f_r\rangle$, where for each $i \in \{1,\ldots,r\}$ it holds $t_i = f_i \circ \ol{s}$.  By induction hypothesis there are only finitely many such $f_i$ for each $i$, and therefore there are finitely many such tuples, which concludes the proof.
      \end{enumerate} 
  
  \[
    \begin{tikzpicture}%[x=0.90cm]
      %[baseline={(0,0.2)}]
          %  \path[use as bounding box] (0,0) rectangle (2.5,-1.00);
            \node (sqtl) at (0,0) {$k$};
            \node (sqtr) at (3,0) {$1$};
            \node (sqbl) at (0,-1) {$m$};
            \node (sqbr) at (3,-1) {$r$};
      
            \draw[->] (sqtl) -- node[overlay,above]{$\ol{t}=h(t_1,\ldots,t_r)$} (sqtr);
            \draw[->] (sqtl) -- node[left]{$\ol{s}$} (sqbl);
            \draw[->] (sqbl) -- node[below]{$\langle f_1, \ldots, f_r\rangle$} (sqbr);
            \draw[->] (sqbr) -- node[right]{$h$} (sqtr);
          \end{tikzpicture}
    \]
    \end{enumerate}
  \end{proof}

  % \ol{t} = \langle t_1, \dots, t_n\rangle
  %\ol{s} = \langle s_1, \dots, s_m\rangle:k \to n
  
  \begin{theorem}[$\mathit{Th}_{\Sigma}$ has weak pushouts]
      \label{th:wp-LawTh}
      The Lawvere Theory of a signature,   $\mathit{Th}_{\Sigma}$, has weak pushouts. That is, 
      for each pair of arrows $\ol{t} \colon k \to n$ and $\ol{s} \colon k \to m$ there exists an object $z$ and arrows $\ol{f} \colon n \to z$ and $\ol{g} \colon m \to z$ such that (1) $\ol{f}\circ \ol{t} = \ol{g} \circ \ol{s}$ and (2) for each 
      $\langle z', \ol{c} \colon n \to z', \ol{d} \colon m \to z'\rangle$ such that  $\ol{c}\circ \ol{t} = \ol{d} \circ \ol{s}$ there is an arrow $\ol{e} \colon z \to z'$ such that $\ol{d} = \ol{e} \circ \ol{g}$ and $\ol{c} = \ol{e} \circ \ol{f}$. Note that uniqueness of arrow $\ol{e}$ is not required.
  \end{theorem}

  \[
    \begin{tikzpicture}%[x=0.90cm]
      %[baseline={(0,0.2)}]
          %  \path[use as bounding box] (0,0) rectangle (2.5,-1.00);
      \node (sqtl) at (0,0) {$k$};
      \node (sqtr) at (2,0) {$n$};
      \node (sqbl) at (0,-1) {$m$};
      \node (sqbr) at (2,-1) {$z$};
      \node (bbrr) at (3,-2) {$z'$};
  
      \draw[->] (sqtl) -- node[above]{$\ol{t}$} (sqtr);
      \draw[->] (sqtl) -- node[left]{$\ol{s}$} (sqbl);
      \draw[->] (sqbl) -- node[above]{$\ol{g}$} (sqbr);
      \draw[->] (sqtr) -- node[right]{$\ol{f}$} (sqbr);
      \draw[->] (sqtr) to [out = 0, in = 90] node[right]{$c$} (bbrr);
      \draw[->] (sqbl) to [out = 270, in = 180] node[below]{$d$} (bbrr);
      \draw[->] (sqbr) -- node[right]{$e$} (bbrr);
    \end{tikzpicture}
  \]
  
  \begin{proof}
  Given two arrows $\ol{t} \colon k \to n$ and $\ol{s} \colon k \to m$, let us proceed as follows.
  \begin{enumerate}
      \item For each $i \in \{1,\ldots,n\}$ let $G_i = \{g \colon m \to 1 \mid
        g \circ \ol{s} = t_i\}$ be the set of all terms $g$ such that
        applying the substitution $\ol{s}$ to it we get
        $t_i$. By \Cref{lem:fin_fac_in_Law_theory} $G_i$ is finite. Let $\alpha_i = |G_i|$ be its cardinality, and $G_i = \{g_{i1}, \ldots, g_{i{\alpha_i}}\}$. Note that for each $j \in \{1,\ldots,\alpha_i\}$ we have $t_i = \langle x_i\rangle  \circ \ol{t} = g_{ij} \circ \ol{s}$, i.e.~we have a commuting square closing the span $(\ol{t},\ol{s})$.
      
      \item Symmetrically, for each $i \in \{1,\ldots,m\}$ let $F_i = \{f \colon n \to 1 \mid f \circ \ol{t} = s_i\}$ be the set of all terms $f$ such that applying the substitution $\ol{t}$ to it we get $s_i$. Let $\beta_i = |F_i|$ be its cardinality, and $F_i = \{f_{i1}, \ldots f_{i{\beta_i}}\}$. Note that for each $j \in \{1,\ldots,\beta_i\}$ we have $s_i = \langle x_i\rangle  \circ \ol{s} = f_{ij} \circ \ol{t}$.
      
      \item Let $z = \sum_{i=1}^n \alpha_i +  \sum_{j=1}^m \beta_j$. We collect all the commuting squares closing the span $(\ol{t},\ol{s})$ identified in the two previous points, which are exactly $z$, into a single one by exploiting the bijection between arrows to $z$ and $z$-tuples of arrows to $1$. Thus we define $\ol{f} \colon n \to z$ and $\ol{g} \colon m \to z$ respectively as
      
      \[\ol{f} = \langle \underbrace{x_1,\ldots,x_1}_{\alpha_1}, \ldots, \underbrace{x_n,\ldots,x_n}_{\alpha_n}, f_{11},\ldots, f_{1\beta_1},\ldots,f_{m1},\ldots, f_{m\beta_m} \rangle
      \]
  
      \[
      \ol{g} = \langle g_{11},\ldots, g_{1\alpha_1},\ldots,g_{n1},\ldots, g_{n\alpha_n},\underbrace{x_1,\ldots,x_1}_{\beta_1}, \ldots, \underbrace{x_m,\ldots,x_m}_{\beta_m}  \rangle
      \]
      The fact that $\ol{f} \circ \ol{t} = \ol{g} \circ \ol{s}$ immediately follows because it holds for all tuple components.
  \end{enumerate}

  Give the construction just described, $(\ol{f}, z, \ol{g})$ is a weak
  pushout of $(\ol{t},\ol{s})$. For, let $c \colon n \to 1$ and $d \colon m \to 1$
  be two arbitrary arrows such that $c \circ \ol{t} = d\circ \ol{s}$. We
  have to show that there is an arrow $e \colon z \to 1$ such that $e \circ
  \ol{f} = c$ and   $e \circ \ol{g} = d$.
  %\todo{B: can we draw a diagram depicting the situation?\\ A: Yes, I can draw it if we keep this part.} Note that there is no need to consider arrows with an arbitrary object $z' \neq 1$ as target, because such arrows are one-to-one with tuples of arrows to $1$. 
  
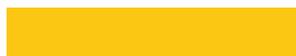
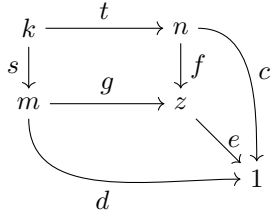
    \[
      \begin{tikzpicture}%[x=0.90cm]
        %[baseline={(0,0.2)}]
            %  \path[use as bounding box] (0,0) rectangle (2.5,-1.00);
        \node (sqtl) at (0,0) {$k$};
        \node (sqtr) at (2,0) {$n$};
        \node (sqbl) at (0,-1) {$m$};
        \node (sqbr) at (2,-1) {$z$};
        \node (bbrr) at (3,-2) {$1$};
    
        \draw[->] (sqtl) -- node[above]{$\ol{t}$} (sqtr);
        \draw[->] (sqtl) -- node[left]{$\ol{s}$} (sqbl);
        \draw[->] (sqbl) -- node[above]{$\ol{g}$} (sqbr);
        \draw[->] (sqtr) -- node[right]{$\ol{f}$} (sqbr);
        \draw[->] (sqtr) to [out = 0, in = 90] node[right]{$c$} (bbrr);
        \draw[->] (sqbl) to [out = 270, in = 180] node[below]{$d$} (bbrr);
        \draw[->] (sqbr) -- node[right]{$e$} (bbrr);
      \end{tikzpicture}
    \]
  
  Arrow $c \colon n \to 1$ is a $\Sigma$-term with variables in $\{x_1,
  \ldots, x_n\}$, and $d \colon m \to 1$ is a $\Sigma$-term with (renamed apart) variables in
  $\{y_1, \ldots, y_m\}$.
  %\todo{B: are the $y_i$ used somewhere?\\ A: not explicity, but the point is to rename apart variables appearing in $c$ and $d$.} 
  %
  Let $e = c \sqcup d:z \to 1$ be the  $\Sigma$-term on $\langle v_1, \ldots, v_z \rangle$  defined inductively as follows:
  
  \[ c \sqcup d = \left\{
      \begin{array}{ll}
          h(c_1 \sqcup d_1, \ldots,c_r \sqcup d_r) & \mbox{ if } c = h(c_1,\ldots,c_r) \mbox{ and } d = h(d_1,\ldots,d_r) \\
          v_j & \mbox{ if } \pi_j \circ \ol{f} = c \mbox{ and } \pi_j \circ \ol{g} = d
      \end{array}
  \right.\]
  
  % \todo{B: Can we use index $i$ or $j$ instead of $x$? $v_j$ vs. $v_x$.\\ A: done}
  
  Term $e = c \sqcup d:z \to 1$ is called an \emph{anti-unifier} or a \emph{least general generalization} of $c$ and $d$~\cite{p:inductive-generalization}.
  Let us show that it is well-defined. Since by hypothesis $c \circ \ol{t} = d\circ \ol{s}$ we know that $c$ and $d$ unify, therefore for each pair of corresponding subterms\footnote{That is, rooted at the same position in the two terms, a notion that we do not formalize further here.} $(c_s,d_s)$ of $c$ and $d$, if $c_s = h_c(c_1, \ldots, c_{r'})$ and $d_s = h_d(d_1, \ldots, d_{r''})$, then $h_c = h_d$ and $r' = r''$. Therefore in the first clause we only consider the case where the two subterms have the same operator symbol as root. The only other possibility is that either $c$ or $d$ or both are variables. If $c = x_i$ with $i \in \{1,\ldots, n\}$, then $c \circ \ol{t} = x_i  \circ \ol{t} = t_i$, thus we have $t_i = d \circ \ol{s}$, meaning that $d$ is one of the terms obtained in the first step of the above construction. By the definition of $\ol{f}$ and $\ol{g}$ there is an index $j$ such that $\pi_j \circ \ol{f} = x_i = c$ and $\pi_j \circ \ol{g} = d$, as desired. This holds even if also $d$ is a variable.  The case where $d$ is a variable and $c$ is not can be justified analogously.
  
  It is then easy to verify that by construction $ e \circ \ol{f} = c$ and $e \circ \ol{g} = d$, as desired.
  
  %\todo{B: relation to anti-unifiers?\\ A: $m(c,d)$ is indeed the anti-unifier: the construction is the same as Plotkin's optimized to the present situation for the second clause. We could add this in a note, but can we add a bibitem?}
  
  \end{proof}
  
  In the rest of this section we assume that a signature $\Sigma$ is given containing at least a constant $a$ and two unary function symbols $g$ and $f$. 
  
  \begin{example}[Construction of weak pushout]
    \label{ex:constr-weak-po}
    Let $a: 0 \to 1$ and $g(a): 0 \to 1$ be two arrows of
    $\mathit{Th_\Sigma}$. For Step~1 of the proof of
    \Cref{th:wp-LawTh}, using item 2(b) of
    \Cref{lem:fin_fac_in_Law_theory} we find that
    $a = a \circ \,! \circ g(a)$ is the only way to factorize $a$
    through $g(a)$. For Step~2, we find two ways of factorizing $g(a)$
    through $a$, namely $g(x_1) \circ a$ and $g(a) \circ \,!  \circ a$.
    Thus according to Step~3 we obtain the weak pushout of $a$ and $g(a)$ by collecting the resulting three commuting squares into a single one, as depicted in the fourth diagram.\\
  
    \small
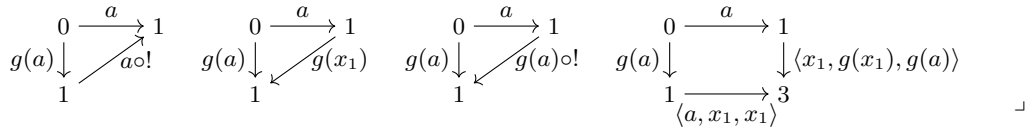
{
    \begin{tikzpicture}
      \node (sqtl) at (0,0) {$0$};
      \node (sqtr) at (1.25,0) {$1$};
      \node (sqbl) at (0,-0.9) {$1$};
     % \node (sqbr) at (1.25,-0.9) {m};
      \draw[->] (sqtl) -- node[overlay,above]{$\ol{a}$} (sqtr);
      \draw[->] (sqtl) -- node[left]{$\ol{g(a)}$} (sqbl);
      \draw[->] (sqbl) -- node[below,right]{$a\circ ! $} (sqtr);
     % \draw[->] (sqbl) -- node[overlay,below]{$\beta$} (sqbr);
    \end{tikzpicture}
    \hfill
    \begin{tikzpicture}
      \node (sqtl) at (0,0) {$0$};
      \node (sqtr) at (1.25,0) {$1$};
      \node (sqbl) at (0,-0.9) {$1$};
     % \node (sqbr) at (1.25,-0.9) {m};
      \draw[->] (sqtl) -- node[overlay,above]{$\ol{a}$} (sqtr);
      \draw[->] (sqtl) -- node[left]{$\ol{g(a)}$} (sqbl);
      \draw[->] (sqtr) -- node[below,right]{$g(x_1)$} (sqbl);
     % \draw[->] (sqbl) -- node[overlay,below]{$\beta$} (sqbr);
    \end{tikzpicture}
    \hfill
    \begin{tikzpicture}
      \node (sqtl) at (0,0) {$0$};
      \node (sqtr) at (1.25,0) {$1$};
      \node (sqbl) at (0,-0.9) {$1$};
     % \node (sqbr) at (1.25,-0.9) {m};
      \draw[->] (sqtl) -- node[overlay,above]{$\ol{a}$} (sqtr);
      \draw[->] (sqtl) -- node[left]{$\ol{g(a)}$} (sqbl);
      \draw[->] (sqtr) -- node[below,right]{$g(a)\circ !$} (sqbl);
     % \draw[->] (sqbl) -- node[overlay,below]{$\beta$} (sqbr);
    \end{tikzpicture}
    \hfill
    \begin{tikzpicture}
      \node (sqtl) at (0,0) {$0$};
      \node (sqtr) at (1.5,0) {$1$};
      \node (sqbl) at (0,-0.9) {$1$};
     \node (sqbr) at (1.5,-0.9) {$3$};
      \draw[->] (sqtl) -- node[overlay,above]{$\ol{a}$} (sqtr);
      \draw[->] (sqtl) -- node[left]{$\ol{g(a)}$} (sqbl);
      \draw[->] (sqtr) -- node[below,right]{$\langle x_1,g(x_1),g(a)\rangle$} (sqbr);
      \draw[->] (sqbl) -- node[overlay,below]{$\langle a,x_1,x_1\rangle$} (sqbr);
    \end{tikzpicture}
    } % small
  \end{example}

\begin{example}[Unsatisfiability of a condition over terms]
  Consider the following condition:
  \begin{align*}
    & \forall\ \id_0 \dotEx 0 \xrightarrow{g(a)} 1 \dotAll 1 \xrightarrow{!} 0 \dotFalse
    & \text{\small($g(a)$ occurs somewhere)} \\
    \land & \forall\ 0 \xrightarrow{a} 1 \dotEx 1 \xrightarrow{f(x_1)} 1 \dotTrue
    & \text{\small(and every $a$ is directly wrapped in $f(\ )$)}
  \end{align*}

  This condition is unsatisfiable since the first subcondition
  requires a subterm $g(a)$ (that is not discarded!), which is in
  conflict with the second subcondition that states that each constant
  $a$ in contained in a function symbol $f$.  Our procedure can show
  its unsatisfiability as follows.  Note that we use the isomorphism
  rule from the special case where possible to reduce complexity of
  intermediate results. A proof using only the section rule would also
  be possible.

  \begin{itemize}
  \item Pull forward iso in first line: we obtain a partial model
    $\id_0$ and the child
      \begin{align*}
        & \exists\ 0 \xrightarrow{g(a)} 1 .\Big(
        \forall\ 1 \xrightarrow{!} 0 \dotFalse \land
        (\forall\ 0 \xrightarrow{a} 1 \dotEx 1 \xrightarrow{f(x_1)} 1 \dotTrue)_{\downarrow 0 \xrightarrow{g(a)} 1}
      \Big) \\
      =& \exists\ 0 \xrightarrow{g(a)} 1 .\Big(
        \forall\ 1 \xrightarrow{!} 0 \dotFalse \land
        \forall\ 1 \xrightarrow{a,x_1,x_1} 3 \dotEx 3 \xrightarrow{f(x_1),x_2} 2 \dotTrue
      \Big)
    \end{align*}
    Note that this requires the representative square computed in
    \Cref{ex:constr-weak-po}.
    \item
      The child condition has no isos, but a section ($f_p = 1 \xrightarrow{a,x_1,x_1} 3$) with two right-inverses ($3 \xrightarrow{x_2} 1$ and $3 \xrightarrow{x_3} 1$).
      We pull it forward using $r_p = 3 \xrightarrow{x_3} 1$.

      First compute ${\mathcal A_p}_{\downarrow r_p} = (\exists\ 3 \xrightarrow{f(x_1),x_2} 2 \dotTrue)_{\downarrow 3 \xrightarrow{x_3} 1} = \exists\ 1 \xrightarrow{!} 0 \dotTrue$.
      Then we obtain
      \begin{align*}
        & \exists\ 1 \xrightarrow{!} 0 . \Big(
          \condtrue \land
          (\forall\ 1 \xrightarrow{!} 0 \dotFalse)_{\downarrow 1 \xrightarrow{!} 0}
          \land
          (\forall\ 1 \xrightarrow{a,x_1,x_1} 3 \dotEx 3 \xrightarrow{f(x_1),x_2} 2 \dotTrue)_{\downarrow 1 \xrightarrow{!} 0}
        \Big) \\
        =& \exists\ 1 \xrightarrow{!} 0 . \Big(
          \forall\ \id_0 \dotFalse \land \dots
        \Big)
      \end{align*}
    \item
      The child condition contains an iso $\id_0$.
      The associated child condition is an empty existential ($\condfalse$).
      Hence pulling it forward again results in an empty existential, which closes the tableau and proves the condition unsatisfiable.
  \end{itemize}

  Note that without the subcondition $\forall\ 1 \xrightarrow{!} 0 \dotFalse$, the condition would be satisfiable, as it would be possible to create and forget parts.
  For example, $0 \xrightarrow{f(a)} 1$ is a model for both $\forall\ 0 \xrightarrow{a} 1 \dotEx 1 \xrightarrow{f(x_1)} 1 \dotTrue$ and for $\exists\ 0 \xrightarrow{g(a)} 1 \dotTrue$, the latter resulting from the decomposition $0 \xrightarrow{g(a)} 1 \xrightarrow{f(a)} 1$ where the ground term $g(a)$ is simply forgotten.
\end{example}

\end{toappendix}

\end{document}